\documentclass[11pt]{article}

\usepackage[utf8]{inputenc} % allow utf-8 input
\usepackage[T1]{fontenc} 
\usepackage{amsthm}
\usepackage{mathtools}

\usepackage{graphicx} % support the \includegraphics command and options
\usepackage{array} % for better arrays (eg matrices) in maths

\usepackage{amsmath, amssymb, amsfonts, verbatim}
\usepackage{hyphenat,epsfig,subcaption,multirow}
\usepackage{nicefrac}
\usepackage{paralist}
\usepackage{enumitem}

\usepackage[font=footnotesize,labelfont=bf]{caption}

\usepackage[usenames,dvipsnames]{xcolor}
\usepackage[ruled]{algorithm2e}

\DeclareFontFamily{U}{mathx}{\hyphenchar\font45}
\DeclareFontShape{U}{mathx}{m}{n}{
      <5> <6> <7> <8> <9> <10>
      <10.95> <12> <14.4> <17.28> <20.74> <24.88>
      mathx10
      }{}
\DeclareSymbolFont{mathx}{U}{mathx}{m}{n}
\DeclareMathSymbol{\bigtimes}{1}{mathx}{"91}

\usepackage{tcolorbox}
\tcbuselibrary{skins,breakable}
\tcbset{enhanced jigsaw}

\usepackage[normalem]{ulem}
\usepackage[compact]{titlesec}

\definecolor{DarkRed}{rgb}{0.5,0.1,0.1}
\definecolor{DarkBlue}{rgb}{0.1,0.1,0.5}
\definecolor{RURed}{rgb}{0.8,0.1,0.1}

\newcommand{\chen}[1]{\textcolor{RURed}{[\textbf{Chen:} #1]}}

\usepackage{nameref}
\definecolor{ForestGreen}{rgb}{0.1333,0.5451,0.1333}
\definecolor{Red}{rgb}{0.9,0,0}
\usepackage[linktocpage=true,
	pagebackref=true,colorlinks,
	linkcolor=DarkRed,citecolor=ForestGreen,
	bookmarks,bookmarksopen,bookmarksnumbered]
	{hyperref}
\usepackage[noabbrev,nameinlink]{cleveref}
\crefname{property}{property}{Property}
\creflabelformat{property}{(#1)#2#3}
\crefname{equation}{eq}{Eq}
\creflabelformat{equation}{(#1)#2#3}

\usepackage{bm}
\usepackage{url}
\usepackage{xspace}
\usepackage[mathscr]{euscript}

\usepackage{tikz}
\usetikzlibrary{arrows}
\usetikzlibrary{arrows.meta}
\usetikzlibrary{shapes}
\usetikzlibrary{backgrounds}
\usetikzlibrary{positioning}
\usetikzlibrary{decorations.markings}
\usetikzlibrary{patterns}
\usetikzlibrary{calc}
\usetikzlibrary{fit}
\usetikzlibrary{snakes}
\tikzset{vertex/.style={circle, black, fill=Yellow, line width=1pt, draw, minimum width=8pt, minimum height=8pt, inner sep=0pt}}
	
\usepackage{mdframed}

\usepackage[noend]{algpseudocode}
\makeatletter
\def\BState{\State\hskip-\ALG@thistlm}
\makeatother

\usepackage{cite}
\usepackage{enumitem}
\usepackage{todonotes}

\usepackage[margin=1in]{geometry}

\newtheorem{theorem}{Theorem}
\newtheorem{lemma}{Lemma}[section]
\newtheorem{proposition}[lemma]{Proposition}
\newtheorem{corollary}[lemma]{Corollary}
\newtheorem{claim}[lemma]{Claim}
\newtheorem{fact}[lemma]{Fact}

\newtheorem{definition}[lemma]{Definition}

\newtheorem{problem}{Problem}

\newtheorem*{claim*}{Claim}
\newtheorem*{theorem*}{Theorem}
\newtheorem*{proposition*}{Proposition}
\newtheorem*{lemma*}{Lemma}
\newtheorem*{problem*}{Problem}

\crefname{lemma}{Lemma}{Lemmas}
\crefname{claim}{Claim}{Claims}

\newtheorem{mdresult}{Result}
\newenvironment{result}{\begin{mdframed}[backgroundcolor=lightgray!40,topline=false,rightline=false,leftline=false,bottomline=false,innertopmargin=2pt, innerleftmargin=10pt]\begin{mdresult}}{\end{mdresult}\end{mdframed}}

\newtheorem{remark}[lemma]{Remark}
\newtheorem*{remark*}{Remark}
\newtheorem{observation}[lemma]{Observation}

\newtheoremstyle{restate}{}{}{\itshape}{}{\bfseries}{~(restated).}{.5em}{\thmnote{#3}}
\theoremstyle{restate}

\theoremstyle{definition}
\newtheorem{mdalg}{algorithm}

\newtheorem{mddist}{Distribution}

\allowdisplaybreaks

\renewcommand{\qed}{\nobreak \ifvmode \relax \else
      \ifdim\lastskip<1.5em \hskip-\lastskip
      \hskip1.5em plus0em minus0.5em \fi \nobreak
      \vrule height0.75em width0.5em depth0.25em\fi}

\newcommand{\tvd}[2]{\ensuremath{\norm{#1 - #2}_{\mathrm{tvd}}}}
\newcommand{\Ot}{\ensuremath{\widetilde{O}}}
\newcommand{\Omt}{\ensuremath{\widetilde{\Omega}}}
\newcommand{\eps}{\ensuremath{\varepsilon}}

\newcommand{\bracket}[1]{\left[#1\right]}
\newcommand{\paren}[1]{\ensuremath{\left(#1\right)}\xspace}
\newcommand{\card}[1]{\left\vert{#1}\right\vert}
\newcommand{\Omgt}{\ensuremath{\widetilde{\Omega}}}

\newcommand{\IR}{\ensuremath{\mathbb{R}}}

\newcommand{\norm}[1]{\ensuremath{\|#1\|}}

\newcommand{\expect}[1]{\Exp\bracket{#1}}

\newcommand{\set}[1]{\ensuremath{\left\{ #1 \right\}}}
\newcommand{\poly}{\mbox{\rm poly}}
\newcommand{\polylog}{\textnormal{polylog}\xspace}

\newcommand{\OPT}{\ensuremath{\mbox{\sc opt}}\xspace}

\newcommand{\ALG}{\ensuremath{\mbox{\sc alg}}\xspace}
\newcommand{\alg}{\ensuremath{\mathcal{A}}\xspace}

\DeclareMathOperator*{\Exp}{\ensuremath{{\mathbb{E}}}}
\DeclareMathOperator*{\Prob}{\ensuremath{\textnormal{Pr}}}
\renewcommand{\Pr}{\Prob}

\newenvironment{tbox}{\begin{tcolorbox}[
		enlarge top by=5pt,
		enlarge bottom by=5pt,
		 breakable,
		 boxsep=2pt,
                  left=5pt,
                  right=7pt,
                  top=10pt,
                  arc=0pt,
                  boxrule=1pt,toprule=1pt,
                  colback=white
                  ]%%
	}
{\end{tcolorbox}}

\newcommand{\event}{\ensuremath{\mathcal{E}}}

\newcommand{\supp}[1]{\ensuremath{\textnormal{\text{supp}}(#1)}}
\newcommand{\distribution}[1]{\ensuremath{\textnormal{dist}(#1)}\xspace}

\newcommand{\kl}[2]{\ensuremath{\mathbb{D}(#1~||~#2)}}
\newcommand{\II}{\ensuremath{\mathbb{I}}}

\newcommand{\mireal}[1][]{
  \ifx\relax#1\relax%
    \II(\mione \,; \mitwo)%
  \else%
    \II(\mione \,; \mitwo\mid #1)%
  \fi
}

\newcommand{\prot}{\pi}
\newcommand{\Prot}{\Pi}

\newcommand{\cost}{\ensuremath{\mathsf{C}}}

\newcommand{\bias}{\ensuremath{\mathrm{bias}}}

\newcommand{\TT}{\ensuremath{\mathcal{T}}}

\renewcommand{\cost}{\ensuremath{\textnormal{\textbf{cost}}}\xspace}

\newcommand{\leaves}{\ensuremath{\textnormal{\textbf{leaf-nodes}}}\xspace}

\newcommand{\cut}{\ensuremath{\textnormal{\textbf{cut}}}\xspace}

\renewcommand{\OPT}{\textnormal{\textsf{OPT}}\xspace}
\newcommand{\hf}{\ensuremath{\hat{f}}}

\title{Hierarchical Clustering in Graph Streams:\\ Single-Pass Algorithms and Space Lower Bounds}
 \author{Sepehr Assadi\footnote{(sepehr.assadi@rutgers.edu) Department of Computer Science, Rutgers University. Research supported in part by a NSF CAREER Grant CCF-2047061, a gift from Google Research, 
 and a Fulcrum award from Rutgers Research Council. 
 } \and 
 Vaggos Chatziafratis\footnote{(vaggos@stanford.edu). Department of Computer Science and Engineering, University of California -- Santa Cruz (UCSC). Part of the work done while this author was a FODSI fellow at  MIT and Northeastern.} \and
   Jakub Łącki\footnote{(jlacki@google.com) Google Research.} \and
 Vahab Mirrokni\footnote{(mirrokni@google.com) Google Research.} \and
 Chen Wang\footnote{(chen.wang.cs@rutgers.edu) Department of Computer Science, Rutgers University. Research supported in part by a NSF CAREER Grant CCF-2047061, and a gift from Google Research.}
}

\date{}

\begin{document}

\maketitle

\pagenumbering{roman}

% !TeX root = main.tex 
%!TEX root = main.tex
\begin{abstract}
	The Hierarchical Clustering (HC) problem consists of building a hierarchy of clusters to represent a given dataset. Motivated by the modern large-scale applications, we study the problem in the \emph{streaming model}, in which the memory is heavily limited and only a single or very few passes over the input are allowed. 
    Specifically, we investigate whether a good  hierarchical clustering can be obtained, or at least whether we can approximately estimate the value of the optimal hierarchy. To measure the quality of a hierarchy, we use the HC minimization objective introduced by Dasgupta~\cite{Dasgupta16}. Assuming that the input is an $n$-vertex weighted graph whose edges arrive in a stream, we derive the following results on space-vs-accuracy tradeoffs:
	\begin{itemize}
	    \item  With $O(n\cdot \polylog\,{n})$ space, we develop a single-pass algorithm, whose approximation ratio matches the currently best \textit{offline} algorithm of~\cite{charikar2017approximate}.
	    \item When the space is more limited, namely, $n^{1-o(1)}$, we prove that no algorithm can even estimate the value of optimum hierarchical tree to within an $o(\frac{\log{n}}{\log\log{n}})$ factor, even when allowed $\polylog{\,{n}}$ passes over the input and exponential time. 
	    \item In the most stringent setting of $\polylog\,{n}$ space, studied extensively in the literature, we rule out algorithms that can even distinguish between ``highly''-vs-``poorly'' clusterable graphs, namely, graphs that have an $n^{1/2-o(1)}$ factor gap between their HC objective value. 
	    \item Finally, we prove that any single-pass streaming algorithm that computes an optimal HC clustering requires to store almost the entire input even if allowed exponential time. 
	\end{itemize}
	Our algorithmic results establish a general structural result that proves that cut sparsifiers of input graph can preserve cost of ``balanced'' hierarchical trees to within a constant factor, and thus can be used in place of the original (dense) graphs when solving HC. Our lower bound results include a new streaming lower bound for a novel problem ``One-vs-Many-Expanders'', which can be of independent interest.

\end{abstract}

\clearpage

\setcounter{tocdepth}{3}
\tableofcontents

\clearpage

\pagenumbering{arabic}
\setcounter{page}{1}

% !TeX root = main.tex 
%!TEX root = main.tex
\section{Introduction}
\label{sec:intro}
Motivated by a variety of data mining and computational biology applications, Hierarchical Clustering (HC) 
is the canonical problem of 
building a hierarchy of clusters to represent a dataset. This hierarchy takes the form of a rooted binary tree (also called a ``dendrogram'') whose leaves are in one-to-one correspondence with the data points, thus capturing their relationships at various levels of granularity. Representing a dataset as a tree structure offers several advantages: there is no need to specify the number of clusters in advance, HC is easy to interpret and visualize, and there are simple-to-implement HC algorithms available (e.g., either top down divisive or bottom up linkage methods). As a result, HC has played a prominent role both in theory and in practice across different domains, with canonical applications ranging from biology and statistics to finance and sociology~\cite{cavalli1967phylogenetic,berkhin2006survey,eisen1998cluster,felsenstein2004inferring,hastie2009,tumminello2010correlation,NIPS2017-affinity,mann2008use}.

Deploying HC algorithms in practice however is a challenging task. In particular, a major challenge is achieving good scalability. With the rise of data-intensive applications, there is dire need to solve HC for extremely large datasets. Additionally, these datasets are typically evolving over time (e.g., new queries/users/videos added in a platform), thus making said scaling issues even harder to deal with. The best known algorithms for some commonly used linkage methods, such as Average Linkage, suffer from quadratic runtime (in the number of data points) which is prohibitive in modern settings. 
To overcome these issues, recent efforts have focused on accelerating 
bottom-up linkage methods~\cite{loewenstein2008efficient,abboud2019subquadratic,bateni2017affinity,monath2019scalable,monath2021scalable,sumengen2021scaling,dhulipala2021hierarchical} or top-down divisive methods~\cite{avdiukhin2019multi} and on exploiting geometric embedding techniques~\cite{naumov2021objective,rajagopalan2021hierarchical,nickel2017poincare}.

In this paper, we study HC in the \textit{graph streaming model}, 
which is a canonical model designed to capture the essence of large-scale computation. Graph streaming algorithms 
process their input by making one (or few) sequential pass(es) 
over their edges while using a limited memory, much smaller than the input size. These constraints capture several challenges of processing massive graphs such as I/O-efficiency or monitoring evolving graphs; see, e.g.~\cite{Muthukrishnan05,FeigenbaumKMSZ05,McGregor14} and references therein. The main motivation behind our work is  the following question: 
\begin{center}
\textit{If we are allowed only a single sequential pass over the data and a limited space, how good a hierarchical clustering can we compute? In general, what are the space-\emph{vs}-accuracy tradeoffs?}
\end{center}

We present several algorithmic and impossibility results that address this question. On the algorithmic front, we design a single-pass HC algorithm minimizing Dasgupta's HC cost function~\cite{Dasgupta16}, that matches the guarantees of known non-streaming algorithms~\cite{charikar2017approximate,cohen2019hierarchical}, while using memory proportional to the number of data points (and thus quadratically smaller than the input size that contains pairwise similarities of the data points). On the lower bounds front, we give several impossibility results across a range of various (sublinear) memory regimes, providing tradeoffs for the space required in order to obtain ``good'' HC trees or to estimate their values, as measured by Dasgupta's objective~\cite{Dasgupta16}. We elaborate more on our results in~\Cref{sec:results-sum}. 

To the best of our knowledge, we are the first to provide theoretical guarantees for streaming HC under Dasgupta's cost function in the general graph similarity setting (i.e., the input need not satisfy triangle inequality), and/or under memory limitations or single-pass/few-pass desiderata. In contrast, recent results in~\cite{rajagopalan2021hierarchical} hold only for metric data in $\mathbb{R}^d$ and their focus is on maximization HC objectives~\cite{moseley2017approximation,cohen2019hierarchical} (which are provably shown to be easier to approximate~\cite{charikar2019hierarchical,alon2020hierarchical,naumov2021objective}).

\subsection{Background, Problem Definition, and Related Work}
Before stating our results in more detail, we start with a brief description of prior work in the literature of \emph{optimization-based} hierarchical clustering. The main motivating question here is ``how does one evaluate the quality of a hierarchical tree on a given dataset?''.

Despite its popularity and importance, HC is underdeveloped 
from a theoretical perspective. In particular, many heuristics for HC are defined procedurally rather than in terms of an optimization objective; as such they lack theoretical analyses on their performance guarantees.
Indeed, until recently, there was no global objective function for HC to evaluate how good or bad a proposed solution is, in stark contrast with the multitude of objectives we typically encounter in ``flat'' clustering (e.g., $k$-means, $k$-medians, $k$-multicut, correlation clustering, etc.). Having an appropriate objective allows us to evaluate the performance of different algorithms, to quantify their success or failures, and in some cases, to add explicit constraints for the hierarchy~\cite{kernelNIPS,dasguptaICML,chatziafratis2018hierarchical}, similar to ``must-link/cannot-link'' constraints in $k$-means~\cite{wagstaff2000clustering,wagstaff2001constrained}.

In an influential work, Dasgupta~\cite{Dasgupta16} proposed a minimization objective for HC based on pairwise similarity information on $n$ data points. Under this objective, the data is embedded as a graph $G=(V,E,w)$, where the vertices are the data points, and edges are obtained by pairwise similarity. The clustering is represented by a rooted tree $\TT$, where each leaf node contains a single vertex, and each non-leaf node of $\TT$ induces a cluster (as such, the root contains $V$). The inclusion of sub-clusters is characterized by the clusters induced by child nodes. The total cost is measured by the summation of the (weighted) pairwise costs, where the cost between vertex pair $(u,v)$ is defined as the (weighted) number of leaf nodes induced by the subtree rooted at the lowest common ancestor between $u$ and $v$ (see \Cref{prob:HC} for the formal definition).

Dasgupta~\cite{Dasgupta16} gave a poly-time $O(\alpha(n) \cdot \log{n})$-approximation algorithm for the aforementioned hierarchical clustering problem, where $\alpha(n)$ denotes the best approximation ratio possible for the \emph{Sparsest Cut} problem (currently, $\alpha(n) = O(\sqrt{\log{n}})$~\cite{arora2009expander}). Follow-up works improved on this result by proving an~$O(\log n)$ approximation via linear programming~\cite{roy2016hierarchical} and an $O(\sqrt{\log n})$ approximation via semidefinite programming~\cite{charikar2017approximate}. In addition,~\cite{charikar2017approximate,cohen2019hierarchical} improved the analysis of~\cite{Dasgupta16} based on sparsest cut problem to achieve an $O(\alpha(n))$-approximation (\!\cite{charikar2017approximate} also provides a similar algorithm using \emph{Balanced Cut} as a subroutine instead of sparsest cut).  
On the hardness front,~\cite{charikar2017approximate} proved that under the \emph{Small Set Expansion (SSE) Hypothesis}, there is no constant factor approximation algorithm for Dasgupta's hierarchical clustering problem in polynomial time. 

More generally, Dasgupta's objective has led to a flurry of both theoretical and empirical results about the computational complexity and optimization of HC, expanding our understanding of HC and mirroring the important progress made in the ``flat'' clustering literature over the past several decades. Such results include approximation guarantees for old linkage algorithms~\cite{moseley2017approximation,cohen2019hierarchical}, explaining success of existing methods~\cite{charikar2019hierarchical,avdiukhin2019multi}, designing novel approaches to HC~\cite{roy2016hierarchical,chatziafratis2018hierarchical}, characterizing its computational complexity and inapproximability~\cite{charikar2017approximate,chatziafratis2021maximizing}, and novel connections to hyperbolic embeddings~\cite{chami2020trees, monath2019gradient}.

In this work, we study Dasgupta's hierarchical clustering problem in the graph streaming model, wherein the edges of the input graph $G$ are arriving one by one in an arbitrary order, and the algorithm is allowed to make a single pass over these edges and compute an HC tree $\TT$ that (approximately) 
minimizes $\cost_G(\TT)$ (or estimate $\cost_G(\TT)$ studied in some of our lower bounds). See~\Cref{sec:prelim} for more formal definitions. 

\subsection{Our Contributions}
\label{sec:results-sum}
We provide a comprehensive treatment of HC in the graph streaming model. Our first result gives an algorithm for obtaining an approximation ratio proportional to the best non-streaming algorithm, while using only $\Ot(n) \coloneqq O(n \cdot \polylog{(n)})$\footnote{Throughout, we use $\Ot(\cdot)$ and $\Omt(\cdot)$ notation to suppress $\poly\log{(n)}$ factors.} space, referred to as \emph{Semi-Streaming} space restriction~\cite{FeigenbaumKMSZ05}, the so-called `sweet spot' for graph streaming algorithms. 

\begin{result}
\label{rst:upper-bound}
There exists a single-pass streaming algorithm for hierarchical clustering (\Cref{prob:HC}) that uses $O(n\cdot \polylog\,{n})$ space and achieves an $O(\sqrt{\log{n}})$-approximation in polynomial time or $O(1)$-approximation 
in exponential time. 
\end{result}
\Cref{rst:upper-bound} gives us the best of both worlds: approximation ratio asymptotically matching best non-streaming algorithm of~\cite{charikar2017approximate} (by using it in a black-box way) and space complexity that is only larger than the output clustering by $\poly\log{n}$ factors. Moreover, as we describe later, this result can use many other HC algorithms or heuristics as a black-box in place of~\cite{charikar2017approximate} (e.g., to gain faster runtime), while achieving asymptotically the same approximation ratio as  the black-box algorithm.  

As we shall explain more in~\Cref{subsec:technique}, our~\Cref{rst:upper-bound} is based on a general \emph{sparsification} approach that can be used in a variety of other settings as well. For instance, it also 
implies a $2$-round Massively Parallel Computation (MPC) algorithm for HC on machines of memory $\Ot(n)$; see, e.g.~\cite{KarloffSV10,AhnGM12b,BeameKS13,KumarMVV13,CzumajLMMOS18,AssadiCK19a,AssadiBBMS19} and references therein for more details on the MPC model and its connection to streaming, among others. 

The space complexity of our algorithm in~\Cref{rst:upper-bound} is nearly optimal as any streaming algorithm that outputs a clustering of input points requires $\Omega(n\log{n})$ bits of space just to store the answer. However, in many scenarios, one is interested in algorithms that can \emph{distinguish} between ``highly clusterable'' inputs versus ``poorly clusterable'' ones; in other words, be able to only \emph{estimate} the cost of the best HC tree in Dasgupta's hierarchical clustering problem (see, e.g.~\cite{KapralovKS15,KapralovK19,GuruswamiT19,AssadiKSY20,ChouGV20,AssadiN21} for a vibrant area of research on these streaming estimation problems for property testing or constraint satisfaction problems). While our algorithm in~\Cref{rst:upper-bound} clearly also works for the estimation problem, its space can no longer be considered nearly-optimal a priori. Our next result addresses this. 

\begin{result}
\label{rst:lb-near-linear}
Any streaming algorithm with a memory of $n^{1-o(1)}$ cannot estimate the hierarchical clustering objective \underline{value} to within an $o(\frac{\log{n}}{\log\log{n}})$ factor even if allowed $\polylog{(n)}$ passes and exponential time. 
\end{result}
This result effectively rules out any algorithm with $n^{1-o(1)}$ memory to achieve an approximation ratio as competitive as~\Cref{rst:upper-bound} for the estimation problem (even when allowed exponential time and an ``unreasonably large'' number of passes). It is  worth noting that while we obtain this result by a reduction from  streaming lower bounds of~\cite{AssadiN21}, this is the first application of these techniques to proving lower bounds for $\omega(1)$ approximation factors as well as $\omega(\log{n})$ passes. 

While quite strong in terms of space (and passes), \Cref{rst:lb-near-linear} still leaves out possibility of algorithms with approximation ratio  of $O(\log{n})$, which are quite acceptable for HC. On the other end of the spectrum, one can ask how well of an approximation can we hope for on the most stringent restriction of $\polylog{(n)}$ space and one pass? (this is the setting most focused on in ``classical'' streaming literature starting from~\cite{AlonMS96}, as well as aforementioned line of work on estimation problems in graph streams). Our next result suggests that the answer is \emph{``not much''}.  
\begin{result}
\label{rst:lb-polylog-space}
Any single-pass streaming algorithm with $\polylog{(n)}$ space cannot estimate the hierarchical clustering \underline{value} with an approximation ratio of $n^{1/2-\delta}$ for any constant $\delta>0$. 
\end{result}

Proof of~\Cref{rst:lb-polylog-space} turned out to be the most technically challenging part of our paper, as we can no longer rely on reductions from existing streaming lower bounds. En route to proving this result, we establish a general streaming lower bound: no $\polylog{(n)}$ space streaming algorithm  can distinguish between inputs consisting of a single \emph{expander} versus a collection of many \emph{small} vertex-disjoint expanders. This problem is the ``expander-variant'' of the by-now famous \emph{gap cycle counting} problem of~\cite{VerbinY11} (where instead of expanders, we have \emph{cycles} in the input) that has found numerous applications in streaming lower bounds (including our~\Cref{rst:lb-near-linear}); see, e.g.~\cite{VerbinY11,KapralovKS15,AssadiKSY20,AssadiN21,KapralovMTWZ22} and references therein. Our expander-variant of this problem seems versatile enough to find other applications and is therefore interesting in its own right. 

Finally, going back to~\Cref{rst:upper-bound} and the $\Ot(n)$-space regime, we can ask whether settling for approximation was even necessary for this problem. In particular, can we match the performance of best non-streaming algorithms \emph{exactly} (not asymptotically) or better yet obtain an exact optimal solution in exponential time? Our final result rules out this possibility also as long as the space of our algorithm is less than the input size (at which point, we can trivially store the entire input and solve the problem offline in exponential time).  

\begin{result}
\label{rst:lb-exact-solution}
Any single-pass streaming algorithm for finding an optimal hierarchical clustering tree (or even determining its cost) requires a memory of $\Omega(n^2)$ bits. 
\end{result}

This concludes the description of our main results. Putting these results together, our paper has the following message. It is possible to solve HC with asymptotically the same approximation ratio as that of best known non-streaming algorithms, while using only $\Ot(n)$ space (\Cref{rst:upper-bound}). But, reducing the space to $n^{1-o(1)}$ prohibits us from getting competitive approximations even in $\polylog{(n)}$ passes (\Cref{rst:lb-near-linear}), and reducing the space further to $n^{o(1)}$ prohibits us from even distinguishing between inputs with $n^{1/2-o(1)}$ factor gap between their optimal HC cost (\Cref{rst:lb-polylog-space}). Finally, even increasing the space to $o(n^2)$ is not going to remove the need for approximation (\Cref{rst:lb-exact-solution}). 

\subsection{Our Techniques}
\label{subsec:technique}
\paragraph{Algorithmic results.} 
Dasgupta's work~\cite{Dasgupta16} on introducing the objective cost for HC resulted in beautiful connections between HC and standard cut-based graph problems such as sparsest cut and balanced cut. For instance, the algorithm proposed by~\cite{Dasgupta16} for HC is to recursively partition the vertices of the graph across an approximate sparsest cut at each level of the hierarchical tree until we reach the leaf-nodes. The work of~\cite{Dasgupta16} shows that approximation ratio of this algorithm is $O(\alpha(n) \cdot \log{n})$ where $\alpha(n)$ is the approximation ratio of the black-box sparsest cut algorithm we use, and follow up works in~\cite{charikar2017approximate,cohen2019hierarchical} improved the analysis to an $O(\alpha(n))$ approximation. 

When it comes to graph streaming, many of cut-based problems including sparsest cut and balanced cut have a standard solution using \emph{Cut Sparsifiers}~\cite{BenczurK15}: these are (re-weighted) subgraphs of the input graph that preserve the value of every \emph{global} cut approximately, while being quite sparse with only $\Ot(n)$ edges. By now, there are simple streaming algorithms for recovering cut sparsifiers in $\Ot(n)$ space (see, e.g.~\cite{McGregor14}), and it is easy to see that running a non-streaming algorithm for the cut-based problem on this sparsifier, results also in solutions of approximately the same quality on the original graph. Yet, this recipe does \emph{not} apply to HC: in the aforementioned connection of sparsest cut and HC, one needs to solve sparsest cut recursively on \emph{induced} subgraphs of the input after the first level of recursion -- this in turn requires our cut sparsifier to not only preserve global cuts but also induced cuts, i.e., the weight of edges between any two subsets $(A,B)$ of vertices (not only $(A,\bar{A})$). It is easy to see that such a ``sparsifier'' requires to store all edges of the graph! 

Our main algorithmic contribution in this paper is to bypass this challenge. Instead of considering each separate (induced) cut that may appear when running standard HC algorithms such as~\cite{Dasgupta16,charikar2017approximate}, we prove a ``global'' structural property of cut sparsifiers for HC directly: the HC cost of any \emph{balanced} hierarchical tree\footnote{By a balanced tree, we mean a tree where at every node, the size of sub-trees of each child-node is within a constant factor of the other ones. See \Cref{def:beta-tree} for the formal definition.} as a whole remains almost the same between the original graph and its cut sparsifier (even though costs of some sub-trees can deviate dramatically). Given that the HC algorithm  of~\cite{charikar2017approximate} also optimizes only over balanced hierarchical trees, we obtain that running that algorithm over the sparsifier, instead of the entire graph, will result in a solution with only a constant factor worse approximation guarantee. This way, we get a general recipe for solving HC using cut sparsifiers also which is applicable to graph streaming among other models such as MPC mentioned earlier (we further show that any HC problem admits an $O(1)$-approximation balanced solution, so restricting ourselves to balanced solutions is never going to cost us much). 

\paragraph{Lower bound results.} The starting point of our lower bound in~\Cref{rst:lb-near-linear} is the streaming lower bound for the (noisy) gap cycle counting problem of~\cite{AssadiN21}\footnote{We note that we use a slight variation of the problem that follows immediately from~\cite{AssadiN21} but is somewhat different from the description in that work.}. Informally speaking,~\cite{AssadiN21} proved that any $\polylog{(n)}$ pass algorithm that can distinguish between graphs composed of vertex-disjoint cycles of length $\Theta(n)$ or vertex-disjoint cycles of length $\polylog{(n)}$ requires $n^{1-o(1)}$ space (the actual problem definition involves also some ``noisy'' paths; see~\Cref{sec:lower}). Using the result of~\cite{Dasgupta16} that characterizes the HC cost of vertex-disjoint graphs as well as cycles, one can show that the cost optimal hierarchical tree differs by a factor of $\Theta(\frac{\log{n}}{\log\log{n}})$ between these two family of graphs, which implies our desired lower bound as well via a reduction to~\cite{AssadiN21}.

Our~\Cref{rst:lb-polylog-space} is considerably more involved and is our main contribution on the lower bound front. The main challenge is that to prove a strong approximation lower bound, we can no longer rely on using ``loosely-connected'' graphs such as cycles (as their optimal HC cost is not going to be that different between the two cases). Because of this, we introduce the \emph{One vs. Many Expanders (OvME)} problem wherein the goal is to distinguish between graphs consisted of a single expander with $\Theta(\log{n})$-degree and $\Theta(\log{n})$-edge expansion (see \Cref{def:graph-expansion}), versus $n^{1/2+o(1)}$ vertex-disjoint expanders with the same guarantees. A simple argument, using properties of expanders, allows to bound the difference in the optimal HC cost between these two families with an $n^{1/2-o(1)}$ factor. The bulk of our effort is then to prove the lower bound for this family of input graphs which are inherently different from cycles\footnote{E.g., being expanders they are way-more well connected and have much shorter diameter; see~\cite{KapralovKS15,AssadiN21,KapralovMTWZ22} for the role of these parameters in prior lower bounds. Note also that the result of~\cite{KapralovKS15,KapralovK19} can be seen as proving a lower bound for distinguishing between a single expander versus \emph{two} expanders as opposed to $n^{1/2-o(1)}$ many in our work.}. On a (very) high level, the proof of this lower bound is by $(i)$ designing a multi-party \emph{communication game} in spirit of~\cite{KapralovKS15,ChouGV20,ChouGSV21a} and reducing it to a two-party one using a standard hybrid argument in~\cite{KapralovKS15}, $(ii)$ applying a \emph{decorrelation step} to this game to reduce the problem to proving a low-probability-of-success lower bound in spirit of~\cite{AssadiN21}, and $(iii)$ using a Fourier analytic method originated in~\cite{GavinskyKKRW07} based on \emph{KKL inequality}~\cite{KahnKL88}, to establish the lower bound (our decorrelation step, based on a new notion of ``advantage'' of protocols using KL-divergence, is the one that greatly deviates from prior work in~\cite{KapralovKS15,ChouGV20,ChouGSV21a,AssadiN21} and allows us to use Fourier analytic tools to analyze our final problem, despite its considerable differences from prior problems). 

Finally,~\Cref{rst:lb-exact-solution} is established using a reduction from the \emph{Index} communication game~\cite{Ablayev93}. We create a family of graphs consisting of  $\Theta(n)$-vertex ``near-cliques'' with few edges between them so that the  value of optimum HC cost depends on a single edge in the input graph, which cannot be detected by an $o(n^2)$-space streaming algorithm. Our proof of this part extends prior work of~\cite{Dasgupta16} on characterizing optimal HC costs on paths and cliques, to slightly more complex graphs. 

\paragraph{Recent Independent Work.} Independently of our work,~\cite{AgarwalKLP22} also studied HC under Dasgupta's cost function in the settings similar to our paper. Whereas our focus has been primarily in the streaming setting
and space complexity of algorithms,~\cite{AgarwalKLP22} focused on designing sublinear {time} algorithms in the query model and sublinear communication algorithms in the MPC model. 
But, similar to our~\Cref{rst:upper-bound} (and~\Cref{thm:sparsification} specifically), 
they also prove a general structural result that shows that a cut sparsifier can be used to recover a $(1+o(1))$-approximation to the underlying HC instance, which is stronger than our $O(1)$-approximation guarantee. As a result, they can also recover our~\Cref{rst:upper-bound} with improved leading constant in the approximation. This improvement also applies to the MPC model where they show that $\tilde{O}(n)$ memory per machine suffices to get a $2$-round algorithm that achieves $(1+o(1))$-approximation (they  prove that any \emph{one}-round $\polylog{(n)}$-approximation MPC algorithm requires $\Omega(n^{4/3-o(1)})$ memory per machine). Beside this algorithmic connection, the rest of our work and~\cite{AgarwalKLP22} are entirely disjoint.

\section{Preliminaries}
\label{sec:prelim}

\label{subsec:prob-def}
In this section, we define the notation to be used throughout the paper, and introduce the notion of hierarchical clustering trees and define Dasgupta's cost function~\cite{Dasgupta16}. 
\paragraph{Notation.} As standard in the literature, we denote a graph $G=(V,E,w)$ with $V$ as the set of the vertices, $E$ as the set of the edges, and $w:E\rightarrow \mathbb{R}^{+}$ be the edge weights. For any subset of vertices $A \subseteq V$, we use $\bar{A}=V \setminus A$ to denote the complementary set of vertices in $G$. We refer to any disjoint sets $A,B \subseteq V$ of vertices in V as a \emph{cut} $(A,B)$. If we further have $B=\bar{A}$ (i.e., cut $(A, V\setminus A)$), then the cut is called a \emph{global} cut. For a cut $(A,B)$, the set of \emph{cut edges} is the set of edges that are between $A$ and $B$, denoted by $\delta(A,B)$. We denote the weight of a cut as $w(A,B) = \sum_{e \in \delta(A,B)} w(e)$.

\subsection{Problem Definition}

Let $G=(V,E,w)$ be an input weighted (undirected) graph, and let $\TT$ be a rooted tree whose leaf nodes correspond to the vertices of $V$. Furthermore, let each internal node $z$ of $\TT$ induce a cluster, and the child nodes of $z$ induce the inclusion-wise sub-clusters. Throughout the rest of the paper, we say that any such tree $\TT$ is a \emph{hierarchical clustering tree} (\emph{HC-tree} for short) of $G$. For any node $z$ in this tree, we define $\TT[z]$ as the sub-tree of $\TT$ rooted at $z$, and let $\leaves(\TT[z]) \subseteq V$ denote the set of leaf-nodes of $\TT[z]$.
We define the set of \emph{clusters induced by $\TT$} to be $\{C \subseteq V \mid \textrm{there is a node } z \textrm{ of } \TT \textrm{ such that } C = \leaves(\TT[z])\}$. 

%We are now ready to define Dasgupta's cost function based on hierarchical clustering trees. 
\begin{problem}[HC under Dasgupta's cost function]\label{prob:HC}
Given an $n$-vertex weighted graph $G=(V,E,w)$ with vertices corresponding to data points and edges measuring their similarity, create a rooted tree $\TT$ whose leaf-nodes are $V$. The goal is to \emph{minimize} the cost of this tree $\TT$ defined as 
    \begin{align}
    \label{eq:hc-cost}
       \cost_G(\TT) := \sum_{e=(u,v) \in E} w(e) \cdot |\textnormal{leaf-nodes}(\TT[u \vee v])|,
    \end{align}
    where $|\text{leaf-nodes}(\TT[u \vee v])|$ is the number of leaf-nodes in the sub-tree of $\TT$ rooted at the lowest common ancestor of $u$ and $v$, denoted by $u \vee v$. 
    
    We use $\OPT(G)$ to denote the cost of an optimal tree for the graph $G$. 
\end{problem}

\subsection{Standard Results on Dasgupta's Hierarchical Clustering Cost}
\label{subsec:hc-stanard-facts}
There have been a fruitful collection of results on understanding the cost function in \Cref{eq:hc-cost} since the work of Dasgupta~\cite{Dasgupta16}. In this section, we present some known results for hierarchical clustering that lay the foundations of our paper. 

\subsubsection*{Optimal Hierarchical Clustering Trees}
We first give a collection of lemmas that characterize the behavior for optimal HC trees. These proofs can all be found in~\cite{Dasgupta16} (or follow immediately from there). %The proofs of the lemmas can be found in \Cref{app:hc-standard-facts}.

We start with the following observations for the optimal costs of the HC trees.
\begin{observation}[\!\!\cite{Dasgupta16}] 
\label{obs:partition}
	Suppose $G$ is any graph, $A$ and $\bar{A}$ are 
	two disjoint subsets of vertices in $G$, and $G_{A}$ and $G_{\bar{A}}$ are induced subgraphs of $G$ on vertices $A$ and $\bar{A}$, respectively. Then, 
	\[
		\OPT(G_A) + \OPT(G_{\bar{A}}) \leq \OPT(G). 
	\]
\end{observation}

\begin{observation}[\!\!\cite{Dasgupta16}]
\label{obs:opt-binary}
	Let $G$ be any graph and $\TT$ be a HC tree for $G$. Then, for every $\TT$ that is not binary, there exists a binary $\TT'$ such that $\cost_G(\TT')\leq \cost_G(\TT)$.
\end{observation}

In \Cref{obs:partition}, the `equals to' relation is attained by graphs of vertex-disjoint disconnected components. More formally, we have

\begin{lemma}[\!\!\cite{Dasgupta16}]\label{lem:disjoint}
	Let $G$ be a vertex-disjoint union of graphs $A,B$. Then, 
	\[
	\OPT(G) = \OPT(A) + \OPT(B).
	\] 
\end{lemma}

As a result of \Cref{lem:disjoint}, for any optimal HC tree on vertex-disjoint union of graphs $A,B$, the top-level node always splits $A$ and $B$.

The following lemmas capture the optimal HC costs on paths and cycles. 
\begin{lemma}[\!\!\cite{Dasgupta16}]
\label{lem:path}
	Let $P_m$ denote a path of length $m$. Then, $\OPT(P_m) = m\log{m} + O(m)$. 
\end{lemma}
The basic idea for \Cref{lem:path} is that since an optimal tree is always binary (\Cref{obs:opt-binary}), the optimal strategy is to `balance' the cost at each level and the cost it incurred for all lower levels. Therefore, by a balanced-tree recursion argument, the optimal cost for splitting a line is to \emph{always split as balanced as possible}, which results in a cost of  $m\log{m} + O(m)$. The following is a simple analogue of this lemma for cycles and follows immediately from~\Cref{lem:path}. 

\begin{lemma}
\label{lem:cycle}
	Let $C_m$ denote a cycle of length $m$. Then, $\OPT(C_m) = m\log{m} + O(m)$. 
\end{lemma}
\begin{proof}
	The first cut of $C_m$ has to cut two edges and partition into two paths, thus, 
	\begin{align*}
	\cost(C_m) &= 2m + \min_{\substack{0\leq m_{1} \leq m-2\\ 0\leq m_{2} \leq m-2 \\ m_{1}+m_{2}=m-2}}\OPT(P_{m_{1}}) + \OPT(P_{m_{2}}),
	\end{align*}
	Applying~\Cref{lem:path} for each of $P_{m_1}$ and $P_{m_2}$ finalizes the proof. 
\end{proof}

Finally, we have the following trivial upper bound on the maximum costs of HC on any graph, by simply splitting all edges in the first level.  
\begin{fact}\label{lem:hc-max-cost-ub}
	For any graph $G$ with $m$ edges and $n$ vertices, $\OPT(G) \leq m \cdot n$. 
\end{fact}

\subsubsection*{Hierarchical clustering cost as a function of cuts}
We now show that the cost function in \Cref{eq:hc-cost} can be represented as a function of cuts in the subgraphs of $G$. 
By \Cref{obs:opt-binary}, we can assume w.log. that the HC-tree is binary. For each non-leaf-node $z$ of $\TT$, we associate a cut $(A,B)$, denoted by $\cut(\TT[z])$.  Let $z_1$ and $z_2$ be the child nodes of $z$ in $\TT$, such that $|\leaves(\TT(z_1))| \leq |\leaves(\TT(z_2))|$. Then, we set $A := \leaves(\TT(z_1))$ and $B:= \leaves(\TT(z_2))$. 
Observe that in \Cref{eq:hc-cost}, the multiplicative factor of $w(e)$ for each edge $e=(u,v)\in E$ is equal to the minimum cluster size for $u$ and $v$ to be in the same cluster. Hence, we can alternatively write $\cost_G(\TT)$ in~\Cref{eq:hc-cost} as follows: 
\begin{align}
	\cost_G(\TT) = \sum_{\substack{\text{$(A,B):= \cut(\TT[z])$ for} \\ \text{internal nodes $z$ of $\TT$}}} w(A,B) \cdot \card{A \cup B}. \label{eq:hc-cost-cut}
\end{align}

\subsubsection*{Approximately optimal hierarchical clustering trees as balanced trees}

Dasgupta's work proved that finding the optimal trees for the hierarchical clustering function is NP-hard \cite{Dasgupta16}. Therefore, major efforts to study efficient HC algorithms have been devoted to \emph{approximation algorithms}. It is known that we can find an approximation of the optimal hierarchical clustering by recursively applying approximate \emph{balanced minimum cuts} on the graph. More formally, we define balanced cuts and balanced trees as follows.

\begin{definition}[\textbf{$\beta$-Balanced Cuts and Trees}]\label{def:beta-tree}
For any parameter $\beta$ such that $0 < \beta < 1$, we say that a cut $(A,B)$ is \emph{$\beta$-balanced} if 
\[ \max\{\card{A}, \card{B}\} \leq (1-\beta) \cdot \card{A \cup B}.\] 
A $\beta$-balanced cut $(A,B)$ is said to be a $\beta$-balanced minimum cut if for any $\beta$-balanced cut $(A',B')\neq (A,B)$, there is $w(A,B)\leq w(A', B')$. Moreover, we say a HC tree $\TT$ is a \emph{$\beta$-balanced tree} if for every non-leaf node $z$ of $\TT$, $\cut(\TT[z])$ is $\beta$-balanced.
\end{definition}

One way to create $\beta$-balanced trees is to recursively apply the $\beta$-balanced minimum cuts to the induced subgraphs, formally defined as follows.

\begin{definition}[\textbf{Recursive $\beta$-balanced Min-cut Procedure}]\label{def:beta-min-procedure}
We say a HC tree $\TT$ is obtained by the \emph{recursive $\beta$-balanced min-cut procedure} on $G$ if for each non-leaf node $z$ of $\TT$, the $\cut(\TT[z])$ is obtained by a $\beta$-balanced minimum cut $(A,B)$ on the subgraph induced by $\TT[z]$. 
\end{definition}

It is known by \cite{charikar2017approximate} that if one applies the procedure in \Cref{def:beta-min-procedure}, it is possible to get a constant approximation of the optimal HC tree. 
\begin{lemma}[cf. \cite{charikar2017approximate}]\label{thm:beta-exist}
	For any graph $G=(V,E,w)$, let $\TT_{balanced}$ be a $(1/3)$-balanced tree obtained by the procedure in \Cref{def:beta-min-procedure} with $\beta=\frac{1}{3}$. There is 
	\[
	\cost_G(\TT_{balanced}) \leq 9 \cdot \OPT(G).
	\]
\end{lemma}

\Cref{thm:beta-exist} was previously proved in \cite{charikar2017approximate} with an unspecified constant ($O(1)$), and we provide a self-contained proof with the exact constant in \Cref{app:proof-beta-tree-approx}. 

Note that \Cref{thm:beta-exist} is structural and computing balanced minimum cut itself is not an easy task. Indeed, finding the exact balanced minimum cut is a NP-hard problem. However, it is known that one can compute an $O(\sqrt{\log{n}})$ approximation for balanced minimum cut in polynomial time by~\cite{AroraHK04}. As such, we can obtain an $O(\sqrt{\log{n}})$-approximation algorithm for $\OPT(G)$ by recursively applying the $O(\sqrt{\log{n}})$-approximate $1/3$-balanced cut.
\begin{proposition}[cf. \cite{charikar2017approximate}]
\label{prop:beta-compute}
There exists a polynomial-time algorithm that given a weighted undirected graph $G=(V,E,w)$, computes a $\frac{1}{3}$-balanced HC-tree $\TT$ such that 
\begin{align*}
\cost_G(\TT) \leq O(\sqrt{\log{n}})\cdot \OPT(G).
\end{align*}
\end{proposition}
We provide a self-contained proof of \Cref{prop:beta-compute} in \Cref{app:proof-beta-tree-approx}.

\subsection{Basic Graph Algorithms Backgrounds}
\label{subsec:def}

In this section, we review a few standard graph algorithm definitions and results related to graph expansion and cut sparsifiers. 
\subsubsection*{Graph expansion} 
For a weighted graph $G=(V,E,w)$, we define the graph (edge) expansion as follows.
\begin{definition}
\label{def:graph-expansion}
The edge expansion of a graph $G=(V,E,w)$ is 
\begin{align*}
\Phi_{G} = \min_{\substack{S \subseteq V\\ \card{S}\leq \frac{n}{2}}}\frac{w(S, \bar{S})}{\card{S}}, 
\end{align*}
where $w(S, \bar{S})$ is the total edge weights between $S$ and $\bar{S}$. 
\end{definition}

The notion of edge expansion gives us a convenient tool to control the upper and lower bound of the hierarchical clustering cost, which is crucial to our proof in \Cref{sec:lb-polylog-memory}.

\subsubsection*{Cut sparsifiers}
We now describe the notion of cut sparsifers. On the high level, a cut sparsifier aims to `sparsify' the edges by redistributing the weights to certain `key edges'. By only storing a substantially smaller number of edges, the resulting graph can still maintain the weight of any \emph{global} cut by a small approximation factor. Formally, 

\begin{definition}[\textbf{Cut Sparsifier}]\label{def:cutsparse}
	Given a graph $G=(V,E,w_G)$, we say that a weighted subgraph $H :=(V,E_H,w_H)$ is a $(1\pm\eps)$-cut sparsifier of $G$ if for all non-empty $A \subset V$, the following holds:
	\[
	(1-\eps)\cdot w_G(A,\bar{A}) \leq w_H(A,\bar{A}) \leq (1+\eps)\cdot  w_G(A,\bar{A}),
	\]
	where $w_G(A,\bar{A})$ (resp. $w_H(A,\bar{A})$) denotes the  weight of cut-edges in $(A,\bar{A})$ in $G$ (resp. in $H$). 
\end{definition}

The work by Bencz{\'{u}}r and Karger \cite{BenczurK96} first shows that such a spasifier exists for any graph, and it can be constructed in polynomial time. Furthermore, in the graph streaming model, it is known that with $\tilde{O}(n)$ memory, one can achieve an $\eps$-sparsifier in a \emph{single} pass. 

\begin{proposition}[\!\!\cite{AhnG09,McGregor14}]\label{prop:streaming-cut-sparsifier}
There exists a single-pass streaming algorithm that given a graph $G=(V,E,w)$, computes a $(1\pm\eps)$-cut spasifier of $G$ with a memory of $O(\frac{n\log^3(n)}{\eps^2})$ words. 
\end{proposition}

\subsection{Standard Definitions from Information-Theory and Fourier Analysis}\label{subsec:info}

Finally, we review basic definitions from information-theory and Fourier analysis that we use in our paper. \Cref{sub-app:info-theoretic-facts} contains the details on these definitions and their key properties for us. 

\begin{definition}[\textbf{KL-divergence}]
\label{def:kl-div}
Let $X$ and $Y$ be two discrete random variables supported over the domain $\Omega$ with distributions $\mu_{X}$ and $\mu_{Y}$. The \textnormal{\textbf{KL-divergence}} between $X$ and $Y$ is defined as 
\begin{align*}
\kl{X}{Y} := \sum_{\omega\in \Omega} \mu_{X}(\omega)\log\paren{\frac{\mu_{X}(\omega)}{\mu_{Y}(\omega)}}.
\end{align*}
\end{definition}

We shall note that KL-divergence does \emph{not} satisfy triangle inequality; however, it does admit a chain-rule which plays an important role in our proofs. 

\begin{definition}
\label{def:tvd}
Let $X$ and $Y$ be two discrete random variables supported over the domain $\Omega$ with distributions $\mu_{X}$ and $\mu_{Y}$. The \textnormal{\textbf{total variation distance (TVD)}} between $X$ and $Y$ is defined as 
\begin{align*}
\tvd{X}{Y} := \frac{1}{2}\sum_{\omega\in \Omega}\card{\mu_X(\omega)-\mu_Y(\omega)}. 
\end{align*}
\end{definition}

The total variation distance is a metric and satisfies triangle inequality; it is also closely related to the probability of success of a \emph{Maximum Likelihood Estimator} (MLE) for distinguishing a source of a sample. 
TVD can be upper bound via KL-divergence by Pinsker's inequality. 

Finally, we use the following definition of Fourier transform on Boolean hypercube. 

\begin{definition}\label{def:fourier}
The \textnormal{\textbf{Fourier transform}} of a function $f: \set{0,1}^{n} \rightarrow \IR$ is a function $\hf : 2^{[n]} \rightarrow \IR$: 
\[
	\hf(S) := \sum_{x \in \set{0,1}^n} \frac{1}{2^n} \cdot f(x) \cdot \mathcal{X}_S(x),
\]
where $\mathcal{X}_S(x) := (-1)^{\sum_{i \in S} x_i}$. We refer to each $\hf(S)$ as a \textnormal{\textbf{Fourier coefficient}}. 
\end{definition}

% !TeX root = main.tex 
%!TEX root = main.tex

\newcommand{\ftE}[3]{\ensuremath{f^{#1}_{\,#2}\paren{\,#3}}}

\section{A Semi-Streaming Algorithm for Hierarchical Clustering}
\label{sec:upper-bound}

We introduce our main upper bound result in this section, which gives a single-pass streaming algorithm that uses a memory of $\tilde{O}(n)$ words and asymptotically matches the approximation factor of the best \emph{offline} HC algorithms. As mentioned before, the high-level idea of our algorithm is to maintain a $(1\pm \eps)$-cut sparsifier throughout the stream, and run offline HC algorithms on the sparsifier graph. Since \Cref{eq:hc-cost-cut} gives a way of expressing the cost of $\TT$ as sum of costs of a series of induced cuts, the cut sparsifier intuitively `preserves' the quality of cut-based heuristic algorithms. However, the main roadblock for such an idea is that the cut sparsifier only (approximately) preserves the value of \emph{global} cuts and \emph{not} necessarily the \emph{induced} ones (as in \Cref{eq:hc-cost-cut}). In this section, we settle the problem in \Cref{sec:sparsification} by establishing the relationship between global cuts and the cost of the HC-trees. We then present the main algorithm in \Cref{sec:streaming}.

\subsection{A Sparsification Result for Hierarchical Clustering}\label{sec:sparsification}

We now give the formal statement for the relationship between the costs of HC trees on $G$ and on its $(1\pm \eps)$-cut sparsifier $H$ as follows. 

\begin{theorem}\label{thm:sparsification}
	Let $G =(V,E,w_G)$ be any weighted undirected graph, $H =(V,E_H,w_H)$ be an $(1\pm\eps)$-cut sparsifier of $G$ for some $\eps \in (0,1)$, and $\TT$ be any $\beta$-balanced HC-tree on vertices $V$. Then, 
	\begin{align*}
		(1-\eps) \cdot \beta \cdot \cost_G(\TT) \leq \cost_H(\TT) \leq (1+\eps) \cdot \frac{1}{\beta} \cdot \cost_G(\TT). 
	\end{align*}
\end{theorem}

The key step to prove \Cref{thm:sparsification} is the following lemma, which `massages' the cost function in~\Cref{eq:hc-cost-cut} to a series of global cuts, albeit with some loss. This is done by crucially using the balanced property of the tree $\TT$. On the high level, such a `massage' is possible from balanced HC trees in the following sense. Suppose for every cut $(A,B)$, instead of charging $w_{G}(A,B)$ with a $\card{A\cup B}$ multiplicative factor, let us additionally charge the edges in $w_{\text{extra}}(A,B) :=  w_{G}(A, \bar{A})\cup w_{G}(B, \bar{B}) \setminus w_{G}(A,B)$ also with a $\card{A\cup B}$ multiplicative factor. Indeed, this introduces some extra terms to the cost. We will show that the extra costs introduced as such is at most a \emph{constant} factor of the hierarchical clustering cost.

Fix a cut $(A^*,B^*)$ and suppose it is associated with node $u$ in $\TT$. Note that by \Cref{eq:hc-cost-cut}, the edges in $w_{G}(A^*,B^*)$ \emph{never} incur any costs outside the induced subtree $\TT[u]$. Furthermore, for nodes inside the induced subtree $\TT[u]$ (other than $u$ itself), edges in $w_{G}(A^*,B^*)$ inccur costs by contributing to $w_{\text{extra}}(A,B)$, where either $A\cup B\subseteq A^*$ or $A\cup B\subseteq B^*$. Crucially, since $\TT$ is \emph{balanced}, the multiplicative factor $\card{A\cup B}$ on $e \in w_{G}(A^*,B^*)$ decreases \emph{exponentially}. Therefore, the extra contribution for edges in  $w_{G}(A^*,B^*)$ on the internal nodes of $\TT[u]$ other than $u$ follows a geometric series. As such, the overhead of the cost introduced by the global cut terms is at most an $O(1)$ multiplicative factor of the HC cost.

We now formalize the above intuition as the following lemma.
	
\begin{lemma}\label{lem:W(G)}
		Let $G=(V,E,w)$ be any arbitrary graph and $\TT$ be a $\beta$-balanced HC-tree of $G$. Define: 
		\[
			W(G) := \hspace{-1cm}\sum_{\substack{\\ \\ \text{$(A,B):= \cut(\TT[u])$ for} \\ \text{internal nodes $u$ of $\TT$}}} \hspace{-1cm} \!\!\frac12 \cdot (w_G(A,\bar{A})+w_G(B,\bar{B})) \cdot \card{A \cup B}
		\]
		then, 
		\[
			\cost_G(\TT) \leq W(G) \leq \frac1{\beta} \cdot \cost_G(\TT). 
		\]
\end{lemma}
\begin{proof}
Firstly, by~\Cref{eq:hc-cost-cut}, 
	\begin{align*}
		\cost_G(\TT) &= \sum_{\substack{\text{$(A,B):= \cut(\TT[u])$ for} \\ \text{internal nodes $u$ of $\TT$}}} w_G(A,B) \cdot \card{A \cup B} \\ 
		&\leq \hspace{-1cm} \sum_{\substack{\\ \\ \text{$(A,B):= \cut(\TT[u])$ for} \\ \text{internal nodes $u$ of $\TT$}}} \hspace{-1cm} \frac12 \cdot (w_G(A,\bar{A})+w_G(B,\bar{B})) \cdot \card{A \cup B} = W(G), 
	\end{align*}
	simply because $w_G(A,\bar{A}),w_G(B,\bar{B}) \geq w_G(A,B)$ as the set of edges in each term of the LHS is a superset of edges in RHS. This proves the first (and easy) part of the lemma. 

We now show that the parameter $W(G)$ is also not much larger than $\cost_G(\TT)$. Fix an edge $e = (u,v) \in E$. 	Let $P(u)$ and $P(v)$ denote the leaf-to-root paths of $u$ and $v$ in $\TT$, respectively. 
Additionally, let 
\[
P^*(u) := \set{w \in P(u) \mid w \neq u, v \notin \leaves(\TT[w])},
\]
\[
P^*(v) := \set{w \in P(v) \mid w \neq v, u \notin \leaves(\TT[w])}.
\]
That is, the paths $P^*(u)$ and $P^*(v)$ are the portions of $P(u)$ and $P(v)$, which are \emph{strictly} between the leaves and $u \vee v$. 

With these definitions, we can alternatively write $W(G)$ as: 
\begin{align*}
	 &W(G) = \sum_{e = (u,v) \in E} w(e) \cdot \paren{\card{\leaves(\TT[u \vee v])} + \frac12 \sum_{w \in P^*(u) \cup P^*(v)} \card{\leaves(\TT[w])}}. 
\end{align*}
This is because in each of the nodes $w \in P^*(u)$, $u$ belongs to either $A$ or $B$, while $v$ does not belong to either, and thus we get a contribution of $\frac12 w(e)$ in exactly one of $w(A,\bar{A})$ or $w(B,\bar{B})$; this is similarly the case for nodes in $P^*(v)$; finally, $u \vee v$ is the only other node that splits $u$ and $v$ and in this case $e$ contributes $\frac12 w(e)$ to both $w(A,\bar{A})$ and $w(B,\bar{B})$. See~\Cref{fig:W(G)} for an illustration. 

\begin{figure}[h!]
	\centering
	\includegraphics[scale=0.2]{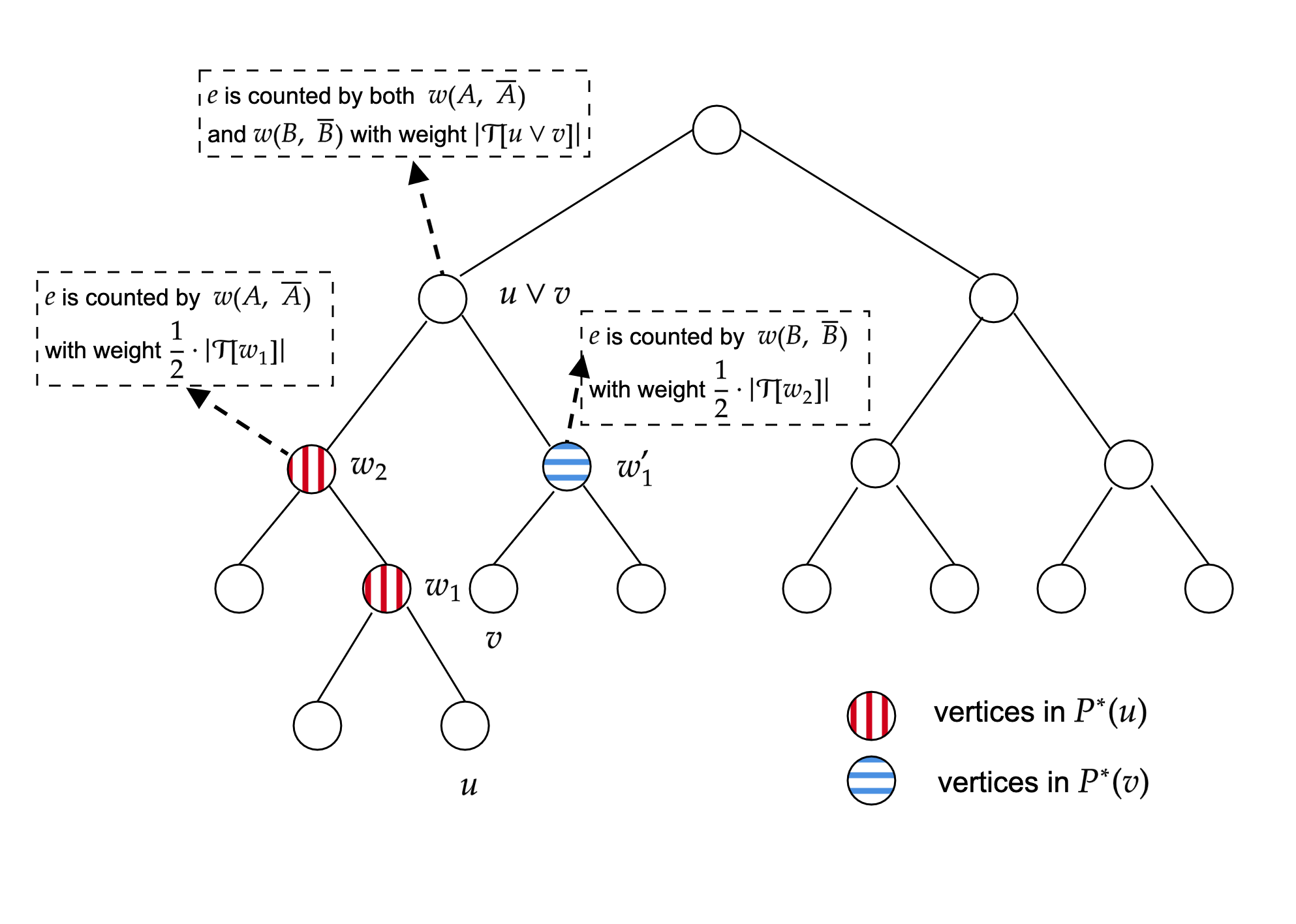}
	\caption{An illustration of the alternative equation for $W(G)$.}\label{fig:W(G)}
\end{figure}

We now use the balancedness of $\TT$ to simplify the above bound further. Let $P^*(u) = (w_1,\ldots,w_k)$, with $w_k$ being a child-node of $u \vee v$, which, for simplicity of notation, we denote by $w_{k+1}$. Considering $\TT$ is 
$\beta$-balanced, we have that for every $i \in [k]$, 
\[
	\card{\leaves(\TT[w_i])} \leq (1-\beta) \cdot \card{\leaves(\TT[w_{i+1}])}. 
\]
As such, $\card{\leaves(\TT[w_i])}$ forms a geometric series and we thus have, 
\begin{align*}
	\sum_{i=1}^{k} \card{\leaves(\TT[w_i])} & \leq \sum_{i=1}^{\infty} (1-\beta)^{i} \card{\leaves(\TT[w_{k+1}])} \\
	&= \frac{1-\beta}{\beta} \cdot \card{\leaves(\TT[w_{k+1}])} \\
	& =  \frac{1-\beta}{\beta} \cdot \card{\leaves(\TT[u \vee v])}. 
\end{align*}
This can similarly be done for $P^*(v)$, thus leaving us with: 
\[
	\frac12 \sum_{w \in P^*(u) \cup P^*(v)} \card{\leaves(\TT[w])} \leq  \frac{1-\beta}{\beta} \cdot \card{\leaves(\TT[u \vee v])}. 
\]
Plugging in this bound in the equation above gives us
\begin{align*}
	W(G) \leq  (1+\frac{1-\beta}{\beta}) \cdot \sum_{e = (u,v) \in E} w(e) \cdot \card{\leaves(\TT[u \vee v])} = \frac{1}{\beta} \cdot \cost_G(\TT), 
\end{align*}
where the final equality is by~\Cref{eq:hc-cost}. 
\end{proof}

\begin{proof}[Proof of~\Cref{thm:sparsification}]
	Consider the values $W(G)$ and $W(H)$ as defined by~\Cref{lem:W(G)} for graphs $G$ and $H$, respectively. Since $W(\cdot)$ is a linear function of weights of global cuts and $H$ is an $\eps$-sparsifier of $G$, 
	we have that, 
	\[
		(1-\eps) \cdot W(G) \leq W(H) \leq (1+\eps) \cdot W(G). 
	\]
	By applying~\Cref{lem:W(G)} for $\cost_G(\TT)$ and $\cost_H(\TT)$, we have that, 
	\begin{align*}
		&\cost_H(\TT) \leq W(H) \leq (1+\eps) \cdot W(G) \leq \frac{1+\eps}{\beta} \cdot \cost_G(\TT), \\
		&\cost_H(\TT) \geq \beta \cdot W(H) \geq \beta \cdot (1-\eps) \cdot W(G) \geq\beta \cdot (1-\eps)  \cdot \cost_G(\TT), 
	\end{align*}
	finalizing the proof. 
\end{proof}

\subsection{A Semi-Streaming Algorithm for Hierarchical Clustering}\label{sec:streaming} 
\Cref{thm:sparsification} implies that a HC tree $\TT$ that works well on the $(1\pm \eps)$-cut sparsifier $H$ also performs well on $G$, provided $\TT$ is balanced. Therefore, we can obtain an algorithm by first maintaining a $(1\pm \eps)$-cut sparsifier, and then finding a balanced HC tree with a good approximation factor on the sparsifier graph. This leads to our main algorithm, presented as follows.

\begin{theorem}\label{thm:streaming}
	There is a single-pass (deterministic) semi-streaming algorithm for hierarchical clustering that uses $O({n\log^3(n)})$ space and achieves an $O(\sqrt{\log{n}})$-approximation in polynomial time and an $O(1)$ approximation in exponential time.
\end{theorem}
\begin{proof}
Throughout the stream, we simply maintain a $(1\pm\eps)$-cut sparsifier $H$ of the input graph $G$ using the algorithms with $O(n\log^3\!{n})$ space, as prescribed in \Cref{prop:streaming-cut-sparsifier} (set $\eps$ as a constant). We then compute an $O(\sqrt{\log{n}})$-approximation to the best $(1/3)$-balanced HC-tree of $H$ using the algorithm in \Cref{prop:beta-compute}.

To analyze the approximation ratio, the resulting $(1/3)$-balanced HC-tree by the algorithm in \Cref{prop:beta-compute} is an $O(\sqrt{\log{n}})$ approximation of the optimal HC tree of $H$. Furthermore, by \Cref{thm:sparsification}, the cost of any $(1/3)$-balanced tree in $H$ remains within an $O(1)$-factor of its cost in $G$. Hence, we have an $O(\sqrt{\log{n}})$ approximation for $\OPT(G)$.

Finally, in exponential time, we can brute-force find the exact minimum $(1/3)$-balanced cut on every subgraph of $H$ induced by the HC tree. By~\Cref{thm:beta-exist}, the HC-tree is a $9$-approximation of the optimal cost on $H$, which provides an $O(1)$-approximation for $\OPT(G)$ by \Cref{thm:sparsification}.
\end{proof}

\begin{remark}
	The algorithms can be extended to dynamic streams by increasing the space by some $\polylog{(n)}$ factors and using randomization -- we simply use a dynamic streaming algorithm of \cite{AhnGM12b} for finding a cut sparsifier instead.
\end{remark}

\section{A Lower Bound for Algorithms with $o(n)$ Memory}\label{sec:lower} 

In \Cref{rst:upper-bound}, we showed that there is a semi-streaming algorithm for the hierarchical clustering problem that asymptotically achieves 
the best approximation ratio possible for offline hierarchical clustering on any graph. The number of passes used by this algorithm is clearly optimal and its space is just within log-factors of its output size, the HC-tree, and is thus again near-optimal. 

Nevertheless, one could consider a potentially more space-efficient algorithm (e.g. $o(n)$-memory) for a simpler variant of the problem where the goal is to simply measure the ``clusterability'' of the input graph, i.e., estimate the \emph{value} (cost) of the optimal solution as opposed 
to returning the entire tree. In this section, we prove that this seemingly easier problem still does not admit a better solution even when allowing $\poly\!\log{(n)}$-passes over the input! Formally, 

\begin{theorem}\label{thm:streaming-lower}
	Any streaming algorithm that can estimate the value of optimal hierarchical clustering on every $n$-vertex graphs with approximation ratio $o(\frac{\log{n}}{\log\log{n}})$ and $\poly\!\log{(n)}$-passes requires $\Omega(n/\poly\!\log{(n)})$ space. 
\end{theorem}

To prove this theorem, we use a reduction from the following variant of the \emph{noisy cycle counting (NOC)} problem of Assadi and N. \cite{AssadiN21}.
\begin{proposition}\label{prop:NOC}
	For infinitely many choices of $n,k$ such that $k < \sqrt{n}$, the following is true. Suppose $\ALG$ is a $p$-pass $s$-space algorithm that distinguishes  the following two families of graphs: 
	\begin{itemize}
		\item[$-$] a vertex-disjoint collection of $2$ cycles of length $n/8$ each and $\frac{3n}{4k}$ paths of length $k$ each; 
		\item[$-$] a vertex-disjoint collection of $\frac{n}{8k}$ cycles of length $2k$ each and $\frac{3n}{4k}$ paths of length $k$ each.
	\end{itemize}
	Then, we have that, 
	\[
		s = \Omega(\frac{1}{p^5} \cdot (n/k)^{1-\gamma \cdot \nicefrac{p}{k}}),
	\]
	for some absolute constant $\gamma \in (0,1)$.\footnote{The extra $k$-paths in the above family are what one considers ``noise''; they are seemingly necessary for the proof of~\Cref{prop:NOC} itself and thus we need
	to prove the reductions \emph{despite} the existence of these extra paths not \emph{because} of their existence.}  
\end{proposition}

The proof of~\Cref{thm:streaming-lower} is by showing that the value of best HC-tree for the two different families of graphs in~\Cref{prop:NOC} differ considerably (for proper choice of parameter $k$). The proof of this separation of the costs crucially relies on the auxiliary lemmas for optimal HC trees in \Cref{subsec:hc-stanard-facts}. 

\begin{proof}[Proof of \Cref{thm:streaming-lower}]
	Let $k = \Theta(\log^{c}{(n)})$ for some fixed constant $c \geq 2$ and suppose $G$ is an $n$-vertex graph from one of the families of graphs in~\Cref{prop:NOC}. Using 
	\Cref{lem:path,lem:cycle,lem:disjoint}, we can infer the following. 
	
	Note that in both cases, each of the $k$-length paths induces a cost of $k\log{k} + O(k)$. In case one, there are two cycles, and each of them incurs a cost of at least $\frac{n}{16} \cdot \log{(\frac{n}{8})} + O(n)$ (the lower bound side of \Cref{lem:cycle}). As such, the total cost is at least
	\begin{align*}
		\OPT(G) \geq 2 \cdot \paren{\frac{n}{16} \cdot \log{(\frac{n}{8})} + O(n)} + \frac{3n}{4k} \cdot \paren{k\log{k} + O(k)} \geq \frac{n}{8}\cdot \log{n}. 
	\end{align*}
	
	On the other hand, in case two, each of the length-$2k$ cycles incurs a cost of at most $2k \cdot \log{(2k)} + O(k)$ (the upper bound side of \Cref{lem:cycle}). Therefore, the total cost is at most
	\begin{align*}
		\OPT(G) \leq \frac{n}{8k} \cdot \paren{2k \cdot \log{(2k)} + O(k)} + \frac{3n}{4k} \cdot \paren{k\log{k} + O(k)} \leq 2n\cdot \log{k}. 
	\end{align*}
	As such, any streaming algorithm that can estimate the value of $\OPT(G)$ to within a factor better than $\frac{\log{n}}{16 \cdot \log{k}}$ can distinguish between the two cases for $G$. 
	
	Considering the choice of $k$, any $o\!\paren{\frac{\log{n}}{\log\log{n}}}$-approximation algorithm would distinguish the graph families of~\Cref{prop:NOC}. Suppose the number of passes of the algorithm is $\sqrt{k} = O(\log^{c/2}{(n)})$. 
	Thus by~\Cref{prop:NOC}, we get that the space of the algorithm is
	\[
		\Omega(\frac{1}{p^5} \cdot (n/k)^{1-\gamma \cdot \nicefrac{p}{k}}) = \Omega(\frac{1}{\poly\!\log{(n)}} \cdot (n/\poly\!\log{(n)})^{1-\gamma \cdot 1/\poly\!\log{(n)}}) = \Omega(n/\poly\!\log{(n)}). 
	\]
	As we can set $c$ to be any arbitrary large constant, we obtain that any $\poly\!\log{(n)}$-pass streaming algorithm for hierarchical clustering requires $\Omega(n/\poly\!\log{(n)})$ space. 
\end{proof}

\newcommand{\nsub}{\ensuremath{n_\textnormal{sub}}\xspace}
\newcommand{\pa}[1]{\ensuremath{\textsf{pa}\left(#1\right)}\xspace}

\newcommand{\OME}{\ensuremath{\textnormal{\textsf{OvME}}}\xspace}
\newcommand{\HLP}{\ensuremath{\textnormal{\textsf{HLP}}}\xspace}
\newcommand{\HLPs}{\ensuremath{\textnormal{\textsf{HLP}}^{\star}}\xspace}

\newcommand{\kplayergame}{\textsf{k-One-vs-Many-Expander}\xspace}
\newcommand{\kplayerabb}{\textsf{k-OvME}\xspace}
\newcommand{\twoplayergame}{\textsf{2-One-vs-Many-Expander}\xspace}
\newcommand{\twoplayerabb}{\textsf{2-OvME}\xspace}
\newcommand{\singleedgegame}{\textsf{Single-Edge-One-vs-Many-Expander}\xspace}
\newcommand{\singleedgeabb}{\textsf{SE-OvME}\xspace}

\newcommand{\baselinecomprob}{\textsf{\chen{UNKNOWN-NAME}}\xspace}

\renewcommand{\prot}{\ensuremath{\Pi}}
\newcommand{\msg}[1]{\ensuremath{\textnormal{\texttt{msg}}(#1)}\xspace}
\newcommand{\out}[1]{\ensuremath{\textnormal{\texttt{out}}(#1)}\xspace}

\newcommand{\protOME}{\prot_{\OME}}
\newcommand{\protHLP}{\prot_{\HLP}}
\newcommand{\protHLPs}{\prot_{\HLPs}}

\newcommand{\unif}{\ensuremath{\mathcal{U}}}

\newcommand{\adv}[1]{\ensuremath{\textnormal{advantage}(#1)}}

\section{A Lower Bound for Algorithms with $\polylog\, {n}$ Memory}
\label{sec:lb-polylog-memory}
In this section, we prove another lower bound that shows that when the space of the algorithm is restricted to just $\poly\!\log\!{(n)}$ bits,  
even distinguishing between `highly clusterable' inputs versus ones that are `very far from being clusterable' is not possible. In particular, we show that,
\begin{theorem}
\label{thm:lb-polylog-space}
Any streaming algorithm that uses $\polylog\, {(n)}$ space cannot estimate the \emph{value} of hierarchical clustering with an approximation ratio of $n^{1/2-\delta}$ for any constant $\delta > 0$ with constant probability strictly better than half. 
\end{theorem}

The proof of~\Cref{thm:lb-polylog-space} is by establishing a novel streaming lower bound of its own independent interest: no $\polylog\, {(n)}$-space streaming algorithm can distinguish between inputs consisting of a single \emph{expander} on the entire set of 
vertices versus a collection of $k = n^{1/2-o(1)}$ vertex-disjoint expanders. It is easy then to prove that the objective value of hierarchical clustering differs by a factor of $n^{1/2-o(1)}$ between the two cases which concludes the proof. Thus, the main
contribution of our work on this front is to establish the mentioned streaming lower bound, formalized as follows. 
 
\begin{theorem}\label{thm:OME}
	For any $\delta \in (0,1/2)$, any streaming algorithm with $o(n^{\delta}/\log{n})$ space cannot distinguish these two families of $n$-vertex (multi-)graphs\footnote{For technical reasons, we allow multi-graphs with edge multiplicity $O(1)$, which is standard; see, e.g. \cite{KapralovKS14}.} with constant probability better than half: 
	\begin{itemize}[leftmargin=15pt]
		\item \textbf{Case 1:} A single expander $G$ on $n$ vertices and $m=10\,n\log{n}$ edges; 
		\item \textbf{Case 2:} A collection of $t :=n^{1/2-\delta}$ vertex-disjoint expanders $G_i$ each on $n_i := n/t = n^{1/2+\delta}$ vertices and $m_i := 10\,n\log{n}/t = 10\,n^{1/2+\delta} \cdot \log{n}$ edges. 
	\end{itemize}
	Here, by an expander, we mean a (multi-)graph with \emph{edge expansion} of $\Omega(\log{n})$ as in~\Cref{def:graph-expansion}. 
\end{theorem}

The problem in~\Cref{thm:OME} is qualitatively similar to the \emph{gap cycle counting} problem studied extensively in the streaming literature (see, e.g., \cite{VerbinY11,KapralovKS14,AssadiKSY20,AssadiN21,KapralovMTWZ22}) wherein the goal is to distinguish between a single Hamiltonian cycle (or a few `long' cycles) and a collection of vertex-disjoint `short' cycles. Owing to its wide range of applications, the gap cycle counting problem has become a staple in graph streaming lower bounds. 
We believe our lower bound for the `expander-variant' of this problem  appears flexible enough to find other applications and is therefore interesting in its own right.

In the following, we first show how~\Cref{thm:lb-polylog-space} follows easily from~\Cref{thm:OME} and then concentrate the bulk of our effort in this section to proving the latter theorem. 

\begin{proof}[Proof of~\Cref{thm:lb-polylog-space} (assuming~\Cref{thm:OME})]
	Suppose we have a streaming algorithm $\alg$ that can estimate the value of hierarchical clustering for every graph $G$ to within a factor $o(n^{1/2-\delta})$ with probability strictly more than half. 
	
	First, consider a graph $G$ according to Case $1$ of~\Cref{thm:OME}. We argue that in this case, $\OPT(G) = \Omega(n^2 \cdot \log{n})$. 
	By \Cref{thm:beta-exist}, we know that the algorithm 
	that picks the minimum $1/3$-balanced cut repeatedly achieves an $O(1)$-approximation to $\OPT(G)$. At the same time, since edge expansion of $G$ is $\Omega(\log{n})$, for any $1/3$-balanced cut $S$, 
	we have that $\card{\delta(S)} = \Omega(n\log{n})$. Thus, the cost of that algorithm on its first level is already $\Omega(n^2\log{n})$, which gives $\OPT(G) = \Omega(n^2\cdot\log{n})$. 
	
	Now, consider a graph $G$ according to Case $2$ of~\Cref{thm:OME}. We have, 
	\[
		\OPT(G) = \sum_{i=1}^{t} \OPT(G_i) \leq t \cdot (n/t) \cdot (10n\log{n}/t) = O(n^2 \cdot \log{n}/t),
	\]
	where the first equality is by \Cref{lem:disjoint} as $G_i$'s are vertex-disjoint components of $G$, and the inequality is by \Cref{lem:hc-max-cost-ub} as each $G_i$ contains $(n/t)$ vertices and $(10n\log{n}/t)$ edges.
		
	Combining the above two arguments, we have that $\OPT(G)$ differs by an $\Omega(t) = \Omega(n^{1/2-\delta})$ factor between the two cases of the problem in~\Cref{thm:OME}. Thus, 
	$\alg$ should be able to distinguish between these two cases with probability strictly more than half. By~\Cref{thm:OME}, we get that $\alg$ has to have space $\Omega(n^{\delta}/\log{n}) \gg \polylog{(n)}$, 
	concluding the proof. 
\end{proof}

\subsection{A High-Level Overview of Proof of~\Cref{thm:OME}}

The proof of~\Cref{thm:OME} is via \emph{communication complexity}, and then using the standard fact that communication complexity can lower bound the space  of streaming algorithms. The communication complexity lower bound
itself goes through several steps as we elaborate below.

\paragraph{Step one: a $k$-party communication problem.} For integers $n,k,t \geq 1$, we define a $k$-party communication problem \emph{One-vs-Many-Expander} ($\OME_{n,k,t}$) on $n$-vertex graphs $G$. 
In $\OME_{n,k,t}$, we have $k$ players and each player $P_i$ receives a \emph{matching} $M_i$ of size $n/4$ on $n$ vertices. In addition, there exists a labeling $\Sigma$ of vertices of $G$ into 
$t$  equal-size classes $\Sigma_1,\ldots,\Sigma_t$. Then, 
\begin{itemize}[leftmargin=15pt]
\item In the \textbf{Yes} case, the input matching of each player is chosen randomly, independent of $\Sigma$. 
\item In the \textbf{No} case, the input matching of each 
player is chosen randomly so that it contains $n/4t$ random edges inside each class $\Sigma_j$ for $j \in [t]$.
\end{itemize} The goal is for the players starting from $P_1$ to each send a message to the next player, so that the last player $P_k$ can output which case the input belongs to. 

We show that proving an $o(n^{\delta}/\log{n})$ communication lower bound for $\OME_{n,k,t}$ for $k=40\log{n}$ and $t=n^{1/2-\delta}$ implies~\Cref{thm:OME}. The proof is by showing that, {with high probability}, the \textbf{Yes}-case of $\OME$ 
results in $G$ corresponding to Case $1$ of~\Cref{thm:OME}, while the \textbf{No}-case is the Case $2$ of that theorem. This argument itself is a simple exercise in random graph theory. 

\paragraph{Step two: a $2$-party communication problem.} In order to prove the lower bound for $\OME_{n,k,t}$, we use a common approach (see, e.g., \cite{KapralovKS14}) and reduce it to a $2$-party problem which we call the \emph{Hidden Labeling Problem} ($\HLP_{n,t})$
on $n$-vertex graphs $G$. In $\HLP_{n,t}$, Alice is given a labeling $\Sigma$ of vertices of $G$ into $t$ equal-size classes $\Sigma_1,\ldots,\Sigma_t$ and Bob is given a single matching $M$ of size $n/4$. The distribution of these labeling $\Sigma$ 
and matching $M$ is the same as the ones in $\OME_{n,k,t}$ (where $M$ can correspond to the input of any one player). 

We prove that an $n^{\delta-o(1)}$ communication lower bound for $\HLP_{n,t}$ for protocols with probability of success $1/2+\Omega(1/k)$ implies our desired lower bound in the previous part for $\OME_{n,k,t}$. 
The proof is via a hybrid argument over the input of $k$ players in $\OME_{n,k,t}$ similar to \cite{KapralovKS14}. 

\paragraph{Step three: a decorrelation step.} We note that the $\HLP_{n,t}$ problem is qualitatively similar to the famous \emph{Boolean Hidden Matching} problem of \cite{GavinskyKKRW07} and many of its variants such as \emph{Boolean Hidden Partition} \cite{KapralovKS14}, or \emph{$p$-ary Hidden Matching} \cite{GuruswamiT19}, and alike (see, e.g., \cite{GuruswamiVV17,ChouGV20}). However, quantitatively, this problem is quite different from all these problems. For instance, all aforementioned problems admit an $\Omega(\sqrt{n})$ 
communication lower bound, while there is a protocol for solving $\HLP$ using $O(\sqrt{n/t} \cdot \log{t})$ communication by focusing only on one class $\Sigma_i$ in Alice's input\footnote{Alice sends $O(\sqrt{n/t})$ vertices of 
her input that belong to the class $\Sigma_1$; in the \textbf{Yes}-case, Bob is unlikely to have any edges inside this set, while in the \textbf{No}-case, one of Bob's edges will belong to this set with a high constant probability.}. As a result, while our lower proof for $\HLP$ borrows ideas from this line of work, and in particular the Fourier-analytic method of \cite{GavinskyKKRW07}, it also requires its own different ideas. 

To prove the lower bound, we first `break the (strong) correlation' on the edges of $M_i$ in the input distribution (see, e.g., \cite{AssadiN21} for a similar argument). This gives us yet another reduction to the following problem, which we denote by $\HLPs_{m}$: 
Alice is given an equipartition $U_0,U_1$ of $m$ vertices and Bob is given a single edge $e$: In \textbf{Yes}-case, the edge $e$ is chosen uniformly among all edges possible on $U_0 \cup U_1$, while in the \textbf{No}-case, the edge 
$e$ is chosen uniformly from either edges entirely in $U_0$ or entirely in $U_1$. We show that for some $m=\Theta(n/t)$, any communication lower bound for $\HLPs_{m}$ for protocols with 
(quite low but non-trivial) probability of success of ${1}/{2}+\tilde{\Theta}(1/{n})$, also implies the same  lower bound for protocols for $\HLP_{n,t}$ that succeed with probability $1/2+\tilde{\Theta}(1)$. We shall note that technically speaking, 
here, we will not consider protocols that solve $\HLPs_{m}$ with certain probability, but rather the ones wherein \emph{KL-divergence} of final `view' of Bob in Yes- and No-cases 
differ by at least $\tilde{\Theta}(1/n)$. This will be crucial for the proof of our next step.

\paragraph{Step four: a low-probability-of-success lower bound.} The very final step of our approach is to prove a lower bound for $\HLPs_m$ that rules out protocols where Bob's view is slightly different between Yes- and No-cases, namely, 
by $\tilde{\Theta}(1/n)$ in KL-divergence.  
This is done using a Fourier-analytic approach initiated in \cite{GavinskyKKRW07}, using the celebrated \emph{KKL inequality} of \cite{KahnKL88}, that allows us to argue any protocol with $c$ bits of communication for $\HLPs_m$ 
can only lead to an advantage of $O((c/m)^2)$ in changing Bob's view of which case the input belongs based on Alice's message. 

Tracing back these parameters implies that to get an advantage of $\tilde{\Theta}(1/n)$ in solving $\HLPs_m$ (as dictated by step three), we need $c$ to be: 
\[
	c = \tilde{\Omega}(\frac{m}{\sqrt{n}}) = \tilde{\Omega}(\frac{n}{t \sqrt{n}}) = \tilde{\Omega}(\frac{\sqrt{n}}{t}) = \tilde{\Omega}(n^{\delta}),
\]
by the choice of $t=n^{1/2-\delta}$ in step one. By plugging in these bounds in the steps two and three, we get a lower bound of $\Omgt(n^{\delta})$ communication for $\OME_{n,k,t}$ for any $k=\tilde{\Theta}(1)$ and $t=n^{1/2-\delta}$. 
Finally, such a lower bound by step one implies our desired streaming lower bound in~\Cref{thm:OME}. This concludes the high level overview of the proof of~\Cref{thm:OME}. 

\subsection{Step One: The One-vs-Many-Expanders Problem}
\label{subsec:kucg-to-HC}

We first give the formal definition of \emph{One-vs-Many-Expanders} ($\OME$) problem, and prove its connection to approximating HC. The problem is defined as follows.
\begin{problem}[One-vs-Many-Expanders (\OME)]
\label{prb:k-player-game}
For  $n,k,t \geq 1$, $\OME_{n,k,t}$ is a 
communication game between $k$ players $P_{1}, P_{2}, \cdots, P_{k}$. The input is a graph $G=(V, E)$ on $n$ vertices. Additionally, there is a labeling $\Sigma$ that partitions vertices of $V$ into $t$ equal-sized classes $(\Sigma_1,\ldots,\Sigma_t)$. The labeling $\Sigma$ is \emph{unknown} to the players. 
For $i \in [k]$, player $P_i$ is given a matching $M_i$ of size $n/4$. We are promised that the input to belongs to one of the following two classes, chosen uniformly at random: 
\begin{itemize}
    \item \textbf{Yes}-case: The input matching $M_i$ to every player $P_i$ is chosen uniformly at random over all possible matchings on $V$. 
    \item \textbf{No}-case: The input matching $M_i$ to every player $P_i$ is chosen uniformly at random from all matchings that have exactly $n/4t$ edges from each class $\Sigma_j$ for $j \in [t]$. 
\end{itemize}

Starting from $P_1$, each player sends a message to the next one and the goal is for $P_k$ to determine whether the input is in \textbf{Yes}-case or the \textbf{No}-case.
\end{problem}

We allow multi-graphs to be created by the definition of~\Cref{prb:k-player-game} so as to not introduce unnecessary correlation between input of players in the Yes-case. 

In the following, we first prove that the inputs  in $\OME_{n,k,t}$, with high probability, results in one expander in \textbf{Yes}-case and $t$ expanders in \textbf{No}-case.

\begin{lemma}\label{lem:OME-expander}
    In $\OME_{n,k,t}$, for $k \geq 40\log(n)$, and $t \leq n^{1/2}$, with probability $1-o(1)$, we have: 
    \begin{itemize}
        \item A graph $G$ sampled from \textnormal{\textbf{Yes}}-case consists is a single expander with edge expansion $\Omega(\log{n})$ and $nk/4$ edges. 
        \item A graph $G$ sampled from \textnormal{\textbf{No}}-case consists of $t$ expanders, each on $n/t$ vertices and $nk/4t$ edges, with edge expansion $\Omega(\log{n})$.
    \end{itemize}
\end{lemma}

Before proving this lemma, we need to introduce a standard result. 
\begin{claim}\label{clm:neg-correlation}
	For any pairs of vertices $u,v$ in $G$ and $i \in [k]$, define an indicator random variable $x_{u,v,i}$ which is $1$ iff $(u,v)$ is sampled in $M_i$. Then, the set of 
	random variables $\set{X_{u,v,i}}_{u,v \in G, i \in [k]}$ are negatively correlated conditioned on $\Sigma$ in both \textnormal{\textbf{Yes}} and \textnormal{\textbf{No}} cases $\OME_{n,k,t}$. 
\end{claim}
\begin{proof}
	Firstly, conditioned on the \textbf{Yes} or \textbf{No} case and $\Sigma$, the choice of $M_i$ and $M_i$ for $i \neq j \in [k]$ are independent. Thus, 
	we only need to show that for every $i \in [k]$, $\set{X_{u,v,i}}_{u,v \in G}$ are negatively correlated. We prove this for any pairs of 
	random variables and one can inductively prove it for all $\set{X_{u,v,i}}_{u,v \in G}$ for every $i \in [k]$ as well.  
	Consider the \textbf{Yes} case first. We have, 
	\[
		(\expect{X_{u',v',i}} =) \expect{X_{u,v,i}} =  \Pr\paren{(u,v) \in M_i} = \frac{n/4}{{{n}\choose{2}}} =  \frac{1}{2 \cdot (n+1)},
	\]
	where the probability calculation holds as each matching is of size $n/4$ chosen uniformly at random. 
	
	Similarly, we have, 
	\begin{align*}
		\expect{X_{u,v,i} \cdot X_{u',v',i}} &= \Pr\paren{(u,v) \in M_i \wedge (u',v') \in M_i} \\
		&\leq \frac{(n/4) \cdot (n/4-1)}{{{n}\choose{2}} \cdot {{n-2}\choose{2}}} \tag{the inequality is tight if $\set{u,v} \cap \set{u',v'} = \emptyset$ and otherwise the probability is zero}\\
		&< \expect{X_{u,v}} \cdot \expect{X_{u',v'}} \tag{by a direct calculation of the bounds}.
	\end{align*}
	By and inductive argument, we can extend this to all subsets of $\set{X_{u,v,i}}_{u,v \in G}$, proving the negative correlation of the variables in this case. 
	
	In the \textbf{No} case, we can repeat the same argument for each individual class $\Sigma_j$ for $j \in [t]$ instead. This finalizes the proof. 
\end{proof}

\begin{proof}[Proof of~\Cref{lem:OME-expander}]
We first prove the result in the \textbf{Yes} case. Let us fix any partition $(S, V\setminus S)$ such that $\card{S}\leq\frac{n}{2}$. 
Consider the random variables $X_{u,v,i}$ for $i \in [k]$ in~\Cref{clm:neg-correlation} for any pairs of vertices $u \in S$ and $v \in V \setminus S$. Define $X_S := \sum_{i=1}^{k} \sum_{u \in S,v \in V \setminus S} X_{u,v}$. 
We have, 
\[
	\expect{X_S} = \sum_{u \in S,v \in V \setminus S} \sum_{i=1}^{k} \expect{X_{u,v}} \leq \card{S} \cdot \frac{n}{2} \cdot \frac{k}{2 \cdot (n+1)} \leq \card{S} \cdot 5\log{n},
\]
as $\card{S} \leq n/2$, by the bound on expectation in~\Cref{clm:neg-correlation}, and since $k \geq 40\log{n}$. 

As $X_S$ is a sum of negatively correlated $\set{0,1}$-random variables by~\Cref{clm:neg-correlation}, we can apply Chernoff bound for negatively correlated random variables (\Cref{prop:chernoff-general}), to have
\begin{align*}
\Pr\paren{X_{S} \leq \card{S} \cdot \log(n)}\leq \exp\paren{-4\card{S} \cdot \log(n)}.
\end{align*}
As such, one can apply union bound for all partitions with size $\card{S}:=s \geq 1$ as
\begin{align*}
\Pr\paren{\exists S \,\, s.t. \,\, (\card{S}=s) \land (X_{S} \leq s \cdot \log(n))}& \leq {n \choose s} \exp\paren{-4s \cdot \log(n)}\leq n^{s}\cdot \exp\paren{-4s\cdot \log(n)} \leq \frac{1}{n^2}. 
\end{align*}
Finally, one can apply union bound for all size $s\leq \frac{n}{2}$, and finalize the statement for the \textbf{Yes} case.

For the proof for the \textbf{No} case, we first observe that no edge will ever be added between two classes $\Sigma_{i}$ and $\Sigma_{j}$, and the edges inside each $\Sigma_{i}$ are exactly of the size $\frac{nk}{4t}$. 
Moreover, distribution of each graph induced on $\Sigma_i$ matches that of \textbf{Yes} case on the whole graph. Thus, we can apply the same argument as before to each $\Sigma_i$ individually and obtain
the same lower bound of $\Omega(\log{n})$ on their expansion. This concludes the proof. 
\end{proof}

By~\Cref{lem:OME-expander}, in order to prove~\Cref{thm:OME}, we need to lower bound the communication complexity of $\OME_{n,k,t}$ for $k=40\log(n)$ and $t=n^{1/2-\delta}$. It is well-known that the one-way communication lower bound implies a single-pass streaming memory lower bound on the same input distribution: the reduction is to simply let each player run the streaming algorithm and send the memory as the message. Thus, the space of the algorithm would be an upper bound 
on the communication in the protocol.

\subsection{Step Two: The Hidden Labeling Problem ($\HLP$)}
\label{subsec:2ucg-to-HC}
To prove a lower bound for $\OME$, we define an intermediate two-player communication problem. 
\begin{problem}[\emph{Hidden Labeling Problem} ($\HLP$)]
\label{prb:2-player-game}
For  $n,t \geq 1$, $\HLP_{n,t}$ is a two player 
communication game between Alice and Bob. We have a 
graph $G=(V,E)$, Alice is given a labeling $\Sigma$ of $V$ into $t$ equal-size classes $(\Sigma_1,\ldots,\Sigma_t)$. 
Bob is given a single matching $M$ of size $n/4$. 
We are promised that the input is one of the following two cases chosen uniformly at random: 

\begin{itemize}
    \item \textnormal{\textbf{Yes}}-case: The matching $M$ of Bob is chosen uniformly at random from all matchings on $V$.
    \item \textnormal{\textbf{No}}-case: The matching $M$ of Bob is chosen uniformly at random from all matchings that contain exactly $n/4t$ edges from each class $\Sigma_j$ for $j \in [t]$. 
\end{itemize}
The goal is for Alice to send a  message to Bob, and Bob outputs which case the input belongs to. 
\end{problem}

Intuitively, if there is a protocol that solves $\OME$ with high probability, it should gain some information about the distribution we used in $\HLP$ also that help outperform random guessing. We formalize this as the following lemma in this step.
\begin{lemma}
\label{lem:reduction-k-2}
Suppose there exists a (possibly randomized) communication protocol that uses $c$ bits and solves $\OME_{n,k,t}$ (\Cref{prb:k-player-game}) with probability at least $1/2+\eps$ for some $\eps > 0$. Then, there exists a deterministic communication protocol that that uses $c$ bits and solves $\HLP_{n,t}$ (\Cref{prb:2-player-game}) correctly with probability at least ${1}/{2}+\eps/k$.
\end{lemma}

We prove~\Cref{lem:reduction-k-2} by a standard argument. We first use a hybrid argument to show the existence of an `informative index' among the message between the players over a hybrid distribution. More formally, we define the distributions $\{\mu(f_{i})\}_{i=0}^{k}$ as follows: 
\begin{itemize}
\item For each $\mu(f_{i})$, let matchings $M_1,\ldots,M_i$ be sampled from the \textbf{No} case of~\Cref{prb:k-player-game}, and the latter $k-i$ matchings be sampled from the \textbf{Yes} distribution. We have, 
\end{itemize}

In the following, fix a protocol $\protOME$ that solves $\OME$ with probability at least $1/2+\eps$. We use $\msg{\protOME}$ to denote the random variable for the messages in $\protOME$, including the final answer. 

\begin{lemma}
\label{lem:hybrid-inf-ind-k-2}
There exists an \textnormal{\textbf{informative index}} $i^{*}\in [k]$ such that 
\begin{align*}
\tvd{(\msg{\protOME} \mid \mu(f_{i^{*}-1}))}{(\msg{\protOME}\mid \mu(f_{i^{*}}))}\geq \frac{\eps}{2k}.
\end{align*}
\end{lemma}
\begin{proof}
Note that by the definition of our hybrid distribution, the distribution in~\Cref{prb:k-player-game} is  
$
\nu = \frac{1}{2}\cdot \paren{\mu(f_{0})+\mu(f_{k})}.$
 Hence, by~\Cref{fact:distinguish-tvd}, to determine whether a draw of $\nu$ is from $\mu(f_{0})$ or $\mu(f_{k})$ with probability at least $1/2+\eps$, there must be that
\begin{align*}
\tvd{(\msg{\protOME} \mid \mu(f_{0}))}{(\msg{\protOME}\mid \mu(f_k))}\geq \frac{\eps}{2}.
\end{align*}
On the other hand, by the triangle inequality of total variation distance, we have
\begin{align*}
&\tvd{(\msg{\protOME} \mid \mu(f_{0}))}{(\msg{\protOME}\mid \mu(f_k))} \\
&\hspace{10pt}\leq \sum_{i=1}^{k} \tvd{(\msg{\protOME} \mid \mu(f_{i-1}))}{(\msg{\protOME}\mid \mu(f_{i}))}.
\end{align*}
An averaging argument now concludes the proof. 
\end{proof}

Based on~\Cref{lem:hybrid-inf-ind-k-2}, we can now design a protocol $\protHLP$ for Alice and Bob to gain an advantage of $\Omega(\frac{\eps}{k})$ for $\HLP$. In what follows, we will use $\msg{\protOME}_{i}$ as the message of player $P_i$ 
and $\msg{\protOME}_{<i}$ as the messages of players $P_1,\ldots,P_{i-1}$ in $\protOME$. 

\clearpage

\begin{tbox}
$\protHLP$: a communication protocol for~\Cref{prb:2-player-game}.

\smallskip

\paragraph{Input:} 
\begin{itemize}
\item $\protOME$: a communication protocol for~\Cref{prb:k-player-game} that uses $c$ bits of communication.
\item The informative index $i^*$ of~\Cref{lem:hybrid-inf-ind-k-2} for $\protOME$
\item $(\Sigma,M)$: the inputs of Alice and Bob as prescribed in~\Cref{prb:2-player-game}.
\end{itemize}
\smallskip

\paragraph{Alice:} 
\begin{enumerate}[label=$(\roman*)$]
    \item Alice samples the first $i^{*}-1$ matchings $M_1,\ldots,M_{i^{*}-1}$ of $\OME$ as input to players $P_1,\ldots,P_{i^*-1}$ following the \textbf{No} distribution conditioned on her input labeling $\Sigma$.
    \item Alice runs $\protOME$ for the first $i^*-1$ players and sends $\msg{\protOME}_{<i^{*}}$ to Bob.
\end{enumerate}

\paragraph{Bob:} 
\begin{enumerate}[label=$(\roman*)$]
    \item Bob picks the input of player $P_{i^*}$ to be the matching $M$ in his input to $\HLP$. 
    \item Bob further samples the inputs to players $P_{i^*+1},\ldots,P_k$ by sampling the matchings $M_{i+1},\ldots,M_k$ from the \textbf{Yes}-distribution (using the fact that in this case, the matchings are independent of $\Sigma$ which is unknown to Bob). 
    \item Bob continues running $\protOME$ on the remaining players using $\msg{\protOME}_{<i^*}$ get the final message $\msg{\protOME}$.
    \item Bob returns the MLE of $\msg{\protOME}$ between the distributions $\mu(f_{i^*-1})$ and $\mu(f_{i^*})$, i.e., returns 
    \textbf{Yes} if $\mu(f_{i^*-1})(\msg{\protOME})\geq \mu(f_{i^*})(\msg{\protOME})$ and \textbf{No} otherwise.
\end{enumerate}
\end{tbox}

It is easy to see that $\protHLP$ is a valid communication protocol with $c$ bits of communication. 
We now prove the correctness of $\protHLP$. 
\begin{claim}
\label{clm:k-inf-ind-dist-equivalence}
The graph $G$ created in $\protHLP$ follows $\mu(f_{i^{*}-1})$ if $(\Sigma,M)$ is a \textnormal{\textbf{Yes}} case, and follows $\mu(f_{i^{*}})$ if $(\Sigma,M)$ is a \textnormal{\textbf{No}} case.
\end{claim}
\begin{proof}
Consider the process of drawing from $\mu(f_{i^{*}-1})$ or $\mu(f_{i^{*}})$: the first $\paren{i^{*}-1}$ coordinates follows the \textbf{No} distribution, which is exactly sampled by Alice. The last $\paren{k- i^{*}}$ coordinates follows the \textbf{Yes} distribution, which is exactly sampled by Bob. The $i^*$-th matching depends on whether $(\Sigma,M)$ is a \textbf{Yes}-case or a \text{No}-case, as desired. 
\end{proof}

\begin{proof}[Proof of Lemma~\ref{lem:reduction-k-2}]
Consider the distribution $\frac{1}{2}\cdot (\mu(f_{i^{*}-1})+\mu(f_{i^{*}}))$ and note that by~\Cref{clm:k-inf-ind-dist-equivalence}, this is 
 the distribution of  graph $G$ when $(\Sigma,M)$ is sampled in $\HLP$. 
By \Cref{lem:hybrid-inf-ind-k-2}, we have 
\begin{align*}
\tvd{(\msg{\protOME} \mid \mu(f_{i^{*}-1}))}{(\msg{\protOME}\mid \mu(f_{i^{*}}))}\geq \frac{\eps}{2k}.
\end{align*}
By \Cref{fact:distinguish-tvd}, we know that Bob can distinguish whether $\msg{\protOME}$ follows $\mu(f_{i^{*}-1})$ or $\mu(f_{i^{*}})$ with probability at least ${1}/{2}+{\eps}/{k}$, which
implies that $\protHLP$ will output the correct answer with the same probability at least. 

Finally, note that the protocol $\protHLP$ designed earlier is randomized. However, given that we are measuring success of the protocol against a fixed hard input distribution, we 
can simply fix its public randomness by an averaging argument and obtain a deterministic protocol with the same probability of success (or equivalently, apply the easy direction of Yao's minimax principle). 
This concludes the proof. 
\end{proof}

\subsection{Step Three: Decorrelation of $\HLP$}
\label{subsec:2ucg-decorrelate}

The challenge in analyzing \HLP directly is that the edges of Bob are highly correlated. As such, we use another type of hybrid argument to decorrelate these edges. That is, we show that if there is a protocol that solves $\HLP$ with 
advantage $\Omega(\frac{1}{k})$, we can construct a protocol that distinguishes a \emph{single edge} from the \textbf{Yes} and \textbf{No} cases of $\HLP$, albeit with an advantage which is roughly a factor $n$ smaller. 
We define the following intermediate problem. 

\begin{problem}[\emph{Single-Edge Labeling Problem} ($\HLPs$)]
\label{prb:HLPs}
For integer $m \geq 1$, $\HLPs_{m}$ is a two player 
communication game between Alice and Bob. We have a 
graph $G=(V,E)$ on $m$ vertices. Alice is given a partitioning of $V$ into two equal-size sets $U_0$ and $U_1$. Bob is given a single edge $e$.
We are promised that the input is one of the following two cases chosen uniformly at random: 
\begin{itemize}
    \item \textnormal{\textbf{Yes}}-case: The edge $e$ of Bob is chosen uniformly at random from all pairs of vertices that are between $U_0$ and $U_1$.
    \item \textnormal{\textbf{No}}-case: The edge $e$ of Bob is chosen uniformly at random from all pairs of vertices that are either both belong to $U_0$ or both to $U_1$. 
\end{itemize}
The goal is for Alice to send a  message to Bob, and Bob outputs which case the input belongs to. 
\end{problem}

We now show that a protocol for $\HLP$ also implies a protocol for $\HLPs$ with non-trivial performance (albeit not measured in terms of the success probability of the protocol). To do so, we need the following definition. 
\begin{itemize}
	\item For any protocol $\protHLPs$ of $\HLPs$ on $(U_0,U_1,e)$, we define the \textbf{advantage} of $\protHLPs$ as 
	\[
		\adv{\protHLPs} := \Exp_{msg}\kl{e \mid \msg{\protHLPs} = msg,\textbf{Yes}}{e \mid \msg{\protHLPs} = msg,\textbf{No}},
	\]
	where $\msg{\protHLPs}$ is the message of Alice to Bob in $\protHLPs$ and $\kl{\cdot}{\cdot}$ is the KL-divergence. 
\end{itemize}

Roughly speaking, the advantage of a protocol is a measure of success of the protocol not in terms of probability of outputting the answer, but rather KL-divergence of the distributions of Bob's view of the input between \textbf{Yes} and \textbf{No} cases, \emph{conditioned} 
on the message he receives from Alice. 
We prove that a good protocol for $\HLP$ in terms of probability of success implies a protocol for $\HLPs$ (on a somewhat smaller instance) with a non-trivial advantage. 

\begin{lemma}
\label{lem:reduction-single-edge}
Let $n,t \geq 1$ be such that $t \leq n^{1/2}$. Suppose there exists a deterministic  protocol $\protHLP$ that uses $c$ bits of communication and solves $\HLP_{n,t}$ (\Cref{prb:2-player-game}) with probability at least $1/2+\eps$ for some $\eps > 1/n$. Then, there also exists a deterministic  protocol $\protHLPs$ for $\HLPs_{m}$ for some $m = \Theta(n/t)$ (\Cref{prb:HLPs}) with $c$ bits of communication and $\adv{\protHLPs} \geq {\eps^2 }/{2n}$. 
\end{lemma}

The proof of this lemma is also based on a hybrid argument, although quite different from that of~\Cref{lem:reduction-k-2}. For the rest of the proof, fix a protocol $\protHLP$ as in~\Cref{lem:reduction-single-edge}. 
By~\Cref{fact:distinguish-tvd},
\[
	\tvd{(M,\msg{\protHLP} \mid \textbf{Yes})}{(M,\msg{\protHLP} \mid \textbf{No})} \geq \frac{\eps}{2}, 
\]
since given only $(M,\msg{\protHLP})$, Bob is able to solve $\HLP$ with probability of success at least $1/2+\eps$. 
By Pinsker's inequality (\Cref{fact:pinsker}), this implies that 
\begin{align}
	\min\set{1,\kl{M,\msg{\protHLP} \mid \textbf{Yes}}{M,\msg{\protHLP} \mid \textbf{No}}}  \geq \frac{\eps^2}{2}, \label{eq:hlp-advantage} 
\end{align}
where we also used the trivial upper bound of $1$ on the total variation distance. 
By the chain rule of KL-divergence (\Cref{fact:kl-chain-rule}), for the LHS of~\Cref{eq:hlp-advantage}, we have, 
\begin{align*}
	&\kl{M,\msg{\protHLP} \mid \textbf{Yes}}{M,\msg{\protHLP} \mid \textbf{No}} \\
	&\hspace{30pt}= \kl{\msg{\protHLP} \mid \textbf{Yes}}{\msg{\protHLP} \mid \textbf{No}}  \\
	&\hspace{60pt}+ \Exp_{msg \mid \textbf{Yes}} \kl{M \mid \msg{\protHLP}=msg, \textbf{Yes}}{M \mid \msg{\protHLP}=msg, \textbf{No}} \\
	&\hspace{30pt}= \Exp_{msg} \kl{M \mid \msg{\protHLP}=msg, \textbf{Yes}}{M \mid \msg{\protHLP}=msg, \textbf{No}},
\end{align*}
as the marginal distribution of $\msg{\protHLP}$ is the same under \textbf{Yes} and \textbf{No} cases (recall that Alice on her own can only guess the correct answer with probability half).

We denote $M=(e_1,\ldots,e_{n/4})$ where $e_i$ is the $i$-th edge we sample in the matching $M$. We further write $M_{<i}$ to denote $e_1,\ldots,e_{i-1}$. Another application of chain rule implies that 
\begin{align*}
	&\Exp_{msg} \kl{M \mid \msg{\protHLP}=msg, \textbf{Yes}}{M \mid \msg{\protHLP}=msg,\textbf{No}} \\
	 &\hspace{20pt}= \sum_{i=1}^{n/4}  \Exp_{msg} \Exp_{M_{<i} \mid msg, \textbf{Yes}} \kl{e_i \mid \msg{\protHLP}=msg, M_{<i}, \textbf{Yes}}{e_i \mid \msg{\protHLP}=msg, M_{<i}, \textbf{No}}.
\end{align*}
Recall that the distribution of $M_{<i} \mid \textbf{Yes}$ is the uniform distribution over all matchings of size $i-1$, independent the input (and thus message) of Alice. We denote this distribution by $\unif_{<i}$ (or $\unif$ if it is clear from the 
context). 
Combining this equation with~\Cref{eq:hlp-advantage} and an averaging argument implies that there exists an index $i^* \in [n/4]$ such that 
\begin{align}
  &\hspace{180pt} \Exp_{M_{<i^*} \sim \unif}  \quad \Exp_{msg} \notag \\
  &\hspace{10pt}\min\set{1,\kl{e_{i^*} \mid \msg{\protHLP}=msg, M_{<i^{*}}, \textbf{Yes}}{e_{i^*} \mid \msg{\protHLP}=msg, M_{<i^*}, \textbf{No}}} \geq \frac{\eps^2}{2n}. \label{eq:info-index-2}
\end{align}
In the rest of the proof, we denote $M_{<i^*}$ and $e_{i^*}$ by $M^*$ and $e^*$  to avoid the clutter. Our goal is now to ``massage'' the LHS of~\Cref{eq:info-index-2} into a more suitable form for obtaining a protocol for $\HLPs$. 
For any matching $M^*$, we define:
\begin{itemize}
	\item  $\Sigma(M^*)$ as the choice of classes of vertices of $M^*$; 
	\item $\event(M^*,\Sigma(M^*))$: the event that the matching $M^*$ matches at most $2n/3t$ vertices from each class $\Sigma_j$ for $j \in [t]$. Note that this event is fully determined by $(M^*,\Sigma(M^*))$.
\end{itemize}
We have the following claim that is based on the fact that $\event(M^*,\Sigma(M^*))$ is quite a likely event. 
\begin{claim}\label{clm:event-M*}
	There exists a choice of $(M^*,\Sigma(M^*))$ such that $\event(M^*,\Sigma(M^*))$ holds and
	\begin{align*}
		 &\hspace{180pt}\Exp_{msg \mid M^*,\Sigma(M^*)} \\
		 &\kl{e^* \mid \msg{\protHLP}=msg, M^*, \Sigma(M^*), \textnormal{\textbf{Yes}}}{e^* \mid \msg{\protHLP}=msg, M^*, \Sigma(M^*),\textnormal{\textbf{No}}} \geq \frac{\eps^2}{4n}.
	\end{align*}
\end{claim}
\begin{proof}
	We can write the LHS of~\Cref{eq:info-index-2} as 
	\begin{align*}
		  &\Exp_{M^*,msg} \min\set{1,\kl{e^* \mid \msg{\protHLP}=msg, M^*,  \textnormal{\textbf{Yes}}}{e^* \mid \msg{\protHLP}=msg, M^*, \textnormal{\textbf{No}}}} \leq \\
		  &\hspace{220pt} \hspace{-25pt}\Exp_{M^*,msg,\Sigma(M^*)}\hspace{-25pt}\\
		  &\hspace{10pt} \min\set{1,\kl{e^* \mid \msg{\protHLP}=msg, M^*, \Sigma(M^*), \textnormal{\textbf{Yes}}}{e^* \mid \msg{\protHLP}=msg, M^*, \Sigma(M^*), \textnormal{\textbf{No}}}}
		  \tag{as conditioning cannot decrease the KL-divergence (\Cref{fact:kl-conditioning})}. 
	\end{align*}
	Consider the probability of the event $\event(M^*,\Sigma(M^*))$. Given that $M^*$ is a matching of size at most $n/4$ chosen over $[n]$ uniformly at random in this case (independent of $\Sigma$), we have that, 
	\[
		\Exp\card{V(M^*) \cap \Sigma_j} = n/2t, 
	\]
	for every $j \in [t]$. Given that $n/t = \omega(\log{n})$ in~\Cref{lem:reduction-single-edge}, an application of Chernoff bound for negatively correlated random variables (which holds by~\Cref{clm:neg-correlation})
	plus union bound implies that 
	\[
		\Pr\paren{\overline{\event(M^*,\Sigma(M^*))}} \leq \sum_{j=1}^{t} \Pr\paren{\card{V(M^*) \cap \Sigma_j} > \frac{4}{3} \cdot \Exp\card{V(M^*) \cap \Sigma_j}} \ll 1/\poly{(n)}. 
	\]
	Thus, we have, 
	\begin{align*}
	&\hspace{220pt} \hspace{-15pt}\Exp_{M^*,msg,\Sigma(M^*)}\hspace{-25pt} \\
	&\hspace{10pt} \min\set{1,\kl{e^* \mid \msg{\protHLP}=msg, M^*, \Sigma(M^*), \textnormal{\textbf{Yes}}}{e^* \mid \msg{\protHLP}=msg, M^*, \Sigma(M^*), \textnormal{\textbf{No}}}} \\
	&\hspace{170pt}\leq\Exp_{M^*,\Sigma(M^*) \mid \event(M^*,\Sigma(M^*))} \quad \Exp_{msg \mid M^*, \Sigma(M^*)} \\
	&\hspace{10pt} \kl{e^* \mid \msg{\protHLP}=msg, M^*, \Sigma(M^*), \textnormal{\textbf{Yes}}}{e^* \mid \msg{\protHLP}=msg, M^*, \Sigma(M^*), \textnormal{\textbf{No}}} + 1/\poly{(n)},
	\end{align*}
	where the $1/\poly{(n)}$ term accounts for the contribution of KL-divergence terms whenever the event $\event(M^*,\Sigma(M^*))$ does not happen, given that we always truncate the value of KL-divergence by at most one in the prior terms. 
	This bound, together with the fact that $\eps > 1/n$ in~\Cref{lem:reduction-single-edge} implies, 
	\begin{align*}
		&\hspace{170pt} \Exp_{M^*,\Sigma(M^*) \mid \event(M^*,\Sigma(M^*))} \quad \Exp_{msg \mid M^*, \Sigma(M^*)} \hspace{-20pt} \\
		&\hspace{10pt}\kl{e^* \mid \msg{\protHLP}=msg, M^*, \Sigma(M^*), \textnormal{\textbf{Yes}}}{e^* \mid \msg{\protHLP}=msg, M^*, \Sigma(M^*), \textnormal{\textbf{No}}} \geq \frac{\eps}{4n}. 
	\end{align*}
	An averaging argument allows us to fix the choice of $M^*$ and conclude the proof. 
\end{proof}

Let us now consider the distribution of underlying variables after we condition on $(M^*,\Sigma(M^*))$ in~\Cref{clm:event-M*}. Define $\Sigma' := (\Sigma'_1,\ldots,\Sigma'_t)$ to be a random 
variable for the choice of classes for remaining vertices. By the definition of $\event(M^*,\Sigma(M^*))$, we have that $\card{\Sigma'_j} \geq n/3t$. We can consider the choice of $(e^*,\Sigma')$ conditioned on $M^*,\Sigma(M^*)$ as 
follows\footnote{The following analogy may help provide more intuition for this step: suppose we have two red balls and three green balls and we want to pick one ball uniformly at random from one of the two colors; we can first choose one random red ball and one random green ball, and then pick one of these two balls uniformly at random.}: 
\begin{enumerate}[label=$(\roman*)$]
	\item Sample $\Sigma^*=(\Sigma^*_1,\ldots,\Sigma^*_t)$ each of size exactly $n/3t$ vertices, conditioned on $M^*,\Sigma(M^*)$ and the event that both endpoints of $e^*$ belong to $\Sigma^*$. Note that by the definition of 
	$\event(M^*,\Sigma(M^*))$ this is valid as each $\Sigma^*_j$ is chosen from a set of \emph{at least} $n/3t$ vertices. 
	\item Sample $j_1,j_2 \in [t]$ uniformly at random. 
	\item Let $f_0$ be an edge chosen uniformly at random with both endpoints from either only in $\Sigma^*_{j_1}$ or only in $\Sigma^*_{j_2}$. Let $f_1$ be an edge chosen uniformly at random with one endpoint from $\Sigma^*_{j_1}$ 
	and another from $\Sigma^*_{j_2}$. 
	\item Define $\mu(0)$ as the distribution of $f_0$ and $\mu(1)$ as the distribution of $f_1$. 
\end{enumerate}
By this construction, we have, 
\begin{itemize}
	\item Distribution of $e^*$ in $\paren{\HLP \mid M^*,\Sigma(M^*), \textbf{Yes}}$  is 
	\[
		\frac{n+3}{2n+3} \cdot \mu(0) + \frac{n}{2n+3} \cdot \mu(1).
	\]
	\item Distribution of $e^*$ in $\paren{\HLP \mid M^*,\Sigma(M^*), \textbf{No}}$ is $\mu(0)$. 
\end{itemize}
To see why this is the case, notice that after committing to $(\Sigma^*_1,\ldots,\Sigma^*_t)$ conditioned on $e^*$ being incident on them, to choose $e^*$ in $\HLP \mid M^*,\Sigma(M^*), \textbf{Yes}$, 
we can first pick two classes uniformly at random, and then pick an edge uniformly at random over these vertices. This way, the edge will have both endpoints inside one of the classes with probability
\[
	\frac{2 \cdot {{n/3}\choose{2}}}{{{2n/3}\choose{2}}} = \frac{2 \cdot (n/3) \cdot (n/3+1)}{(2n/3) \cdot (2n/3+1)} = \frac{n+3}{2n+3}. 
\]
Thus, picking $f_0$ with this probability and $f_1$ with one minus this probability is equivalent to sampling $e^*$ under these conditions. The second case also holds analogously.

We have the following claim using convexity of KL-divergence and an averaging argument. 

\begin{claim}\label{clm:next-step-M*}
	There exists a choice of $j_1,j_2 \in [t]$ and $\Sigma^*_{-(j_1,j_2)} := \Sigma - \set{\Sigma^*_{j_1},\Sigma^*_{j_2}}$, such that 
	\begin{align*}
		 &\hspace{140pt}\Exp_{msg \mid M^*,\Sigma(M^*),j_1,j_2,\Sigma^*_{-(j_1,j_2)}} \\
		 &\kl{f_1 \mid \msg{\protHLP}=msg, j_1,j_2,\Sigma^*_{-(j_1,j_2)}}{f_0 \mid \msg{\protHLP}=msg, j_1,j_2,\Sigma^*_{-(j_1,j_2)}} \geq \frac{\eps^2}{2n}.
	\end{align*}
\end{claim}
\begin{proof}
	To avoid clutter, we define $Y := j_1,j_2,\Sigma^*_{-(j_1,j_2)}$. 
	We expand the LHS of~\Cref{clm:event-M*} as follows: 
	\begin{align*}
		 &\hspace{200pt}\Exp_{msg \mid M^*,\Sigma(M^*)} \\
		 &\kl{e^* \mid \msg{\protHLP}=msg, M^*, \Sigma(M^*), \textnormal{\textbf{Yes}}}{e^* \mid \msg{\protHLP}=msg, M^*, \Sigma(M^*),\textnormal{\textbf{No}}} \\
		 &\hspace{160pt} \leq \Exp_{Y \mid M^*,\Sigma(M^*)} \quad \Exp_{msg \mid M^*,\Sigma(M^*),Y} \\
		 &\kl{e^* \mid \msg{\protHLP}=msg, M^*,Y, \textnormal{\textbf{Yes}}}{e^* \mid \msg{\protHLP}=msg, M^*, Y,\textnormal{\textbf{No}}} \tag{as conditioning cannot
		  decrease KL-divergence (\Cref{fact:kl-conditioning})} \\
			 &\hspace{160pt} \leq \Exp_{Y \mid M^*,\Sigma(M^*)} \quad \Exp_{msg \mid M^*,\Sigma(M^*),Y} \\
		 &\kl{\frac{(n+3) \cdot \mu(0) + n \cdot \mu(1)}{2n+3} \mid \msg{\protHLP}=msg, M^*, Y}{\mu(0) \mid \msg{\protHLP}=msg, M^*, Y} 
		\tag{by the construction stated above} \\
		 &\hspace{160pt} \leq \Exp_{Y \mid M^*,\Sigma(M^*)} \quad \Exp_{msg \mid M^*,\Sigma(M^*),Y} \\
		 &\frac{1}{2} \cdot \kl{\mu(1) \mid \msg{\protHLP}=msg, M^*, Y}{\mu(0) \mid \msg{\protHLP}=msg, M^*, j_1,j_2,\Sigma^*_{-(j_1,j_2)}} 
		\tag{by the convexity of KL-divergence (\Cref{fact:kl-convexity}) and since $(n/2n+3) < 1/2$}. 
	\end{align*}
	This, combined with the RHS of~\Cref{clm:event-M*} and expanding the definition of $Y$ implies that 
	\begin{align*}
		 &\hspace{140pt} \Exp_{j_1,j_2,\Sigma^*_{-(j_1,j_2)} \mid M^*,\Sigma(M^*)} \quad \Exp_{msg \mid M^*,\Sigma(M^*),Y} \\
		 &\kl{f_1 \mid \msg{\protHLP}=msg, M^*, j_1,j_2,\Sigma^*_{-(j_1,j_2)}}{f_0 \mid \msg{\protHLP}=msg, M^*, j_1,j_2,\Sigma^*_{-(j_1,j_2)}} \geq \frac{\eps^2}{2n},
	\end{align*}
	which, together with an averaging argument concludes the proof. 
\end{proof}

 We are now ready to design a protocol $\protHLPs$ 
for $\HLPs_{m}$ for $m=2n/3t$ as follows: 

\begin{tbox}
$\protHLPs$: a communication protocol for~\Cref{prb:HLPs}.

\smallskip

\paragraph{Input:} 
\begin{itemize}
\item $\protHLP$: a communication protocol for~\Cref{prb:2-player-game} that uses $c$ bits of communication.
\item A choice of $M^*,j_1,j_2,\Sigma^*_{-(j_1,j_2)}$ as prescribed by~\Cref{clm:event-M*} and~\Cref{clm:next-step-M*}.
\item $(U_0,U_1,e)$: the inputs of Alice and Bob as prescribed in~\Cref{prb:HLPs}.
\end{itemize}
\smallskip

\paragraph{Alice:} 
\begin{enumerate}[label=$(\roman*)$]
	\item Alice sets $\Sigma^*_{j_1} = U_0$ and $\Sigma^*_{j_2} = U_1$ to obtain the entire input of Alice-player in $\protHLP$. 
	\item She then sends the same exact message as $\protHLP$ on this input to Bob. 
\end{enumerate}
\end{tbox}

Note that as we are only interested in the advantage of the protocol $\protHLPs$ (and not its output), Bob has no task in this protocol. 

\begin{proof}[Proof of~\Cref{lem:reduction-single-edge}]
	By the definition of distributions $\mu(0)$ and $\mu(1)$, as well as the randomness of $U_0$ and $U_1$, we have that in the protocol $\protHLPs$: 
	\begin{align*}
		\distribution{U_0,U_1,e} = \distribution{\Sigma_{j_1},\Sigma_{j_2},e^* \mid M^*,j_1,j_2,\Sigma^*_{-(j_1,j_2)}},
	\end{align*}
	which implies that
	\begin{align*}
		\distribution{\msg{\protHLPs}} &= \distribution{\msg{\protHLP} \mid M^*,j_1,j_2,\Sigma^*_{-(j_1,j_2)}}, \\
		\distribution{e \mid \msg{\protHLPs}, \textbf{Yes}} &=  \distribution{f_1 \mid \msg{\protHLP}, M^*,j_1,j_2,\Sigma^*_{-(j_1,j_2)}} \\
		\distribution{e \mid \msg{\protHLPs}, \textbf{No}} &=  \distribution{f_0 \mid \msg{\protHLP}, M^*,j_1,j_2,\Sigma^*_{-(j_1,j_2)}} \\
	\end{align*}
	As such,
	\begin{align*}
		\adv{\protHLPs} &= \Exp_{msg \sim \protHLPs}\kl{e \mid \msg{\protHLPs} = msg,\textbf{Yes}}{e \mid \msg{\protHLPs} = msg,\textbf{No}} \\
		&=  \hspace{100pt} \Exp_{msg \sim \msg{\protHLP} \mid M^*,j_1,j_2,\Sigma^*_{-(j_1,j_2)}} \\
		&\hspace{1pt} \kl{f_1 \mid \msg{\protHLP}=msg, j_1,j_2,\Sigma^*_{-(j_1,j_2)}}{f_0 \mid \msg{\protHLP}=msg, j_1,j_2,\Sigma^*_{-(j_1,j_2)}},
	\end{align*}
	which is at least $\eps^2/2n$ by~\Cref{clm:next-step-M*}. This proves the bound on $\adv{\protHLPs}$. Given that $\protHLPs$ only communicates 
	a subset of messages communicated by~$\protHLP$, we have that $\protHLPs$ also used at most $c$ bits of communication, concluding the proof (note that $\protHLPs$ is deterministic 
	as long as $\protHLP$ was deterministic). 
\end{proof}

\subsection{Step Four: A Lower Bound for $\HLPs$}

The last step of the proof is the following lemma that gives a lower bound for $\HLPs$. Recall the notion of advantage of a protocol for $\HLPs$ defined in the previous subsection.

\renewcommand{\bias}{\textnormal{bias}}

\begin{lemma}\label{lem:HLP-lower}
    For any integer $m \geq 1$, any deterministic protocol $\protHLPs$ for $\HLPs_m$ with $4\log{m} \leq c \leq m/4$ bits of communication has
    \[
    	\adv{\protHLPs} \leq O(1) \cdot (\frac{c}{m})^2.
    \]
    
\end{lemma}

    To prove this lemma, we need some definition. Let us denote the input of Alice with a string $x \in \set{0,1}^m$ with $\norm{x}_0 = m/2$ where $x_i = 1$ means the $i$-th vertex belongs to $U_1$ and $x_i=0$ means it is in $U_0$. 
    The input of Bob can also be seen as a pair $(i,j)$ sampled from $[m]$. Define a random variable $Z = x_i \oplus x_j$. When the input is a \textbf{Yes}-case, we have that 
    $Z=1$ and when the input is a \textbf{No}-case, $Z=0$. We denote the message of Alice by $\Prot$ in this proof.  
    For a message $\Prot$, let $X(\Prot)$ denote the set of inputs $x$ that are mapped to the message $\Prot$. By the definition of $\HLPs$, conditioned on a message $\Prot$, the input of Alice is chosen uniformly at random from $X(\Prot)$. 
	Define:
	\[
	\bias(\Prot,i,j) := \Pr_{x \sim X(\Prot)}(x_i \oplus x_j=1) - \Pr_{x \sim X(\Prot)}(x_i \oplus x_j = 0),
	\]
	which also gives that 
	\[
		\Pr\paren{Z=1 \mid \Prot,(i,j)} = \frac{1}{2} + \frac{\bias(\Prot,i,j)}{2} \qquad \text{and} \qquad \Pr\paren{Z=0 \mid \Prot,(i,j)} = \frac{1}{2} - \frac{\bias(\Prot,i,j)}{2}.
	\]
	The first part of the proof is to relate $\adv{\HLPs}$ to the $\bias$ of underlying variables. 
	
	\begin{claim}[``advantage is bounded by squared-biases'']\label{clm:adv-bias}
		\[
			\adv{\HLPs} = O(1) \cdot \Exp_{\Prot,(i,j)} \bracket{\bias(\Prot,i,j)^2}.
		\]
	\end{claim}
	\begin{proof}
	By the definition of advantage, 
    \begin{align*}
    	\adv{\protHLPs} &= \Exp_{\Prot} \bracket{\kl{(i,j) \mid \Prot,\textbf{Yes}}{(i,j) \mid \Prot,\textbf{No}}} \\
	&= \Exp_{\Prot}\bracket{\kl{(i,j) \mid \Prot, Z=1}{(i,j) \mid \Prot, Z=0}}.
    \end{align*}
 We now expand this KL-divergence term for any choice of message $\Prot$ of $\protHLPs$: 
    \begin{align*}
    	&\kl{(i,j) \mid \Prot, Z=1}{(i,j) \mid \Prot, Z=0} \\
	&\hspace{10pt} =\sum_{(i,j)} \Pr\paren{(i,j) \mid \Prot, Z=1} \cdot \log{\paren{\frac{\Pr\paren{(i,j) \mid \Prot, Z=1}}{\Pr\paren{(i,j) \mid \Prot, Z=0}}}}  \tag{by the definition of KL-divergence} \\
	&\hspace{10pt} =\sum_{(i,j)} \frac{\Pr\paren{Z=1 \mid \Prot, (i,j)} \cdot \Pr\paren{(i,j) \mid \Prot}}{\Pr\paren{Z=1 \mid \Prot}} \cdot \log{\paren{\frac{\Pr\paren{Z=1 \mid \Prot, (i,j)} \cdot \Pr\paren{Z=0 \mid \Prot}}{\Pr\paren{Z=0 \mid \Prot, (i,j)} \cdot \Pr\paren{Z=1\mid \Prot}}}} 
	\tag{by Bayes' rule} \\
	&\hspace{10pt} =2 \cdot \sum_{(i,j)} {\Pr\paren{Z=1 \mid \Prot, (i,j)} \cdot \Pr\paren{(i,j) \mid \Prot}} \cdot \log{\paren{\frac{\Pr\paren{Z=1 \mid \Prot, (i,j)}}{\Pr\paren{Z=0 \mid \Prot, (i,j)}}}}, 
    \end{align*}
	as $\Pr\paren{Z=0 \mid \Prot} = \Pr\paren{Z=1 \mid \Prot} = 1/2$ since conditioned only on Alice's input/message, the answer, namely, $Z$, is still uniform.  

	By continuing the above expansion of KL-divergence,  we have, 
	   \begin{align*}
    	&\kl{(i,j) \mid \Prot, Z=1}{(i,j) \mid \Prot, Z=0} \\
	&\hspace{10pt} =2 \cdot \sum_{(i,j)}  \Pr\paren{(i,j) \mid \Prot} \cdot \Pr\paren{Z=1 \mid \Prot, (i,j)} \cdot \log{\paren{\frac{\Pr\paren{Z=1 \mid \Prot, (i,j)}}{\Pr\paren{Z=0 \mid \Prot, (i,j)}}}} \\
	&\hspace{10pt} =2 \cdot \sum_{(i,j)}  \Pr\paren{(i,j) \mid \Prot} \cdot \Pr\paren{Z=1 \mid \Prot, (i,j)}\cdot \log\paren{\exp\paren{\frac{\Pr\paren{Z=1 \mid \Prot, (i,j)}-\Pr\paren{Z=0 \mid \Prot,(i,j)}}{\Pr\paren{Z=0 \mid \Prot, (i,j)}}}} \tag{as $1+x \leq e^x$ for all $x \in \IR$} \\
	&\hspace{10pt} =2 \cdot \log{e} \sum_{(i,j)}  \Pr\paren{(i,j) \mid \Prot} \cdot \frac{\Pr\paren{Z=1 \mid \Prot, (i,j)}}{\Pr\paren{Z=0 \mid \Prot, (i,j)}} \cdot \paren{\Pr\paren{Z=1 \mid \Prot, (i,j)}-\Pr\paren{Z=0 \mid \Prot,(i,j)}} \\
	&\hspace{10pt} =2 \cdot \log{e} \sum_{(i,j)}  \Pr\paren{(i,j) \mid \Prot} \cdot \paren{{\frac{1}{2} + \frac{\bias(\Prot,i,j)}{2}}}\paren{{\frac{1}{2} - \frac{\bias(\Prot,i,j)}{2}}}^{-1} \cdot {\bias(\Prot,i,j)} \tag{by the definition of $\bias(\Prot,i,j)$}\\
	&\hspace{10pt} \leq 4 \sum_{(i,j)} \Pr\paren{(i,j) \mid \Prot} \cdot \paren{1+2 \cdot \bias(\Prot,i,j)} \cdot \bias(\Prot,i,j) \tag{as $\log{e} \leq 2$ and $(1/2+x) \cdot (1/2-x)^{-1} \leq (1+2x)$ for small $x$} 
    \end{align*}
  By bringing back the expectation over the choice of $M$, we have, 
  \begin{align*}
  	\Exp_{\Prot} \bracket{ \kl{(i,j) \mid \Prot, Z=1}{(i,j) \mid \Prot, Z=0}} &\leq 4\Exp_{\Prot} \bracket{\sum_{(i,j)} \Pr\paren{(i,j) \mid \Prot} \cdot \paren{1+2 \cdot \bias(\Prot,i,j)} \cdot \bias(\Prot,i,j)} \\
	&= 8\Exp_{\Prot} \Exp_{(i,j) \mid \Prot} \bracket{\bias(\Prot,i,j)^2},
  \end{align*}
  where we used the fact that 
  \begin{align*}
  	\Exp_{\Prot} \bracket{\sum_{(i,j)} \Pr\paren{(i,j) \mid \Prot} \cdot \bias(\Prot,i,j)} &= 	\Exp_{\Prot,(i,j)} \bracket{\Pr\paren{Z=1 \mid \Prot,(i,j)}} - \Exp_{M,(i,j)}\bracket{\Pr\paren{Z=0 \mid \Prot,(i,j)}} \\
	&= \Pr\paren{Z=1} - \Pr\paren{Z=0} = 0,
	\end{align*}
	by the law of total expectation and since $Z \in \set{0,1}$ is chosen uniformly with no conditioning. This concludes the proof of the claim.
	\end{proof}
	
The last main part of the argument is to bound the RHS of~\Cref{clm:adv-bias} using standard tools from Fourier analysis. 
\begin{claim}\label{clm:bias-fourier}
	For any message $\Prot$, 
	\[
	 \Exp_{(i,j)}[\bias(\Prot,i,j)^2] = O(\frac{\log(2^m/\card{X(\Prot)})}{m})^2
	 \]
\end{claim}
\begin{proof}
For any message $\Prot$, define: 
    \begin{itemize}
        \item $f_\Prot: \set{0,1}^m \rightarrow \set{0,1}$ as the characteristic function of $X(\Prot)$.
        \item $\mathcal{X}_{i,j}: \set{0,1}^m \rightarrow \set{0,1}$ as the character function over $[m]$ (see~\Cref{sub-app:info-theoretic-facts}). 
    \end{itemize}
 	We have that
    \begin{align*}
        \bias(\Prot,i,j) &= \Pr_{x \in X(\Prot)}\paren{x_i \oplus x_j = 1} -  \Pr_{x \sim X(\Prot)}\paren{x_i \oplus x_j = 0} = \dfrac{\card{\set{x : x_i \oplus x_j = 1}} - \card{\set{x: x_i \oplus x_j=0}}}{\card{X(\Prot)}}  \\
       &=  - \Exp_{x \sim X(\Prot)}[\mathcal{X}_{i,j}(x)] = \frac{-1}{\card{X(\Prot)}} \sum_{x \in \set{0,1}^m} f_\Prot(x) \cdot \mathcal{X}_{i,j}(x) = \frac{- 2^m \cdot \hat{f}_\Prot(i,j)}{\card{X(\Prot)}} \tag{by the definition of Fourier coefficients}  
    \end{align*}
    At the same time, by KKL inequality (\Cref{prop:kkl}), for any fixed message ${\card{X(\Prot)}}$,
    \begin{align*}
        \sum_{(i,j)} \bias(\Prot,i,j)^2 &= \frac{2^{2m}}{\card{X(\Prot)}^2} \cdot \sum_{(i,j)}\hat{f}_\Prot(i,j)^2 
        \leq \gamma^{-2} \cdot 
        \paren{\frac{2^m}{\card{X(\Prot)}}}^{\frac{2\gamma}{1+\gamma}} \tag{by the previous equation and KKL inequality in~\Cref{prop:kkl}}.
    \end{align*}
    By setting $\gamma = 2 \cdot \log{(2^m/\card{X(\Prot)})}^{-1}$, we get that, for any message $\Prot$, 
    \begin{align*}
        \Exp_{(i,j)}[\bias(\Prot,i,j)^2] = O(\frac{\log(2^m/\card{X(\Prot)})}{m})^2, 
    \end{align*}
    as desired.
\end{proof}

	To conclude the proof of~\Cref{lem:HLP-lower}, we need to consider the expectation in RHS of~\Cref{clm:bias-fourier} over the choices of $\Prot$. To do so, we need a short detour to bound the size of $X(\Prot)$ for a ``typical'' message $\Prot$. 
	Define the following event: 
	\begin{itemize}
		\item $\event(\Prot)$: the set of inputs mapped to the message $\Prot$ satisfies 
		\[
			\card{X(\Prot)} < 2^{-2c} \cdot {{m}\choose{m/2}};
		\]
		recall that ${{m}\choose{m/2}}$ is the number of choices for $x$ as input to Alice. 
	\end{itemize}
	\begin{claim}[``typical messages have large pre-image'']\label{clm:typical-M-large}
		\[
			\Pr\paren{\event(\Prot)} \leq 2^{-c}. 
		\]
	\end{claim}
	\begin{proof}
	We have, 
    \begin{align*}
    	\Pr_{\Prot} (\event(\Prot)) &= {{m}\choose{m/2}}^{-1} \sum_{x} \mathbb{I}(\text{$x$ is mapped to some $\Prot$ such that $\event(\Prot)$ true}) \\
	&= {{m}\choose{m/2}}^{-1} \hspace{-15pt}\sum_{\Prot: \text{$\event(\Prot)$ is true}}\hspace{-15pt} \card{X(\Prot)} \\
	&\leq {{m}\choose{m/2}}^{-1} \hspace{-15pt}\sum_{\Prot: \text{$\event(\Prot)$ is true}}\hspace{-15pt} 2^{-2c} \cdot {{m}\choose{m/2}} \tag{by the definition of $\event(\Prot)$} \\
	&\leq {{m}\choose{m/2}}^{-1} \cdot 2^{c} \cdot 2^{-2c} \cdot {{m}\choose{m/2}} \tag{as there are at most $2^{c}$ messages in total}  \\
	&= 2^{-c}, 
    \end{align*}
    finalizing the proof. 
    \end{proof}
    We now have everything to conclude the proof of~\Cref{lem:HLP-lower}. 
    
    \begin{proof}[Proof of~\Cref{lem:HLP-lower}]
    	For any protocol $\protHLPs$, we have, 
	\begin{align*}
		\adv{\protHLPs} &\leq O(1) \cdot \Exp_{\Prot,(i,j)} \bracket{\bias(\Prot,i,j)^2} \tag{by~\Cref{clm:adv-bias}} \\
		&\leq \Pr\paren{\event(\Prot)} \cdot 1 +  \Exp_{\Prot \mid \overline{\event(\Prot)}} \Exp_{(i,j) \mid \Prot}[\bias(\Prot,i,j)^2] \tag{the event $\event(\Prot)$ is independent of $(i,j)$ and is only a function of $\Prot$} \\
	 &\leq 2^{-c} + O(1) \cdot \paren{\frac{\log(\frac{2^m}{{{m}\choose{m/2}}\cdot 2^{-2c}})}{m}}^2 \tag{by~\Cref{clm:typical-M-large} for the first term and~\Cref{clm:bias-fourier} for the second} \\
	 &\leq 2^{-c} + O(1) \cdot \paren{\frac{2c+\log{m}}{m}}^2 \tag{as ${{m}\choose{m/2}} \geq 2^{m}/m$ for $m \geq 5$ and $c < m/4$} \\
	 &\leq \frac{1}{m^2} + O(1) \cdot \paren{\frac{3c}{m}}^2 \tag{as $c \geq 4\log{m}$}  \\
	 &\leq O(1) \cdot \paren{\frac{c}{m}}^2. 
	\end{align*}
	   This concludes the proof.
    \end{proof}

\subsection{Putting Everything Together: Proof of~\Cref{thm:OME}}\label{sec:proof-thm-OME}

We now put all these last four steps together and prove~\Cref{thm:OME}. Suppose towards a contradiction that~\Cref{thm:OME} is not true. 
\begin{enumerate}
	\item By~\Cref{lem:OME-expander}, a streaming algorithm with $o(n^{\delta}/\log{n})$ space with success probability $2/3$ implies a communication protocol $\protOME$ for $\OME_{n,k,t}$ for $k=40\log{n}$ and $t=n^{1/2-\delta}$ with $o(n^{\delta}/\log{n})$ communication
	and $2/3-o(1)$ probability of success. 
	\item By~\Cref{lem:reduction-k-2}, $\protOME$ for $\OME_{n,k,t}$ for $k=40\log{n}$ and $t=n^{1/2-\delta}$ with $o(n^{\delta}/\log{n})$ communication implies a deterministic protocol $\protHLP$ for $\HLP_{n,t}$ with  the same communication 
	and $1/2+\Omega(1/\log{n})$ probability of success (take $\eps = 1/6-o(1)$ in the lemma). 
	\item By~\Cref{lem:reduction-single-edge}, $\protHLP$ for $\HLP_{n,t}$ with $o(n^{\delta}/\log{n})$ communication and $1/2+\Omega(1/\log{n})$ success probability, implies 
	a deterministic protocol $\protHLPs$ for $\HLPs_{m}$ for $m = \Theta(n/t)$ with same communication and $\adv{\protHLPs} = \Omega(1/n\log^2{n})$ (take $\eps = \Omega(1/\log{n})$ in the lemma). 
	\item By~\Cref{lem:HLP-lower}, any deterministic protocol $\protHLPs$ for $\HLPs_{m}$ of $m= \Theta(n/t) = \Theta(n^{1/2+\delta})$ with communication cost $o(n^{\delta}/\log{n})$ can only have
	\[
		\adv{\protHLPs} = O(1) \cdot \paren{\frac{o(n^{\delta}/\log{n})}{\Theta(n^{1/2+\delta})}}^{2} = o(\frac{1}{n\log^2{n}}).
	\]
	(Here, we could apply~\Cref{lem:HLP-lower} as $n^{\delta}/\log{n} = o(m)$ and $n^{\delta}/\log{n} = \omega(\log{m})$.)
\end{enumerate}
But now Lines 3 and 4 contradict each other, finalizing our proof by contradiction of~\Cref{thm:OME}.

% !TeX root = main.tex 
%!TEX root = main.tex

\newcommand{\probone}{\ensuremath{\frac{49}{50}}\xspace}
\newcommand{\probtwo}{\ensuremath{\frac{99}{100}}\xspace}

\section{A Lower Bound for Exact Hierarchical Clustering Solution}\label{subsec:lower-exact}
The previous results have established a picture for approximation of HC in the streaming model. One might also be interested in using more memory to circumvent \emph{any} approximation factor. In particular, a natural question to ask is if we can obtain the \emph{exact} HC solution if we {increase} the memory to some value $o(n^2)$, which would still be non-trivial from space complexity perspective. 
We answer the above question in the negative in this section in the following theorem.
\begin{theorem}
\label{thm:exact-HC-strong}
Any single-pass streaming algorithm that outputs the optimal \emph{value} of hierarchical clustering with probability at least $2/3$ uses $\Omega(n^2)$ memory even with unbounded computation time. 
\end{theorem}
Our lower bound effectively rules out any streaming algorithm that (asymptotically) outperforms the naive algorithm that stores every edge and solves the problem offline in exponential time. This further justifies the `fitness' of our semi-streaming algorithm in \Cref{sec:upper-bound}. Note that similar to the previous lower bound, the lower bound in \Cref{thm:exact-HC-strong} is stronger than a standard lower bound for algorithms that output the \emph{hierarchical clustering tree}: for any single-pass streaming algorithm, our lower bound states that a memory of $\Omega(n^2)$ is necessary even to get the exact \emph{value}. 

\paragraph{High-level overview of the proof for \Cref{thm:exact-HC-strong}.} The lower bound follows from a reduction from the following variant of the well-known Index communication problem: let Alice's input be a random bipartite graph $G=(V,E)$ with
each edge appearing with probability half, and let Bob's input be a vertex pair $(i,j)$. Alice sends a message to Bob, and Bob is required to output whether $(i,j)\in E$. This problem is equivalent to the Index problem on a universe of size ${{n}\choose{2}}$
and thus requires $\Omega(n^2)$ communication~\cite{Ablayev93}.

We then reduce the problem to hierarchical clustering, which is the main technical step in the proof of \Cref{thm:exact-HC-strong}. We provide a new construction that reduces the existence of edge $(i,j)$ to the exact optimal HC cost by adding edges on Bob's side. In particular, for all vertices \emph{except} $i$ on the left partition, Bob connects them with a large clique; similarly, for all vertices \emph{except} $j$ in the right partition, Bob connects them with another large clique. Finally, Bob connects $i$ and $j$ respectively with a large clique, and he ensures the sizes of the four cliques are equal (see \Cref{fig:hc-reduction-4-clique}). Ideally, if we can \emph{control} the split pattern of the graphs constructed by the two players in the optimal HC tree, 
we can get that the optimal cost differ slightly based on the existence of edge $(i,j)$. As such, Bob can use the exact optimal cost as a signal to distinguish the corresponding Index problem.

What remains is to understand the pattern of splits for a graph prescribed as above. To this end, we show a structural lemma that characterizes optimal HC trees on such graphs: we prove that the optimal tree always \emph{first separates} the desired vertices pair into different components, and then split the rest of the graph in a \emph{fixed order}. This is a generalization of the previous work of~\cite{Dasgupta16} on computing the 
optimal HC trees of simpler graphs such as cliques and cycles.

En route to the proof of our main lower bound, we establish a weaker structural result that controls the pattern of split among \emph{two} sparsely-connected cliques. This weaker result is necessary for the main proof of \Cref{thm:exact-HC-strong}. We also note that the `two-clique' version of the structural result already gives us a (weaker) $\Omega(n^2)$ lower bound for streaming algorithms that output the hierarchical clustering tree \emph{with} the split costs.

\subsection{Warm-up: A Lower Bound for Outputting the Optimal HC Tree}
\label{subsec:warmup-exact-lb}

We first show a weaker lower bound for optimal HC algorithms that output the \emph{clustering and the split costs}. More concretely, we give the following lemma.
\begin{lemma}[Exact Hierarchical Clustering Lower Bound -- Weak version]
\label{lem:exact-HC-weak}
Any single-pass streaming algorithm that outputs the optimal hierarchical clustering together with the cost of splitting at each node with probability at least $5/6$ requires a memory of $\Omega(n^2)$ bits.
\end{lemma}

To this end, we adopt the following one-way communication game as the machinery.
\begin{problem}
\label{pbm:rand-graph-cut-com}
Suppose we give Alice a random graph $G=(V,E)$ such that there is an edge between each pair of vertices with probability half. 
Furthermore, we give Bob a partition of $V=(S \cup \bar{S})$. Alice sends a single message to Bob, and Bob outputs the \emph{exact} cut value of $\delta(S, \bar{S})$ in the end.
\end{problem}
We lower bound the communication complexity of \Cref{pbm:rand-graph-cut-com} in the following. 
\begin{lemma}[Communication Complexity of \Cref{pbm:rand-graph-cut-com}]
\label{lem:rand-graph-cut-com}
Any algorithm that solves \Cref{pbm:rand-graph-cut-com} with probability at least $4/5$ requires $\Omega(n^{2})$ communication. 
\end{lemma}

We prove the lower bound via reduction from Index.  As a reminder (and for completeness), the definition of Index is as follows.
\begin{problem}
\label{pbm:index}
Alice is given a random $N$-bit string $x \in \set{0,1}^{N}$ and Bob is given a random index $i \in [N]$. Alice sends a single message to Bob and Bob outputs $x_i$. 
\end{problem}
It is well-known that any communication protocol with success probability $1/2+\Omega(1)$ for Index requires $\Omega(N)$ communication~\cite{Ablayev93}. 
We can now prove~\Cref{lem:rand-graph-cut-com}.
\begin{proof}[Proof of~\Cref{lem:rand-graph-cut-com}]
Let $\ALG$ be a protocol for~\Cref{pbm:rand-graph-cut-com} that uses $o(n^2)$ communication and suppose towards a contradiction that it solves the problem with probability at least $4/5$. 
 
The reduction goes as follows. Given an index $x \in \set{0,1}^N$ for $N={{n}\choose{2}}$, Alice creates a random graph $G=(V,E)$ on $n$ vertices such that there is an edge between $(u,v)$ iff $x_{uv} = 1$. 
By the distribution of $x$, $G$ is also a random graph as desired in~\Cref{pbm:rand-graph-cut-com}. Alice then runs $\ALG$ on $G$ and sends its message, together with degrees of all vertices, to Bob. 
This requires $o(n^2) + O(n\log{n}) = o(n^2)$ communication. 

Bob let $(u,v) \in {{V}\choose{2}} $ be the vertex pair corresponding to index $i \in [N]$ of his input in Index problem. Bob considers the following two cuts: the cut $S_1 = \set{u}$ and the cut $S_2 = \set{u,v}$ and 
run $\ALG$ for both these cuts separately. The choice of these cuts implies that 
\begin{itemize}
\item If an edge $(u,v)$ exists, then $\delta(S_{2}, V\setminus S_{2}) - \delta(S_{1}, V\setminus S_{1})=\deg(v)-2$;
\item If for $(u,v)$ there is not an edge, then $\delta(S_{2}, V\setminus S_{2}) - \delta(S_{1}, V\setminus S_{1})=\deg(v)$.
\end{itemize}

By the guarantee of $\ALG$, the probability that Bob finds the right answer to both cuts is $\geq 1-1/5-1/5 =3/5$. Moreover, Bob can know $\deg{v}$ exactly as it is sent by Alice separately. 
Thus, with probability at least $3/5$, Bob can determine whether or not the edge $(u,v) \in E$ which is equivalent to checking if $x_i = 1$ or not, i.e., solve Index. Given the lower bound for Index, we obtain a contradiction, concluding
the proof of the lemma. 
\end{proof}

\subsubsection*{Proof of \Cref{lem:exact-HC-weak}}
We now establish the lower bound in \Cref{lem:exact-HC-weak}. To this end, we design the following reduction that forms a communication protocol for \Cref{pbm:rand-graph-cut-com}, conditioning on a streaming algorithm \texttt{ALG-HC} for HC that outputs the desired information as prescribed in \Cref{lem:exact-HC-weak} with a memory of $o(n^2)$ bits and a success probability of at least $\frac{99}{100}$. 

The reduction goes as follows. We first create $n$ \emph{additional} vertices, and send them to both Alice and Bob. Alice runs \texttt{ALG-HC} with her input graph (with her random edges, the original vertices and the isolated additional vertices), and send the memory of the algorithm to Bob. Bob will perform the following operations from his end: Bob assigns the additional vertices to $S$ and $\bar{S}$ to make the augmented partition balanced (denote them as $S'$ and $\bar{S}'$). Furthermore, Bob adds edges between the vertex pairs inside $S'$ and $\bar{S}'$ to create two complete graphs on his side. Then, Bob receives the message from Alice, and Bob runs \texttt{ALG-HC} from Alice's memory and the input he creates. Finally, Bob examines the cost of the first split (denote it as $C$), and return $C/2n$ as the value of $\delta(S, \bar{S})$.

To prove the above protocol solves \Cref{pbm:rand-graph-cut-com} with probability at least $\frac{49}{50}$, we only need to show that with high probability, the optimal hierarchical clustering tree will split $S$ and $\bar{S}$. As the first step, we bound the number of edges between $S$ and $\bar{S}$ (and resp. $S'$ and $\bar{S}'$):
\begin{claim}
\label{clm:edges-ub-S-bar-S}
In the graph jointly created by Alice and Bob, the number of edges between $S$ and $\bar{S}$ (and resp. $S'$ and $\bar{S}'$) is at most $\frac{n^2}{4}$ with probability $1-o(1)$.
\end{claim}
\begin{proof}
Let $X$ be the random variable that denotes the number of edges between $S$ and $\bar{S}$. Note that $X$ is only affected by the randomness of Alice's edges and the partition of Bob (and \emph{not} affected by the edges Bob adds). Therefore, we have
\begin{align*}
\expect{X} \leq \frac{n^2}{4}\cdot \frac{1}{2}\leq \frac{n^2}{8}.
\end{align*}
Furthermore, $X$ is a sum of independent indicator random variables. Hence, by  Chernoff bound, we have
\begin{align*}
\Pr\paren{X\geq \frac{n^2}{4}} &= \Pr\paren{X\geq (1+1)\cdot \expect{X}} \leq \exp\paren{-\frac{1/8\cdot \expect{X}}{2+1}} \leq \exp\paren{-\frac{1/8\cdot n}{3}} = o(1),
\end{align*}
as $\expect{X} \geq n$.
\end{proof}
We now show that conditioning on the event of \Cref{clm:edges-ub-S-bar-S}, the optimal hierarchical clustering tree always first split the edges between $S$ and $\bar{S}$. To this end, we show the following proposition:

\begin{proposition}[Sparsity Split Lemma -- Weak Version]
\label{prop:sparse-split-weak}
Suppose a graph $G$ has 2 cliques $S$ and $\bar{S}$ such that $\card{S}+\card{\bar{S}}=n$, and suppose the number of edges between $S$ and $\bar{S}$ (denote as $E(S, \bar{S})$) is at most $\frac{\card{S}\cdot \card{\bar{S}}}{2}$. Then, the optimal hierarchical clustering tree always first split $S$ and $\bar{S}$.
\end{proposition}
\begin{proof}
We prove this by induction. As the base case, suppose when $n=3$, and $\card{E(S,\bar{S})}=1$. Then, to first split the $E(S,\bar{S})$ edge is optimal. Now suppose for $n<n'$ this holds. For the graph with $n'$ vertices, if we first split $S$ and $\bar{S}$, the cost is at most
\begin{align*}
\mathcal{C}_{1} = \card{E(S, \bar{S})}\cdot n + 2\cdot \frac{1}{3} \cdot (\frac{n^3}{8}-\frac{n}{2}).
\end{align*}
On the other hand, suppose the optimal clustering starts with splitting some other vertices, one can denote the components after the first split as follows:
\begin{itemize}
\item $S_{a} \cup \bar{S}_{a}$: let $E_{a}$ be the set of edges edges between them.
\item $S_{b} \cup \bar{S}_{b}$: let $E_{b}$ be the set of edges edges between them.
\item The set of edges $E_{c}$ between 1). $S_{a}$ and $\bar{S}_{b}$ and 2). $S_{b}$ and $\bar{S}_{a}$.
\end{itemize}

Note that with the induction hypothesis, the optimal clustering tree will split $S_{a} \cup \bar{S}_{a}$ and $S_{b} \cup \bar{S}_{b}$ in the way that $E_{a}$ and $E_{b}$ are cut first. Therefore, the cost induced by \emph{not} splitting $\card{E(S, \bar{S})}$ is
\begin{align*}
\mathcal{C}_{2} = \paren{\card{S_{a}}\cdot \card{S_{b}} + \card{\bar{S}_{a}}\cdot \card{\bar{S}_{b}} + \card{E_{c}}} \cdot n + (\card{S_{a}}+\card{\bar{S}_{a}})\cdot \card{E_{a}} + (\card{S_{b}}+\card{\bar{S}_{b}})\cdot \card{E_{b}} \\
+ \frac{1}{3}(\card{S_{a}}^3 - \card{S_{a}} + \card{\bar{S_{a}}}^3 - \card{\bar{S_{a}}} + \card{S_{b}}^3 - \card{S_{b}} + \card{\bar{S_{b}}}^3 - \card{\bar{S_{b}}}),
\end{align*}
such that $\card{S_{a}}+\card{S_{b}}=\frac{n}{2}$, $\card{\bar{S}_{a}}+\card{\bar{S}_{b}}=\frac{n}{2}$, and $\card{E_{a}} + \card{E_{b}} + \card{E_{c}} = \card{E(S, \bar{S})} \leq \frac{n^2}{4}$. 
By merging and canceling out different terms, we can show that
\begin{align*}
\mathcal{C}_{1} - \mathcal{C}_{2} = \card{S_{a}}^2\cdot \card{S_{b}} &+ \card{S_{a}}\cdot \card{S_{b}}^2 +  \card{\bar{S}_{a}}^2\cdot \card{\bar{S}_{b}} + \card{\bar{S}_{a}}\cdot \card{\bar{S}_{b}}^2 - (\card{S_{a}} + \card{\bar{S}_{a}})\cdot \card{E_{a}} - (\card{S_{b}} + \card{\bar{S}_{b}})\cdot \card{E_{b}}\\
&+ \card{E(S, \bar{S})}\cdot n - (\card{S_{a}}\cdot \card{S_{b}} + \card{\bar{S}_{a}}\cdot \card{\bar{S}_{b}}+\card{E_{c}})\cdot n.
\end{align*}
By switching terms in the above inequality, we note that to show $\mathcal{C}_{1} - \mathcal{C}_{2} \leq 0$, it suffices to show 
\begin{align*}
(n-\card{S_{a}}-\card{\bar{S}_{a}})\cdot \card{E_{a}} + (n-\card{S_{b}}-\card{\bar{S}_{b}})\cdot \card{E_{b}} \leq \frac{n}{2}\cdot \paren{\card{S_{a}}\cdot \card{S_{b}}+\card{\bar{S}_{a}}\cdot \card{\bar{S}_{b}}}.
\end{align*}
We show the above inequality is indeed true. Note that by our constraints, there is $\card{E_{a}}\leq \card{S_{a}}\cdot \card{\bar{S}_{a}}$ and $\card{E_{b}}\leq \card{S_{b}}\cdot \card{\bar{S}_{b}}$. Hence, we have
\begin{align*}
(n-\card{S_{a}}-\card{\bar{S}_{a}})\cdot \card{E_{a}} + (n-\card{S_{b}}-\card{\bar{S}_{b}})\cdot \card{E_{b}} &\leq n\cdot \card{E_{a}} -(\frac{n}{2}-\card{S_{b}})\cdot \card{S_{a}}\cdot \card{\bar{S}_{a}}\\
& -(\frac{n}{2}-\card{\bar{S}_{b}})\cdot \card{S_{a}}\cdot \card{\bar{S}_{a}} + n \card{S_{b}}\cdot \card{\bar{S}_{b}} \\
& - (\frac{n}{2}-\card{S_{a}})\cdot \card{S_{b}}\cdot \card{\bar{S}_{b}} - (\frac{n}{2}-\card{\bar{S}_{a}})\cdot \card{S_{b}}\cdot \card{\bar{S}_{b}}\\
&= \card{S_{a}}\cdot \card{S_{b}}(\card{\bar{S}_{a}} + \card{\bar{S}_{b}}) + \card{\bar{S}_{a}} \cdot \card{\bar{S}_{b}}(\card{S_{a}}+ \card{S_{b}})\\
&= \frac{n}{2}\cdot\paren{\card{S_{a}}\cdot \card{S_{b}} + \card{\bar{S}_{a}} \cdot \card{\bar{S}_{b}}}.
\end{align*}
That is to say, for graphs in the form as prescribed in \Cref{prop:sparse-split-weak}, the strategy to first split $S$ and $\bar{S}$ results in the minimum cost. Therefore, the optimal HC tree must split $S$ and $\bar{S}$ first. 
\end{proof}

We can now finalize the proof of \Cref{lem:exact-HC-weak}. By \Cref{clm:edges-ub-S-bar-S}, with probability $1-o(1)$, the number of edges between $S'$ and $\bar{S}'$ created by Alice and Bob satisfies the condition as in \Cref{prop:sparse-split-weak}. Moreover, although the joint graph has some multi-edges on the top of the complete graph, it does not change the order of split. Therefore, we can apply \Cref{prop:sparse-split-weak} to argue that the edges between $S'$ and $\bar{S}'$ are those to be first split. Finally, the failure probability is bounded by a union bound over the event of \Cref{clm:edges-ub-S-bar-S} not happening and the event that the algorithm fails, which is at most $o(1) + 1/6 < 1/5$. 
The lower bound now follows from \Cref{lem:rand-graph-cut-com}. 

\subsection{A Lower Bound for Outputting the Optimal Value of HC}
\label{subsec:main-exact-lb}

We now proceed to the proof of the main result of this section. Similar to the proof of \Cref{lem:exact-HC-weak}, here we give the following communication game to reduce the hardness from.
\begin{problem}
\label{pbm:rand-graph-edge-exist}
Suppose we give Alice a random bipartite graph $G=(L \cup R,E)$ such that for every pair of vertices $(u,v) \in L \times R$, there is
\begin{equation*}
\begin{cases}
(u,v) \in E, \quad \text{w.p.} \, \frac{1}{2};\\
(u,v) \not\in E, \quad \text{w.p.} \, \frac{1}{2}.
\end{cases}
\end{equation*}
Furthermore, we give Bob an index of vertex pair $(i,j) \in L\times R$. Alice is allowed to send a message to Bob \emph{once}, and one of the two players has to output if $(i,j)$ is an edge.
\end{problem}

The rest of this section is to prove \Cref{thm:exact-HC-strong} in steps.

\subsubsection*{Step 1: Complexity of \Cref{pbm:rand-graph-edge-exist}}
Intuitively, \Cref{pbm:rand-graph-edge-exist} answers in the same way of INDEX if we treat each vertex pair as an entry in the array of INDEX. We now formalize this complexity result to show that it requires $\Omega(n^2)$ bits to solve \Cref{pbm:rand-graph-edge-exist}.

\begin{lemma}[Communication Complexity of \Cref{pbm:rand-graph-edge-exist}]
\label{lem:rand-graph-edge-com}
Any algorithm that solves \Cref{pbm:rand-graph-edge-exist} with probability at least $\probone$ requires a communication complexity of $\Omega(n^{2})$ bits. 
\end{lemma}
\begin{proof}
Again, we are going to design a reduction to use the the complexity of INDEX. Suppose we have a streaming algorithm \texttt{ALG} that solves \Cref{pbm:rand-graph-edge-exist} with probability at least $\probone$. We use this to design a protocol that solves INDEX.

The protocol is as follows. Alice constructs a random graph with $n$ vertices such that $N=\frac{n^2}{4}$, i.e. every possible vertex pair for a bipartite graph with $\card{L}=\card{R}=\frac{n}{2}$. Bob is given the same vertices, and he transform his index $i^{*}$ to the corresponding index of the vertex pair $(u,v)$. Alice runs \texttt{ALG} from her end, send the memory to Bob; Bob runs \texttt{ALG} conditioning on Alice's message, and output $0$ if $(u,v)\not\in E$, and $1$ otherwise.

It is straightforward to see that the index value exactly corresponds to the existence of the edge. Therefore, the protocol succeeds with the same probability of \texttt{ALG}. This implies any such streaming algorithm \texttt{ALG} has to use a memory of $\Omega(N)=\Omega(n^2)$ bits.
\end{proof}

\subsubsection*{Step 2: A Reduction to Hierarchical Clustering}
We now proceed to the reduction from \Cref{pbm:rand-graph-edge-exist} to hierarchical clustering. Given a streaming hierarchical clustering algorithm \texttt{ALG}, we can make a protocol for \Cref{pbm:rand-graph-edge-exist} as follows:

\clearpage

\begin{tbox}
\texttt{PORT}: a communication protocol for \Cref{pbm:rand-graph-edge-exist}.

\smallskip

\paragraph{Input:} \texttt{ALG} -- a streaming algorithm that outputs the optimal hierarchical clustering for any signed complete graph with probability at least $\frac{99}{100}$.

\smallskip

\paragraph{The Construction:} Alice and Bob construct a graph $G=(V,\, E)$ as follows.
\begin{enumerate}[label=$(\roman*)$]
    \item Both Alice and Bob are given $n=16 N$ vertices.
	\item Alice constructs the random bipartite graph on $2N$ vertices ($G'=(L\cup R, E)$, $\card{L}=\card{R}=N$) as in \Cref{pbm:rand-graph-edge-exist}; Bob holds an index $i^{*}=(u,v)\in [N^2]$.
	\item Bob adds edges to $G$ in the following manner:
	\begin{enumerate}
	    \item Bob connects $u$ with a set of $4N-1$ vertices (call them $\tilde{S}_{1}$), and make $S_{1}= \{u\}\cup \tilde{S}_{1}$ a clique.
	    \item In the same manner, Bob connects $v$ with another set of $4N-1$ vertices (call them $\tilde{S}_{2}$), and make $S_{2}= \{v\}\cup \tilde{S}_{2}$ a clique.
	    \item Bob connects every vertex in $L$ \emph{except} $u$ with a set of $3N+1$ vertices (call them $\tilde{S}_{3}$), and make $S_{3}=L\setminus \{u\}\cup \tilde{S}_{3}$.
	    \item In the same manner, Bob connects every vertex in $R$ \emph{except} $v$ with another set of $3N+1$ vertices (call them $\tilde{S}_{4}$), and make $S_{4}=R\setminus \{v\}\cup \tilde{S}_{4}$.
	\end{enumerate}
  \item We emphasize that $\tilde{S}_{1}$, $\tilde{S}_{2}$, $\tilde{S}_{3}$, and $\tilde{S}_{4}$ are disjoint, and Bob can pick the vertices based on the lexicographical orders.
\end{enumerate}

\smallskip

\paragraph{The Message and Output: }
\begin{enumerate}[label=$(\roman*)$]
	\item Alice runs \texttt{ALG}, and sends the memory of the algorithm as the message to Bob.
	\item Alice further sends the degrees of the first $2N$ vertex to Bob (with $\tilde{O}(N)$ bits). 
	\item Bob runs \texttt{ALG} based on Alice's message, and outputs based on the cost:
	\begin{enumerate}
	    \item Assume w.log. that $\deg(u)<\deg(v)$.
	    \item If the cost equals to
    	    \begin{align*}
    	    \cost = \deg(u)\cdot 16N + (\deg(v)-1)\cdot 12N + \frac{1}{2}\cdot \sum_{w \neq u,v}\deg(w)\cdot 8N +\frac{4}{3}\paren{(4N)^3 - 4N},
	    \end{align*}
	    then return $(u,v)\in E$.
	    \item Otherwise, if the cost equals to 
    	    \begin{align*}
    	    \cost = \deg(u)\cdot 16N + \deg(v)\cdot 12N + \frac{1}{2}\cdot \sum_{w \neq u,v}\deg(w)\cdot 8N +\frac{4}{3}\paren{(4N)^3 - 4N},
    	    \end{align*}
	    then return $(u,v)\not \in E$.
	    \item Otherwise, return FAIL.
	\end{enumerate}
\end{enumerate}

\end{tbox}

An illustration of the constructed graph $G$ can be found as the left plot of \Cref{fig:hc-reduction-4-clique}. It is straightforward to see that the reduction does not increase the communication complexity as long as the memory of \texttt{ALG} is $\Omega(n)$. As such, our task now is to prove that with high constant probability, the optimal cost agrees with the desired value. To this end, we introduce the following proposition which characterizes the optimal tree on a graph constructed by Alice and Bob. 

\begin{figure}[h!]
\centering
\includegraphics[scale=0.22]{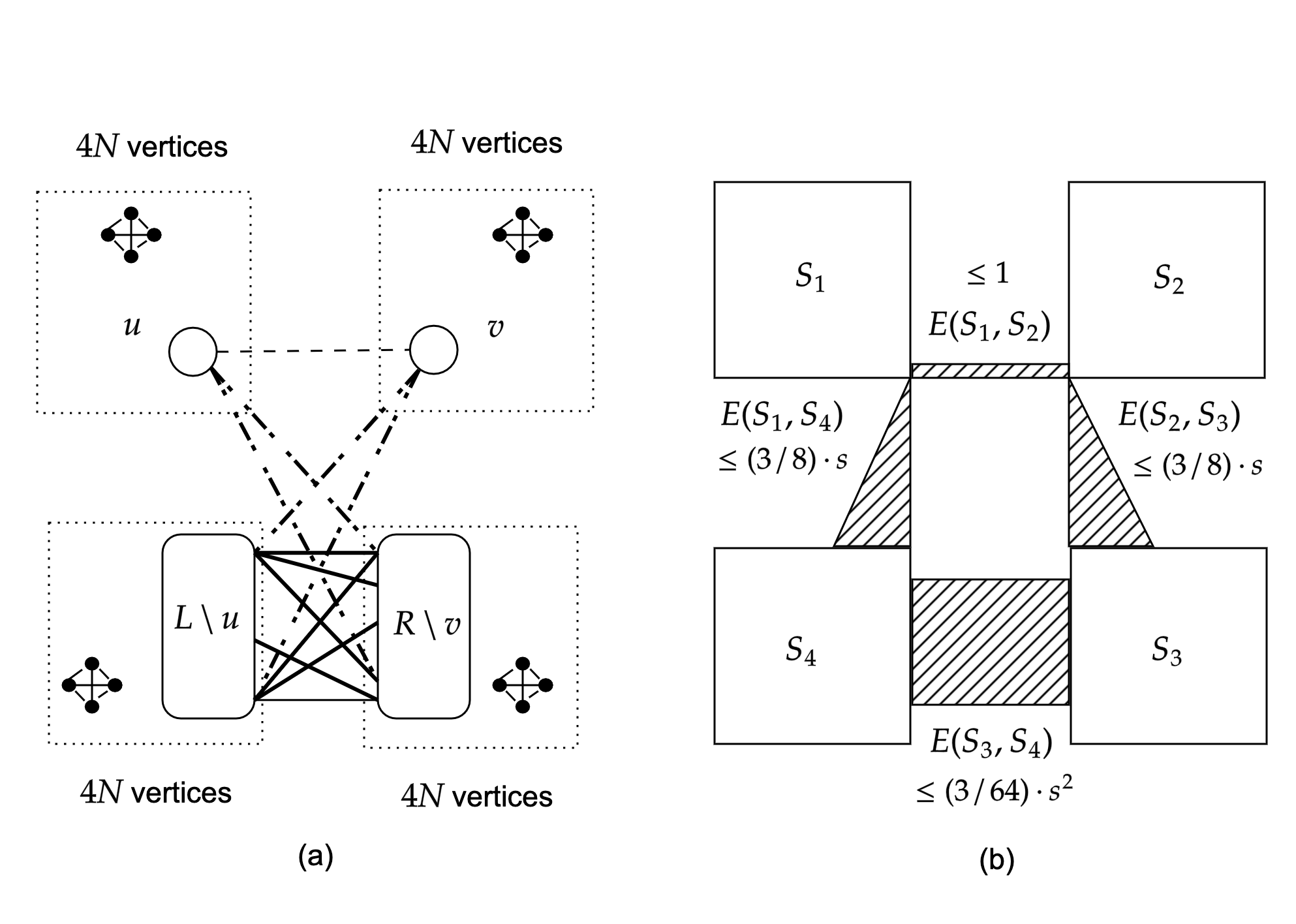}
\caption{An illustration of the graph constructed by Alice and Bob (a) and the desired graph property of \Cref{prop:sparse-split-strong} (b). In (a), the bold solid lines are the random edges added by Alice, minus the edges indent on $u$ and $v$. The bold dashed lines are the random edges indent on $u$ and $v$, minus the (possible) edge $(u,v)$ itself. And finally, the dashed line indicates the (possible) edge between $u$ and $v$. Bob puts the four parts into four components with $4N$ vertices each, and add edges in the components to make them cliques. In (b), the edges of a graph are denoted as the shaded areas. With high constant probability (over the randomness of Alice), the graph $G$ constructed in $(a)$ satisfies the conditions illustrated in $(b)$. }\label{fig:hc-reduction-4-clique}
\end{figure}

\begin{proposition}[Sparsity Split Lemma -- Strong Version]
\label{prop:sparse-split-strong}
Suppose a graph $G$ has 4 cliques $S_{1}$, $S_{2}$, $S_{3}$ and $S_{4}$ such that $\card{S_{1}}=\card{S_{2}}=\card{S_{3}}=\card{S_{4}}=s$, and suppose there are only 4 set of edges between them: $E(S_{1}, S_{2})$, $E(S_{2}, S_{3})$, $E(S_{3}, S_{4})$ and $E(S_{1}, S_{4})$. Furthermore, assume w.log. that $\card{E(S_{1}, S_{2})}+\card{E(S_{1}, S_{4})}\leq \card{E(S_{2}, S_{3})}\leq \card{E(S_{3}, S_{4})}$, the edges of $G$ be with the following properties:
\begin{enumerate}[label=\alph*).]
\item~\label{line:sparse-u-v-edge}$\card{E(S_{1}, S_{2})}\leq 1$, $1\leq\card{E(S_{1}, S_{4})}\leq \card{E(S_{2}, S_{3})}\leq \frac{3}{8}\cdot s$.
\item~\label{line:single-v-out}Among all vertices $v \in S_{1}$ (resp. $v\in S_{2}$), only a \emph{single} vertex $v$ has neighbors $u \not\in S_{1}$ (resp. neighbors $u \not\in S_{2}$).
\item~\label{line:sparse-s3-s4}$\frac{1}{64}\cdot s^2 \leq \card{E(S_{3}, S_{4})}\leq \frac{3}{64}\cdot s^2$. Furthermore, for any $S'_{3}\subseteq S_{3}$ and $S'_{4}\subseteq S_{4}$, $\card{E(S'_{3}, S'_{4})}\leq \card{S'_{3}}\cdot \card{S'_{4}}$.
\end{enumerate}
An illustration of such a graph $G$ can be found in the right column of \Cref{fig:hc-reduction-4-clique}. Then, the optimal hierarchical clustering tree on $G$ follows the below pattern:
\begin{enumerate}
\item The first split separates $S_{1}$ from the rest of the graph.
\item The second split separates $S_{2}$ from the rest of the graph.
\item The third split separates $S_{3}$ and $S_{4}$.
\item Each clique is clustered by the induced hierarchical clustering tree after it is separated from the rest of the graph.
\end{enumerate}
In other words, the hierarchical clustering tree is as \Cref{fig:hc-tree-4-cliques}.
\end{proposition}

\begin{figure}[h!]
\centering
\includegraphics[scale=0.22]{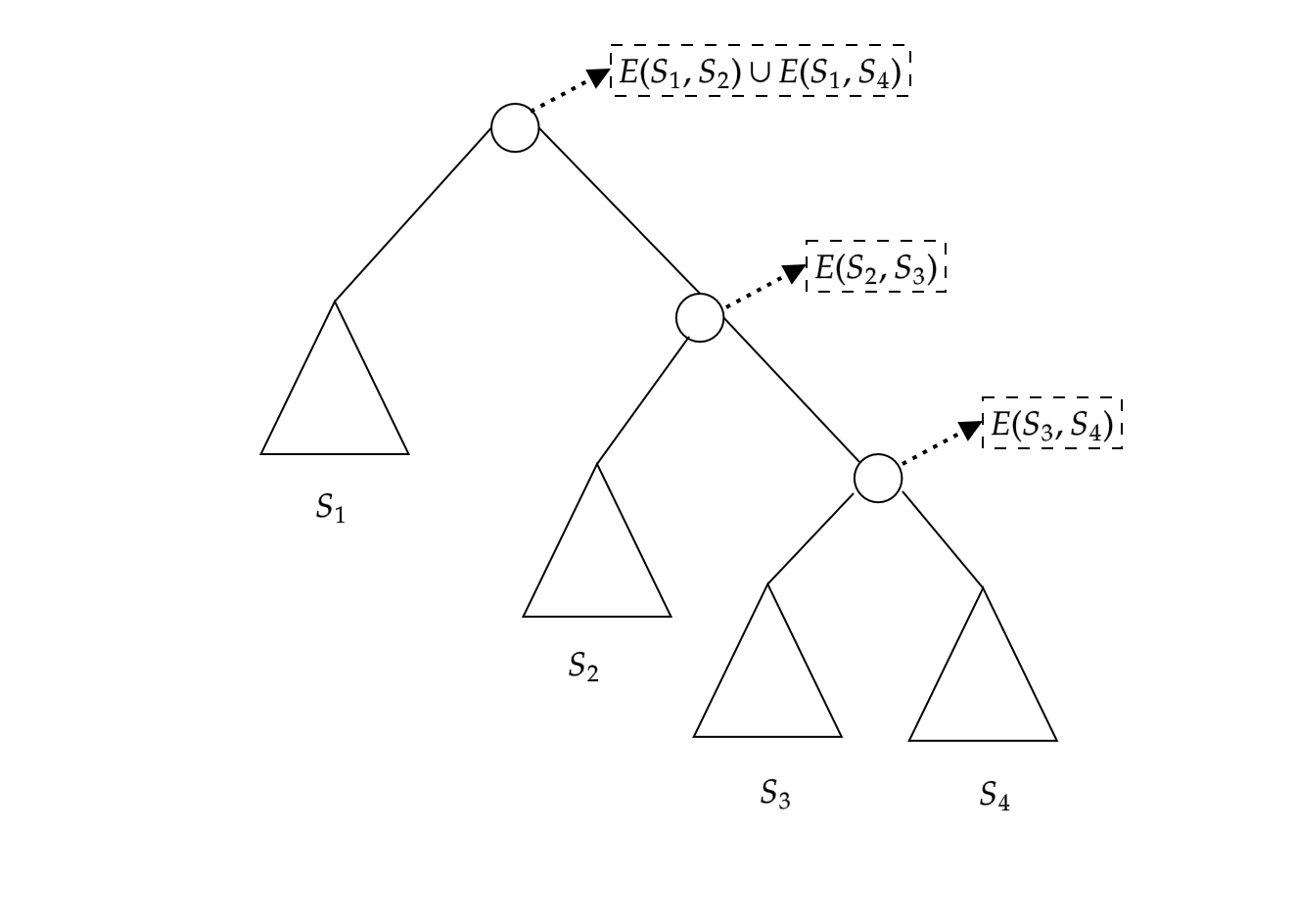}
\caption{An illustration of the hierarchical clustering tree as described in \Cref{prop:sparse-split-strong}.}\label{fig:hc-tree-4-cliques}
\end{figure}

We defer the proof of \Cref{prop:sparse-split-strong} to the next step. Conditioning on the statement of \Cref{prop:sparse-split-strong}, we can show the correctness of \texttt{PROT}. We first show that with high probability, the graph $G$ constructed by Alice and Bob satisfies the conditions required by \Cref{prop:sparse-split-strong}. More formally, we have
\begin{claim}
\label{clm:random-bipartite-conditions}
With probability at least $\frac{99}{100}$, the graph $G$ created by Alice and Bob satisfies the conditions prescribed in \Cref{prop:sparse-split-strong} for sufficiently large $N$.
\end{claim}
\begin{proof}
We first verify property~\ref{line:single-v-out}. Note that by our construction, all the edges that go outside $S_{1}$ are indent on $u$, and all edges that go outside $S_{2}$ are indent on $v$. Therefore, property~\ref{line:single-v-out} holds deterministically. 

We now turn to properties~\ref{line:sparse-u-v-edge} and \ref{line:sparse-s3-s4}. For property~\ref{line:sparse-u-v-edge}, note that $E(S_{1}, S_{2})$ can only contain the (possible) edge $(u,v)$, which means $\card{E(S_{1}, S_{2})}\leq 1$ always holds. Let $s=4N$ as constructed by Alice and Bob. For both $E(S_{1}, S_{4})$ and $E(S_{2}, S_{3})$, the expectation of there size is
\begin{align*}
\Exp{\card{E(S_{1}, S_{4})}} = \Exp{\card{E(S_{2}, S_{3})}} = \frac{1}{2}\cdot (\frac{s}{4}-1) = \frac{s}{8}-\frac{1}{2}.
\end{align*}
Therefore, by a Chernoff bound argument, we can show that with probability at least $1-2^{-O(s)}$, there are $\card{E(S_{1}, S_{4})}\leq \frac{3}{8}\cdot s$ and $\card{E(S_{2}, S_{3})}\leq \frac{3}{8}\cdot s$. With a sufficiently large $s$, this probability is at least $\frac{199}{200}$.

Finally, for property~\ref{line:sparse-s3-s4}, note that the expected number of edges between $S_{3}$ and $S_{4}$ is
\begin{align*}
\expect{\card{E(S_{3}, S_{4})}} = \frac{1}{2}\cdot (\frac{s}{4}-1)^2 = \frac{s^2}{32}-\frac{s}{4}+\frac{1}{2}.
\end{align*}
Hence, by a Chernoff bound argument, with probability at least $1-2^{-O(s)}$, there is $\frac{1}{64}\cdot s^2\leq \card{E(S_{3}, S_{4})}\leq \frac{3}{64}\cdot s^2$. The probability is at least $\frac{199}{200}$ for sufficiently large $s$. Furthermore, since the graph is simple, the second statement of property~\ref{line:sparse-s3-s4} trivially holds.

A union bound over the failure probability of the above events gives us the desired conclusion.
\end{proof}

With \Cref{prop:sparse-split-strong} and \Cref{clm:random-bipartite-conditions}, we can establish the correctness of \texttt{PROT} as follows.
\begin{lemma}
\label{lem:exact-cost-reduction}
\texttt{PROT} solves \Cref{pbm:rand-graph-edge-exist} correctly with probability at least $\frac{49}{50}$.
\end{lemma}
\begin{proof}
Conditioning on the event of \Cref{clm:random-bipartite-conditions}, the optimal hierarchical clustering tree for $G$ follows the pattern prescribed by \Cref{prop:sparse-split-strong}. Therefore, if $(u,v)\in E$, which means $E(S_{1}, S_{2})=1$, the optimal tree will first split all the edges on $u$ (which cost $\deg(u)$ edges), then split the remaining edges on $v$ (which cost $\deg(v)-1$ edges). On the other hand, if $(u,v)\not\in E$, which means $E(S_{1}, S_{2})=0$, the optimal tree will first split all the edges on $u$ (which cost $\deg(u)$ edges), then split the all the edges on $v$ (which cost $\deg(v)$ edges). The remaining part of the splits are the same, and it always confirms to the value described in the reduction. Hence, conditioning on the success of \texttt{ALG}, \texttt{PROT} can correctly distinguish if the edge $(u,v)$ exist.

The failure probability for \texttt{PROT} is at most the union bound over the failure probability of the event of \Cref{clm:random-bipartite-conditions} and the failure probability of \texttt{ALG}, which is at most $\frac{1}{50}$.
\end{proof}

\begin{proof}[Proof of \Cref{thm:exact-HC-strong}]
Since \texttt{PROT} solves \Cref{pbm:rand-graph-edge-exist} with probability at least $\frac{49}{50}$, by \Cref{lem:rand-graph-edge-com}, \texttt{PROT} must send $\Omega(N^2)$ bits. Furthermore, by the reduction, we observe that $n=16N$. Hence, the message of Alice must be of size at least $\Omega((n/16)^2)=\Omega(n^2)$, which implies the memory of such \texttt{ALG} has to be $\Omega(n^2)$.
\end{proof}

The rest of our task is to prove \Cref{prop:sparse-split-strong}.

\subsubsection*{Step 3: Proof of \Cref{prop:sparse-split-strong}}
The proof of \Cref{prop:sparse-split-strong} shares a similar idea to the proof of \Cref{prop:sparse-split-weak}, albeit the process becomes much more involved. On the high-level, we prove this result in the following steps:
\begin{itemize}
\item We first show that conditioning the optimal clustering tree first split the edge \emph{between} cliques, then the optimal strategy is to follow the splits in \Cref{prop:sparse-split-strong}. This reduces our task to proving the optimal tree always first splits the edges between the cliques.
\item The desired statement now is very similar to \Cref{prop:sparse-split-weak}. However, we need more care in this proof: since a clustering tree splits multiple cliques \emph{in order}, directly applying the inductive proof as in \Cref{prop:sparse-split-weak} will create too many cases to handle. Therefore, we instead establish our argument in two steps. The first step is to show that conditioning on the inductive hypothesis, a optimal hierarchical clustering tree will \emph{not} start the edge cuts from $S_{1}$ or $S_{2}$. This step relies on the fact that there is only a single vertex inside $S_{1}$ and $S_{2}$ that connects to vertices outside, and any optimal tree that first splits the clique edges inside has to \emph{entirely} cut the clique edges. We show that this leads to sub-optimal costs.
\item The only concern now is the hierarchical clustering tree may start split from the edges inside $S_{3}$ or $S_{4}$. Since the pattern of split is now controlled, we can employ an inductive argument similar to the proof of \Cref{prop:sparse-split-weak}, and show that the graph will \emph{not} start split from edges inside $S_{1}$ and $S_{2}$.
\item Finally, we still have to control the behavior of the hierarchical clustering tree \emph{after} the first cut. This part is straightforward: for the subgraph with $3$ cliques, we can repeat the argument for $4$ cliques. And after we get the subgraph of $2$ cliques, we can employ \Cref{prop:sparse-split-weak} to get the desired split pattern.
\end{itemize}

We now formalize the above intuitions. We start with introducing the lemma that controls the behavior of the clustering tree if we restrict the cut to first split edges between cliques. 
\begin{lemma}
\label{lem:cut-sparsest-connect}
Let $G$ be a graph as prescribed in \Cref{prop:sparse-split-strong}. For any clustering tree $\TT$, if its first $2$ cuts are restricted to the edges among $E(S_{1}, S_{2})\cup E(S_{2}, S_{3}) \cup E(S_{1}, S_{4})\cup E(S_{3}, S_{4})$, then the optimal cost is induced by the following order: first cut $E(S_{1},S_{4})\cup E(S_{1}, S_{2})$, then cut $E(S_{2}, S_{3})$. 
\end{lemma}
\begin{proof}
Note that if the first two cuts of $\TT$ is restricted to edges among  $E(S_{1}, S_{2})\cup E(S_{2}, S_{3}) \cup E(S_{1}, S_{4})\cup E(S_{3}, S_{4})$, then the third cut should also be among $E(S_{1}, S_{2})\cup E(S_{2}, S_{3}) \cup E(S_{1}, S_{4})\cup E(S_{3}, S_{4})$ by \Cref{prop:sparse-split-weak}. Therefore, the cost of such hierarchical clustering trees is can be characterized as
\begin{align*}
\mathcal{C} = \card{E_{1}}\cdot 4s + \card{E_{2}}\cdot 3s +\card{E_{3}}\cdot 2s + \frac{4}{3}\cdot (s^3-s),
\end{align*}
where $E_{1}$, $E_{2}$ and $E_{3}$ are the set of edges to be split in the first, the second, and the third cuts that separates the graph into disconnected components.  As a result, it is easy to observe that the minimizer of the cost is attained by always splitting the edges with smaller weights. With the graph described in \Cref{prop:sparse-split-strong}, it means to first split $E(S_{1},S_{4})\cup E(S_{1}, S_{2})$, then split $(S_{2}, S_{3})$.
\end{proof}

We now proceed to the next step, which aims to show that the cuts never start with a cut that splits the clique edges inside $S_{1}$ or $S_{2}$. More formally, we have

\begin{lemma}
\label{lem:s1-s2-improp}
Let $G$ be a graph as prescribed in \Cref{prop:sparse-split-strong}, and $\TT$ be the optimal clustering tree of $G$. Suppose for such a graph $G$ with sizes less than $4s$ (the sizes of $S_{i}$'s are not necessarily equal), the optimal clustering tree always restrict the first 2 cuts among $E(S_{1}, S_{2})\cup E(S_{2}, S_{3})\cup E(S_{1}, S_{4})\cup E(S_{3}, S_{4})$. Then, for the first cut $G\rightarrow(A,B)$ of $\TT$, there is either $S_{1}\subseteq A$ or $S_{1}\subseteq B$. The same statement holds for $S_{2}$. 

Furthermore, if we remove $S_{1}$ and all edges indent on $S_{1}$ and obtain an induced subgraph $G'$, and let $\TT'$ be the optimal tree on $G'$. Then, for the first cut $G'\rightarrow(A',B')$ of $\TT'$, there is either $S_{2}\subseteq A'$ or $S_{2}\subseteq B'$.
\end{lemma}

\begin{proof}
We first observe that there is only one vertex with non-clique edges in $S_{1}$ (resp. $S_{2}$); therefore, if an optimal tree starts the first split that involves clique edges in $S_{1}$ and $S_{2}$, it must be \emph{entirely} inside the clique, as the optimal tree never splits the graph into more than two disconnected components. The same holds for $S_{2}$. 

We now show that restricting the first cut to clique edges inside $S_{1}$ or $S_{2}$ is sub-optimal. To see this, let a tree $\TT'$ be a tree that first splits clique edges in $S_{1}$ to induce $S_{1}\rightarrow (S^{in}_{1}, S^{out}_{1})$. Based on the assumption, the optimal tree of the subgraph $G\setminus S^{out}_1$ will restrict its first two cuts among edges of $E(S^{in}_{1}, S_{2})\cup E(S_{2}, S_{3})\cup E(S^{in}_{1}, S_{4})\cup E(S_{3}, S_{4})$. As such, comparing with an optimal tree $\TT$ that restrict its first two splits among the edges of $E(S_{1}, S_{2})\cup E(S_{2}, S_{3})\cup E(S_{1}, S_{4})\cup E(S_{3}, S_{4})$, the cost of $\TT'$ has the following changes:
\begin{itemize}
\item An extra cost of $\card{S^{in}_{1}}\cdot \card{S^{out}_{1}}\cdot 4s$ for the first cut.
\item A decreased cost of at most $3s$ multiplicative factor for each of the edges in $E(S_{1}, S_{2})\cup E(S_{2}, S_{3})\cup E(S_{1}, S_{4})$. Hence, this part of decreased cost is at most 
\begin{align*}
\card{E(S_{1}, S_{2})\cup E(S_{2}, S_{3})\cup E(S_{1}, S_{4})}\cdot 3s\leq \frac{9}{4}\cdot s^2.
\end{align*}
\item A decreased cost between splitting $S_{1}$ (where the cost is $\frac{1}{3}\cdot (s^3-s)$) and the cost of splitting $S^{in}_{1}$ and $S^{out}_{1}$ separately. This part is at most $\card{S^{in}_{1}}^2\cdot \card{S^{out}_{1}} + \card{S^{in}_{1}}\cdot \card{S^{out}_{1}}^2 =  \card{S^{in}_{1}}\cdot \card{S^{out}_{1}}\cdot s$.
\end{itemize}
As such, the gap between the costs of $\TT'$ and $\TT$ is at least 
\begin{align*}
\cost(\TT') - \cost(\TT) \geq \card{S^{in}_{1}}\cdot \card{S^{out}_{1}}\cdot 4s - \frac{9}{4}\cdot s^2 - \card{S^{in}_{1}}\cdot \card{S^{out}_{1}}\cdot s >0. \tag{$\card{S^{in}_{1}}\cdot \card{S^{out}_{1}} \geq s-1$}
\end{align*}
Therefore, such a $\TT'$ cannot be an optimal tree on $G$.
\end{proof}

We emphasize that the only condition for \Cref{lem:s1-s2-improp} to hold is the behavior of the graph in this family with size less than $4s$ , which is crucial in our inductive argument of the proof of \Cref{prop:sparse-split-strong}, established as follows.

\begin{lemma}
\label{lem:s3-s4-opt-cost}
Let $G$ be a graph as prescribed in \Cref{prop:sparse-split-strong} (the sizes of $S_{i}$'s are not necessarily equal), and let $\TT$ be the optimal tree whose first 2 cuts are restricted to the edges among $E(S_{1}, S_{2})\cup E(S_{2}, S_{3}) \cup E(S_{1}, S_{4})\cup E(S_{3}, S_{4})$, and let $\TT'$ be the optimal tree that whose first cut involved the clique edges. Then, we have
$
\cost(\TT)<\cost(\TT').
$
\end{lemma}
\begin{proof}
We first prove that a tree $\TT$ whose \emph{first} cut is restricted to $E(S_{1}, S_{2})\cup E(S_{2}, S_{3}) \cup E(S_{1}, S_{4})\cup E(S_{3}, S_{4})$ induces a smaller optimal cost than a tree $\TT'$ that first splits edges inside cliques. We can prove this by induction. For the base case, consider all $S_{i}$ to be single vertices, and there is no clique edges. As such, the statement trivially holds. 

For the induction step, suppose the statement holds on such a $G$ with size less than $4s$ (and the sizes of the cliques are not necessarily equal). By \Cref{lem:s1-s2-improp}, we know that the first cut will \emph{not} start from edges inside $S_{1}$ or $S_{2}$. Now when $\card{S_{3}}=\card{S_{4}}=s$, if a clustering tree $\TT$ do \emph{not} first split from edges inside $S_{3}$ and $S_{4}$, the optimal cost is at most
\begin{align*}
\mathcal{C}_{1} = \paren{\card{E(S_{1}, S_{2})} + \card{E(S_{1}, S_{4})}}\cdot 4s + \card{E(S_{2}, S_{3})}\cdot 3s + \card{E(S_{2}, S_{3})}\cdot 2s + \frac{4}{3}\cdot (s^3-s).
\end{align*}

On the other hand, suppose a clustering tree $\TT'$ starts with splitting edges that involve the clique edges of $S_{3}$ and $S_{4}$. One can denote the components after the first split as follows:
\begin{itemize}
\item $S_{1} \cup S_{2} \cup S^{in}_{3} \cup S^{in}_{4}$: let $E^{in}(S_{3}, S_{4})$ be the set of edges between $S^{in}_{3}$ and $S^{in}_{4}$ after the split.
\item $S^{out}_{3} \cup S^{out}_{4}$: let $E^{out}(S_{3}, S_{4})$ be the set of edges edges between $S^{out}_{3}$ and $S^{out}_{4}$.
\item The set of edges $E^{cross}:=E^{cross}(S_{3}, S_{4})\cap E^{cross}(S_{1}, S_{4})\cup E^{cross}(S_{2}, S_{3})$ that split to separate $S^{in}_{3}$ from $S^{out}_{4}$ and to separate $S^{in}_{4}$ from $S^{out}_{3}$, plus the edges among $E(S_{1}, S_{4})$ that have one vertex in $S^{out}_{4}$ and the edges among $E(S_{2}, S_{3})$ that have one vertex in $S^{out}_{3}$.
\end{itemize}

We can apply the induction hypothesis such that the subgraphs of $S_{1} \cup S_{2} \cup S^{in}_{3} \cup S^{in}_{4}$ will \emph{not} start with cutting edges inside $S^{in}_{3}$ and $S^{in}_{4}$. Also, not that the optimal tree \emph{unconditionally} will \emph{not} cut the clique edges of $S_{1}$ and $S_{2}$. 

Assuming $\card{E(S_{1}, S^{in}_{4})}+\card{E(S_{2}, S^{in}_{3})}\leq \card{E(S^{in}_{3}, S^{in}_{4})}$, the order of split on the subgraph $S_{1} \cup S_{2} \cup S^{in}_{3} \cup S^{in}_{4}$ does \emph{not} change. As such, the optimal cost induced by this strategy is

\begin{align*}
\mathcal{C}_{2}  = & \paren{\card{S^{in}_{3}}\cdot \card{S^{out}_{3}} + \card{S^{in}_{4}}\cdot \card{S^{out}_{4}} + \card{E^{cross}(S_{3}, S_{4})} +\card{E^{cross}(S_{1}, S_{4})} + \card{E^{cross}(S_{2}, S_{3})}} \cdot 4s \\
& + (\card{E(S_{1}, S_{2})}+\card{E(S_{1},S_{4})}-\card{E^{cross}(S_{1}, S_{4})}) \cdot (4s-\card{S^{out}_{3}}-\card{S^{out}_{4}})\\
& + (\card{E(S_{2},S_{3})}-\card{E^{cross}(S_{2}, S_{3})}) \cdot (3s-\card{S^{out}_{3}}-\card{S^{out}_{4}}) \\
& + \card{E^{in}(S_{3}, S_{4})}\cdot(\card{S^{in}_{3}}+\card{S^{in}_{4}}) + \card{E^{out}(S_{3}, S_{4})}\cdot(\card{S^{out}_{3}}+\card{S^{out}_{4}})\\
&+\frac{2}{3}(s^3-s) + \frac{1}{3}(\card{S^{in}_{3}}^3+\card{S^{in}_{4}}^3+\card{S^{out}_{3}}^3+\card{S^{out}_{4}}^3-2s),
\end{align*}
such that the conditions prescribed in \Cref{prop:sparse-split-strong} are satisfied. A lower bound of $\mathcal{C}_{2}$ can be obtained by ignoring the higher cost of $E^{cross}(S_{1}, S_{4})$ and $E^{cross}(S_{2}, S_{3})$:
\begin{align*}
\mathcal{C}_{2} \geq & \paren{\card{S^{in}_{3}}\cdot \card{S^{out}_{3}} + \card{S^{in}_{4}}\cdot \card{S^{out}_{4}} + \card{E^{cross}(S_{3}, S_{4})}} \cdot 4s \\
& + (\card{E(S_{1}, S_{2})}+\card{E(S_{1},S_{4})}) \cdot (4s-\card{S^{out}_{3}}-\card{S^{out}_{4}})\\
& + (\card{E(S_{2},S_{3})}) \cdot (3s-\card{S^{out}_{3}}-\card{S^{out}_{4}}) \\
& + \card{E^{in}(S_{3}, S_{4})}\cdot(\card{S^{in}_{3}}+\card{S^{in}_{4}}) + \card{E^{out}(S_{3}, S_{4})}\cdot(\card{S^{out}_{3}}+\card{S^{out}_{4}})\\
&+\frac{2}{3}(s^3-s) + \frac{1}{3}(\card{S^{in}_{3}}^3+\card{S^{in}_{4}}^3+\card{S^{out}_{3}}^3+\card{S^{out}_{4}}^3-2s).
\end{align*}
Note that the above expressions are based on the assumption that $\card{E(S_{1}, S^{in}_{4})}+\card{E(S_{2}, S^{in}_{3})}\leq \card{E(S^{in}_{3}, S^{in}_{4})}$, and we now remove this assumption by using a uniform lower bound. Note that no matter how we switch the order of split, the edges $E(S_{1},S_{2})$, $E(S_{1},S_{4})$, and $E(S_{2},S_{3})$ have to pay a multiplicative factor of at least $2s$. Furthermore, the edges $E^{in}(S_{3}, S_{4})$ and $E^{out}(S_{3}, S_{4})$ have to pay the multiplicative factors of $(\card{S^{in}_{3}}+\card{S^{in}_{4}})$ and $(\card{S^{out}_{3}}+\card{S^{out}_{4}})$, respectively. Therefore, we can establish a lower bound for $\mathcal{C}_{2}$ regardless the order of split:
\begin{align*}
\mathcal{C}_{2} > & \paren{\card{S^{in}_{3}}\cdot \card{S^{out}_{3}} + \card{S^{in}_{4}}\cdot \card{S^{out}_{4}} + \card{E^{cross}(S_{3}, S_{4})}} \cdot 4s \\
& + (\card{E(S_{1}, S_{2})}+\card{E(S_{1},S_{4})}+\card{E(S_{2},S_{3})}) \cdot 2s \\
& + \card{E^{in}(S_{3}, S_{4})}\cdot(\card{S^{in}_{3}}+\card{S^{in}_{4}}) + \card{E^{out}(S_{3}, S_{4})}\cdot(\card{S^{out}_{3}}+\card{S^{out}_{4}})\\
&+\frac{2}{3}(s^3-s) + \frac{1}{3}(\card{S^{in}_{3}}^3+\card{S^{in}_{4}}^3+\card{S^{out}_{3}}^3+\card{S^{out}_{4}}^3-2s).
\end{align*}

By merging and canceling out different terms, we can show that
\begin{align*}
\mathcal{C}_{2} - \mathcal{C}_{1} > & \paren{\card{S^{in}_{3}}\cdot \card{S^{out}_{3}} + \card{S^{in}_{4}}\cdot \card{S^{out}_{4}}}\cdot 4s + \card{E^{cross}(S_{3},S_{4})}\cdot 3s \\
& - \paren{\card{E(S_{1}, S_{2})}+\card{E(S_{1},S_{4})} +\card{E(S_{2}, S_{3})}}\cdot 2s\\
& - \card{E^{in}(S_{3}, S_{4})}\cdot (\card{S^{out}_{3}}+\card{S^{out}_{4}}) - \card{E^{out}(S_{3}, S_{4})}\cdot (\card{S^{in}_{3}}+\card{S^{in}_{4}})\\
& -(\card{S^{in}_{3}}\cdot \card{S^{out}_{3}} + \card{S^{in}_{4}}\cdot \card{S^{out}_{4}})\cdot s,
\end{align*}
where the second line comes from the gap of the costs for the edges $E(S_{1}, S_{2})$, $E(S_{1}, S_{4})$ and $E(S_{2}, S_{3})$, the third line comes from the gap for the $\card{E^{in}(S_{3}, S_{4})}$ and $\card{E^{out}(S_{3}, S_{4})}$ edges, and the final line comes from the gap between $\frac{2}{3}(s^3-s)$ and $\frac{1}{3}(\card{S^{in}_{3}}^3+\card{S^{in}_{4}}^3+\card{S^{out}_{3}}^3+\card{S^{out}_{4}}^3-2s)$. We can upper bound the absolute value of the third line by
\begin{align*}
\card{E^{in}(S_{3}, S_{4})}\cdot & (\card{S^{out}_{3}}+\card{S^{out}_{4}}) + \card{E^{out}(S_{3}, S_{4})}\cdot (\card{S^{in}_{3}}+\card{S^{in}_{4}})\\
&\leq \card{S^{in}_{3}}\card{S^{in}_{4}}\cdot (\card{S^{out}_{3}}+\card{S^{out}_{4}}) + \card{S^{out}_{3}}\card{S^{out}_{4}}\cdot (\card{S^{in}_{3}}+\card{S^{in}_{4}}) \tag{since $\card{E^{in}(S_{3}, S_{4})}\leq \card{S^{in}_{3}}\card{S^{in}_{4}}$ and $\card{E^{out}(S_{3}, S_{4})}\leq\card{S^{out}_{3}}\card{S^{out}_{4}}$}\\
&= s\cdot \paren{\card{S^{in}_{3}}\cdot \card{S^{out}_{3}} + \card{S^{in}_{4}}\cdot \card{S^{out}_{4}}}.
\end{align*}
Therefore, the gap between the costs is at least 
\begin{align*}
\mathcal{C}_{2} - \mathcal{C}_{1} & > \paren{\card{S^{in}_{3}}\cdot \card{S^{out}_{3}} + \card{S^{in}_{4}}\cdot \card{S^{out}_{4}}}\cdot 2s - (\card{E(S_{1}, S_{2})}+\card{E(S_{1},S_{4})}+\card{E(S_{2}, S_{3})})\cdot 2s \\
&\geq \paren{\card{S^{in}_{3}}\cdot \card{S^{out}_{3}} + \card{S^{in}_{4}}\cdot \card{S^{out}_{4}}-\frac{3}{4}\cdot s}\cdot 2s \tag{$\card{E(S_{1}, S_{2})}+\card{E(S_{1},S_{4})}+\card{E(S_{2}, S_{3})}\leq \frac{3}{4}\cdot s$} \\
&\geq 2s^2 >0. \tag{since $\card{S^{in}_{3}}\cdot \card{S^{out}_{3}}\geq s-1$ and $\card{S^{in}_{4}}\cdot \card{S^{out}_{4}}\geq s-1$}
\end{align*}

That is, the quantity of $\mathcal{C}_{2} - \mathcal{C}_{1}$ is always \emph{positive}. Hence, a tree the pattern of $\TT$ is strictly better than following the pattern of $\TT'$, which means the optimal tree will not start the first split involving edges inside $S_{3}$ and $S_{4}$. Hence, the first cut of a optimal tree should be restricted to $E(S_{1}, S_{2})\cup E(S_{2}, S_{3}) \cup E(S_{1}, S_{4})\cup E(S_{3}, S_{4})$.

We then prove that, conditioning the first cut splits the edges among $E(S_{1}, S_{2})\cup E(S_{2}, S_{3}) \cup E(S_{1}, S_{4})\cup E(S_{3}, S_{4})$, the second cut is also restricted to the edges between cliques. Once again, let $\TT$ be the best possible tree that first splits the edges between cliques (and therefore edges among $E(S_{2}, S_{3})$), and let $\TT'$ be the best possible tree that first split edges inside $S_{3}$ and $S_{4}$. By the same induction argument, we show the second cut is not inside $S_{3}$ or $S_{4}$. Denote $\mathcal{C}_{1}$ as the cost of $\TT$ and $\mathcal{C}_{2}$ as the cost of $\TT'$, we have
\begin{align*}
\mathcal{C}_{2} - \mathcal{C}_{1} &\geq \paren{\card{S^{in}_{3}}\cdot \card{S^{out}_{3}} + \card{S^{in}_{4}}\cdot \card{S^{out}_{4}}}\cdot 2s -(2s-\card{S^{out}_{3}}-\card{S^{out}_{4}})\cdot(\card{S^{out}_{3}}\cdot \card{S^{out}_{4}}) \\
& \qquad - (2s-\card{S^{in}_{3}}-\card{S^{in}_{4}})\cdot(\card{S^{in}_{3}}\cdot \card{S^{in}_{4}}) - \card{E(S_{2}, S_{3})}\cdot 2s\\
&= \paren{\card{S^{in}_{3}}\cdot \card{S^{out}_{3}} + \card{S^{in}_{4}}\cdot \card{S^{out}_{4}}}\cdot s - \card{E(S_{2}, S_{3})}\cdot 2s\\
&\geq 2(s-1)\cdot s - \frac{3s}{8}\cdot 2s >0, \tag{$\card{S^{in}_{3}}\cdot \card{S^{out}_{3}}\geq s-1$, $\card{E(S_{2}, S_{3})}\leq \frac{3s}{8}$}
\end{align*}
where the equality is obtained by again using $(\card{S^{in}_{3}}-\card{S^{in}_{4}})\cdot(\card{S^{out}_{3}}\cdot \card{S^{out}_{4}}) + (\card{S^{out}_{3}}-\card{S^{out}_{4}})\cdot\card{S^{in}_{3}}\cdot (\card{S^{in}_{4}}) =  \paren{\card{S^{in}_{3}}\cdot \card{S^{out}_{3}} + \card{S^{in}_{4}}\cdot \card{S^{out}_{4}}}\cdot s$. As a result, the second cut among any optimal tree is restricted to the edges among $E(S_{1}, S_{2})\cup E(S_{2}, S_{3}) \cup E(S_{1}, S_{4})\cup E(S_{3}, S_{4})$.
\end{proof}

\begin{proof}[Proof of \Cref{prop:sparse-split-strong}]
By \Cref{lem:s1-s2-improp} and \Cref{lem:s3-s4-opt-cost}, we effectively rule out any optimal clustering tree $\TT$ that cuts clique edges in its first two splits. Therefore, by \Cref{lem:cut-sparsest-connect}, the optimal HC tree first separates $S_{1}$ from the rest of the graph and obtains $G'$, and then separate $S_{2}$ from the rest of the graph to obtain graph $G''$. Finally, by \Cref{prop:sparse-split-weak}, the optimal hierarchical clustering tree on $G''$ must first separate $S_{3}$ and $S_{4}$. As such, the behavior of the optimal HC tree is exactly as characterized in \Cref{prop:sparse-split-strong}.
\end{proof}

\subsection*{Acknowledgement} 
We thank Sanjeev Khanna for communicating their results in~\cite{AgarwalKLP22} to us and helpful conversations about their work and its connection to ours. 

\clearpage
%%\appendix

\bibliographystyle{alpha}
\bibliography{general}

\clearpage

\appendix
\part*{Appendix}

\section{Standard Technical Tools}
\label{app:info-theoretic-facts}

\subsection{Concentration Inequalities}
\label{sub-app:concentration-inequ}
We now present the standard concentration inequalities used in our proofs. We start from the following standard variant of Chernoff-Hoeffding bound. 

\begin{proposition}[Chernoff-Hoeffding bound]\label{prop:chernoff}
	Let $X_1,\ldots,X_n$ be $n$ independent random variables with support in $[0,1]$. Define $X := \sum_{i=1}^{n} X_i$. Then, for every $\delta \in (0,1]$, there is 
	\begin{align*}
		\Pr\paren{\card{X - \expect{X}} > \delta\cdot \expect{X}} \leq 2 \cdot \exp\paren{-\frac{\delta^{2} \expect{X}}{3}}. 
	\end{align*}
\end{proposition}

The standard Chernoff bound works on \emph{independent} random variables. Going beyond the independent case, it is also known that Chernoff bound applies to \emph{negatively correlated} random variables. Informally speaking, two random variables $X_{i}$ and $X_{2}$ are negatively correlated if conditioning on $X_{i}=1$, the probability for $X_{j}=1$ decreases. Formally, we define negatively correlated random variables as follows.

\begin{definition}[Negatively Correlated Random Variables]
\label{def:neg-cor-rvs}
Random variables $X_1,\ldots,X_n$ are said to be \emph{negatively correlated} if and only if
\begin{align*}
\expect{\prod_{i=1}^{n} X_{i}} \leq \prod_{i=1}^{n}\expect{X_{i}}.
\end{align*}
In particular, if $X_{i}$'s are independent, we have $\expect{\prod_{i=1}^{n} X_{i}} = \prod_{i=1}^{n}\expect{X_{i}}$.
\end{definition}

\begin{proposition}[Generalized Chernoff]\label{prop:chernoff-general}
Let $X_1,\ldots,X_n$ be $n$ negatively correlated random variables supported on $\{0,1\}$. Then, the concentration inequality in \Cref{prop:chernoff} still holds.
\end{proposition}

\subsection{Standard Tools for Lower Bound Proofs}
\label{sub-app:info-theoretic-facts}

We shall use the following standard properties of KL-divergence and TVD defined in~\Cref{subsec:info}. For the proof of this results, see the excellent textbook by Cover and Thomas~\cite{CoverT06}. 

The following facts state the chain rule property and convexity of KL-divergence. 

\begin{fact}[Chain rule of KL divergence]
\label{fact:kl-chain-rule}
For any random variables $X=(X_{1}, X_{2})$ and $Y=(Y_{1},Y_{2})$ be two random variables, 
\begin{align*}
\kl{X}{Y} = \kl{X_{1}}{Y_{1}} + \Exp_{x \sim X_1} \kl{X_{2}\mid X_{1}=x}{Y_{2}\mid Y_{1}=x}.
\end{align*}
\end{fact}

\begin{fact}[Convexity KL-divergence]
\label{fact:kl-convexity}
For any distributions $\mu_1,\mu_2$ and $\nu_1,\nu_2$ and any $\lambda \in (0,1)$, 
\begin{align*}
\kl{\lambda \cdot \mu_1 + (1-\lambda) \cdot \mu_2}{\lambda \cdot \nu_1 + (1-\lambda) \cdot \nu_2} \leq \lambda \cdot \kl{\mu_1}{\nu_1} + (1-\lambda) \cdot \kl{\mu_2}{\nu_2}. 
\end{align*}
\end{fact}

\begin{fact}[Conditioning cannot decrease KL-divergence]\label{fact:kl-conditioning}
	For any random variables $X,Y,Z$, 
	\[
		\kl{X}{Y} \leq \Exp_{z \sim Z} \kl{X \mid Z=z}{Y \mid Z=z}. 
	\]
\end{fact}

Pinsker's inequality relates KL-divergence to TVD. 

\begin{fact}[Pinsker's inequality]
\label{fact:pinsker}
For any random variables $X$ and $Y$ supported over the same $\Omega$, 
\begin{align*}
\tvd{X}{Y} \leq \sqrt{\frac{1}{2}\cdot \kl{X}{Y}}.
\end{align*}
\end{fact} 

The following fact characterizes the error of MLE for the source of a sample based on the TVD of the originating distributions. 

\begin{fact}
\label{fact:distinguish-tvd}
	Suppose $\mu$ and $\nu$ are two distributions over the same support $\Omega$; then, given one sample $s$ from either $\mu$ or $\nu$, the best probability we can decide whether $s$ came from $\mu$ or $\nu$ 
	is 
	\[
	\frac12 + \frac12\cdot\tvd{\mu}{\nu}.
	\]
\end{fact}
%%
%%Let $X$ and $Y$ be two random variables supported over the same $\Omega$ with  $\mu_{X}$ and $\mu_{Y}$. Given one sample draws from either $\mu_{X}$ or $\mu_{Y}$, the probability to distinguish them is 
%%\begin{align*}
%%\frac{1}{2} + \frac{1}{2}\cdot \tvd{X}{Y}.
%%\end{align*}
%%\end{fact} 

\paragraph{Fourier analysis on Boolean hypercube.} 
For any two functions $f,g: \set{0,1}^{n} \rightarrow \IR$, we define the \emph{inner product} between $f$ and $g$ as: 
\[
	\langle f, g\rangle = \Exp_{x \in \set{0,1}^{n}} \bracket{f(x) \cdot g(x)} = \sum_{x \in \set{0,1}^n} \frac{1}{2^n} \cdot f(x) \cdot g(x). 
\]
For a set $S \subseteq \set{0,1}$, we define the \emph{character} function $\mathcal{X}_S : \set{0,1}^n \rightarrow \set{-1,+1}$ as: 
\[
	\mathcal{X}_S(x) = (-1)^{(\sum_{i \in S} x_i)} = \begin{cases} 1 & \text{if $\oplus_{i \in S}~ x_i = 0$} \\ -1 &  \text{if $\oplus_{i \in S}~ x_i = 1$} \end{cases}. 
\]
The \emph{Fourier transform} of $f: \set{0,1}^{n} \rightarrow \IR$ is a function $\hf : 2^{[n]} \rightarrow \IR$ such that: 
\[
	\hf(S) =  \langle f,\mathcal{X}_S \rangle = \sum_{x \in \set{0,1}^n} \frac{1}{2^n} \cdot f(x) \cdot \mathcal{X}_S(x). 
\]
We refer to each $\hf(S)$ as a \emph{Fourier coefficient}. 

\medskip
We use \emph{KKL inequality} of~\cite{KahnKL88} for bounding sum of \emph{squared} of Fourier coefficients. 

\begin{proposition}[\!\!\cite{KahnKL88}]\label{prop:kkl}
	For every function $f \in \set{0,1}^n \rightarrow \set{-1,0,+1}$ and every $\gamma \in (0,1)$
	\[
		\sum_{S \subseteq [n]} \gamma^{\card{S}} \cdot \hf(S)^2 \leq \paren{\frac{\supp{f}}{2^n}}^{\frac{2}{1+\gamma}}. 
	\]
\end{proposition}

\section{Missing Proofs of \Cref{thm:beta-exist} and \Cref{prop:beta-compute}}
\label{app:proof-beta-tree-approx}

If we use the balanced minimum cut to establish a lower bound on the value of the optimal cost, we can obtain a clean proof of $O(\log n)$ approximation. We include this proof in \Cref{app:weaker-beta-tree}. Proving the stronger bound in \Cref{thm:beta-exist} requires some more involved techniques first developed by \cite{charikar2017approximate}.

\begin{definition}
\label{def:edge-footprint}
Let $G=(V,E)$ be a graph, $\TT$ be a HC-tree of $G$ and $(u,v)\in E$. The \emph{footprint of $(u,v)$ at size $t$}, denoted by $\ftE{t}{\TT}{(u,v)}$, is defined as follows.
\begin{align*}
\ftE{t}{\TT}{(u,v)} = 
\begin{cases}
0, \, \text{if there exists a cluster $C$ induced by $\TT$, s.t. $|C| \leq t$ and  $|\{u,v\} \cap C| = 1$}\\
w(u,v), \, \text{otherwise}.
\end{cases}
\end{align*}
\end{definition}

We first observe the relationship between the edge footprint and the cost of a hierarchical clustering. Consider  an edge $(u,v)$ and let $C = \leaves(\TT[u \wedge v])$. Recall that $(u, v)$ contributes a cost of $|C| \cdot w((u,v))$.
Thus, the footprint of the edge is equal to its weight for any $t < |C|$.
As such, we have the following:

\begin{lemma}
\label{clm:cost-to-footprint}
Using the assumptions of \Cref{def:edge-footprint}, and assuming that $G$ has $n$ vertices, 
\begin{align*}
\frac{1}{3}\cdot \sum_{t=0}^{n}\sum_{(u,v)\in E}\ftE{(2/3)\cdot t}{\TT}{(u,v)} \leq \sum_{t=0}^{n}\sum_{(u,v)\in E}\ftE{t}{\TT}{(u,v)} =  \cost(\TT).
\end{align*}
\end{lemma}
\begin{proof}
We first prove the equality.
By the definition of $\ftE{t}{\TT}{(u,v)}$, the weight of an edge $(u,v)$ is counted $r$ times where $r = |\leaves(\TT[u \wedge v])|$ (note that the first sum starts with $t=0$).

To prove the inequality, let us assume WLOG that $n$ is a multiplier of $3$. Note that the number of terms between $A:=\sum_{t=0}^{n}\sum_{(u,v)\in E}\ftE{t}{\TT}{(u,v)}$ and $B:=\sum_{t=0}^{n}\sum_{(u,v)\in E}\ftE{(2/3)\cdot t}{\TT}{(u,v)}$ are the same, and we charge $B$ into $3$ copies of $A$. Note that when $(1/3)\cdot t\in\{0,2,4,\cdots, \frac{2}{3}\cdot n\}$, we can charge this part of $B$ to $C:=\sum_{t=0}^{\frac{2}{3}\cdot n}\sum_{(u,v)\in E}\ftE{t}{\TT}{(u,v)}$, which in tern is at most $A$. For the second case, consider $(2/3)\cdot t\in\{\frac{2}{3}, \frac{8}{3}, \cdots, \frac{2}{3}\cdot n-2\}$, we introduce another copy of $A$, and since $\ftE{t}{\TT}{(u,v)}<\ftE{t-2/3}{\TT}{(u,v)}$, this part of $B$ is also upper-bounded by $A$. Finally, we consider the case $(2/3)\cdot t\in\{\frac{4}{3}, \frac{10}{3}, \cdots, \frac{2}{3}\cdot n-1\}$. With the same reasoning as above, the quantity of this part of $B$ is again at most $A$. Therefore, we have $3 A\geq B$, as claimed.
\end{proof}

Note that a result similar to \Cref{clm:cost-to-footprint} was first obtained by \cite{charikar2017approximate}. However, the subtle difference makes their statement not directly applicable for our purpose.

\begin{lemma}
\label{clm:split-cost-charge-balanced-min}
Let $\TT^*$ be a tree of cost $\OPT(G)$ and $\TT$ be a tree obtained by  recursively applying $\frac13$-balanced min cut.
Let $w$ be an internal node of $\TT$, $S = \leaves(\TT[w])$ and $(S_1, S_2) = \cut(\TT[w])$. Denote $r:=|S|$, $s:=|S_{1}|$. Then,
\begin{align*}
r\cdot w(S_{1}, S_{2}) \leq 3\cdot s \cdot \sum_{(u,v) \in S(E)} \ftE{(2/3)\cdot r}{\TT^{*}}{(u,v)}.
\end{align*}
\end{lemma}
\begin{proof}
Let $S^{\star}_{1}$, $S^{\star}_{2}$, $\cdots$, $S^{\star}_{k}$ be the maximal (w.r.t. inclusion) clusters induced by $\TT^*$, which have size at most $\frac23 \cdot r$ and a nonempty intersection with $S$.
Observe that these clusters are all disjoint.

We claim there exist two sets of indices $L = \{l_{1}, l_{2}, \cdots\}$ and  $R = \{r_{1}, r_{2}, \cdots\}$, such that 
\begin{itemize}
\item $L \cup R = [k]$;
\item If we denote $A:=(\cup_{i \in L} S^{\star}_{i})\cap S$ and $B:=(\cup_{i \in R} S^{\star}_{i})\cap S$, we have $\max\{A, B\} \leq \frac{2}{3}\cdot r$. 
\end{itemize}

To prove this claim we use the fact that the intersection of each $S^{\star}_i$ with $S$ is at most $\frac{2}{3}\cdot r$, and the following observation.

\begin{observation}
Let $x_1, \ldots, x_k$ be a sequence of positive real numbers, such that $\sum_{i=1}^k x_i= 1$, and $x_i \leq \frac23$.
There exists a sequence $1 \leq j_1 < \ldots < j_l \leq k$, such that $\frac13 \leq \sum_{i=1}^l x_{j_i} \leq \frac23$.
\end{observation}

This implies that the cut $(A, B)$ is $\frac13$-balanced, so in particular its weight is at least the weight of the minimum $\frac13$-balanced cut.
We have
\begin{align}
w(S_{1}, S_{2}) &\leq w\paren{\cup_{i \in L} S^{\star}_{i})\cap S, (\cup_{i \in R} S^{\star}_{i})\cap S} \tag{by balanced minimum cut}\\
&\leq \sum_{(u,v) \in S(E)} \ftE{(2/3)\cdot r}{\TT^{*}}{(u,v)}.\label{equ:balanced-min-charge}
\end{align}
The second inequality holds, since each edge in the cut $\paren{\cup_{i \in L} S^{\star}_{i})\cap S, (\cup_{i \in R} S^{\star}_{i})\cap S}$ has exactly one endpoint in some $S^{\star}_i$ (whose size is at most $(2/3)\cdot r$), and thus a nonzero footprint at level $(2/3)\cdot r$

Since the cut $(S_1, S_2)$ is $\frac13$-balanced and, by the definition of $\cut$, $|S_1| \leq |S_2|$, we have $\frac{r}{3}\leq s$.
Hence, we get
\begin{align*}
r\cdot w(S_{1}, S_{2}) \leq 3\cdot s \cdot \sum_{(u,v) \in S(E)} \ftE{(2/3)\cdot r}{\TT^{*}}{(u,v)},
\end{align*}
as claimed.
\end{proof}

\begin{corollary}
\label{cor:split-cost-charge-balanced-min}
Using the assumptions of~\Cref{clm:split-cost-charge-balanced-min}:
$$
r\cdot w(S_{1}, S_{2}) \leq 3\cdot \sum_{t=|S_2|+1}^{|S_1|+|S_2|} \sum_{(u,v) \in S(E)} \ftE{(2/3)\cdot t}{\TT^{*}}{(u,v)}
$$
\end{corollary}

\begin{proof}
This follows directly from \Cref{clm:split-cost-charge-balanced-min} and since for any $i$, $\ftE{i}{\TT^{*}}{(u,v)}\leq \ftE{i-1}{\TT^{*}}{(u,v)}$.
\end{proof}

\begin{lemma}
\label{clm:cluster-disjoint}
Let $\TT$ be a tree obtained by recursively applying $\frac13$-balanced minimum cut.
Consider the sum
$$
\sum_{\substack{\\ \\ \text{$(S_1,S_2):= \cut(\TT[w])$ for} \\ \text{internal nodes $w$ of $\TT$}}}  \sum_{t=|S_2|+1}^{|S_1|+|S_2|} \sum_{(u,v) \in (S_1 \cup S_2)(E)} F(t, u, v),
$$
Then, for any $0 \leq t' \leq n$, and any  $u',v'$, the term $F(t', u', v')$ appears at most once in the sum.
\end{lemma}
\begin{proof}
Clearly, any possible overlap in the terms can only come from two nodes $w \neq w'$, such that $S = \leaves(\TT[w])$, $S' = \leaves(\TT[w'])$ and $S \cap S' \neq \emptyset$.
WLOG we can assume that $w$ is an ancestor of $w'$.

Denote $(S_1, S_2) := \cut(\TT[w])$, where $|S_1| \leq |S_2|$.
Since $S'$ is either a subset of $S_1$ or $S_2$, we have $|S'| \leq |S_2|$.
But then, the largest index $t$ we can obtain when we consider all summands corresponding to $w'$ is $|S_2|$.
However, the smallest index $t$ corresponding to $w$ is $|S_2|+1$.
\end{proof}

\begin{proof}[Proof of \Cref{thm:beta-exist}]
With \Cref{clm:cost-to-footprint}, \Cref{cor:split-cost-charge-balanced-min} and \Cref{clm:cluster-disjoint} in our hands, now we can establish the approximation ratio of $\TT$ that is obtained by recursive $\frac13$-balanced min-cut. 

\begin{align*}
\hspace{-1cm}\sum_{\substack{\\ \\ \text{$(S_1,S_2):= \cut(\TT[u])$ for} \\ \text{internal nodes $u$ of $\TT$}}} \hspace{-1cm} \!\!\ w(S_{1}, S_{2})\cdot r &\leq 3\cdot \hspace{-1cm}\sum_{\substack{\\ \\ \text{$(S_1,S_2):= \cut(\TT[u])$ for} \\ \text{internal nodes $u$ of $\TT$}}} \sum_{t=|S_2|+1}^{|S_1|+|S_2|} \sum_{(u,v) \in S(E)} \ftE{(2/3)\cdot t}{\TT^{*}}{(u,v)} \tag{By \Cref{cor:split-cost-charge-balanced-min} }\\
&= 3\cdot \sum_{t=0}^{n} \sum_{(u,v) \in E} \ftE{(2/3)\cdot t}{\TT^{*}}{(u,v)} \tag{By \Cref{clm:cluster-disjoint} and the disjointness}\\
&\leq 9\cdot \cost(\TT^{*}). \tag{By \Cref{clm:cost-to-footprint}} \qed
\end{align*}

\end{proof}

\begin{proof}[Proof of \Cref{prop:beta-compute}]
The polynomial-time algorithm is to \emph{recursively apply the $O(\sqrt{\log{n}})$ approximation algorithm for balanced minimum cuts on the subgraphs of $G$}. Suppose $S_{1}$ and $S_{2}$ are obtained by applying the $O(\sqrt{\log{n}})$-approximation of the balanced minimum cut, by changing the line in \Cref{equ:balanced-min-charge}, we have
\begin{align}
w(S_{1}, S_{2}) &\leq O(\sqrt{\log{n}})\cdot w\paren{\cup_{i \in L} S^{\star}_{i})\cap S, (\cup_{i \in R} S^{\star}_{i})\cap S} \tag{by $O(\sqrt{\log{n}})$-approximation of balanced minimum cut}\\
&\leq O(\sqrt{\log{n}})\cdot \sum_{(u,v) \in S(E)} \ftE{(2/3)\cdot r}{\TT^{*}}{(u,v)}. \label{equ:balanced-min-approx-charge}
\end{align}
As such, for a tree $\TT$ obtained by recursive $O(\sqrt{\log{n}})$ approximation of the balanced minimum cut, 
\begin{align*}
\hspace{-1cm}\sum_{\substack{\\ \\ \text{$(S_1,S_2):= \cut(\TT[u])$ for} \\ \text{internal nodes $u$ of $\TT$}}} \hspace{-1cm} \!\!\ w(S_{1}, S_{2})\cdot r &\leq O(\sqrt{\log{n}})\cdot \hspace{-1cm}\sum_{\substack{\\ \\ \text{$(S_1,S_2):= \cut(\TT[u])$ for} \\ \text{internal nodes $u$ of $\TT$}}} \sum_{t=|S_2|+1}^{|S_1|+|S_2|} \sum_{(u,v) \in S(E)} \ftE{(2/3)\cdot t}{\TT^{*}}{(u,v)} \tag{By \Cref{equ:balanced-min-approx-charge} and the fact that $\ftE{i}{\TT^{*}}{(u,v)}\leq \ftE{i-1}{\TT^{*}}{(u,v)}$}\\
&= O(\sqrt{\log{n}})\cdot \sum_{t=0}^{n} \sum_{(u,v) \in E} \ftE{(2/3)\cdot t}{\TT^{*}}{(u,v)} \tag{By \Cref{clm:cluster-disjoint} and the disjointness}\\
&\leq O(\sqrt{\log{n}})\cdot \cost(\TT^{*}). \tag{By \Cref{clm:cost-to-footprint}}
\end{align*}
We now analyze the time complexity. Note that each approximate balanced minimum cut takes polynomial time. Furthermore, there are at most polynomially-many nodes in a HC-tree since there are at most $n$ leaves. Therefore, the algorithm runs in polynomial time.
\end{proof}

\section{A Weaker Version of \Cref{thm:beta-exist}}
\label{app:weaker-beta-tree}

In this section, we present a weaker version of \Cref{thm:beta-exist} with an $O(\log n)$ approximation factor. The value of the weaker version is that $(i)$ the proof is much simpler; and $(ii)$ it gives some results on binary tree analysis in addition to the edge cost charging as we used in \Cref{thm:beta-exist}, which may be of independent interests.

The formal statement of the weaker result is as follows.

\begin{proposition}
\label{prop:beta-exist-weak}
	For any graph $G=(V,E,w)$, there exists a $(1/3)$-balanced tree $\TT_{balanced}$ such that 
	\[
	\cost_G(\TT_{balanced}) \leq O(\log n)\cdot \OPT(G).
	\]
\end{proposition}

We first use the balanced minimum cut problem to lower bound the cost of optimum solution. 

\begin{lemma}\label{lem:balanced-cut}
	For any graph $G$ and any tree HC-tree $\TT$, 
	\begin{align*}
		\cost_G(\TT) \geq n/3 \cdot \min_{S,\bar{S} \subseteq V} \,\, w(S,\bar{S}) \quad \text{s.t.} \quad n/3 \leq \card{S},\card{\bar{S}} \leq 2n/3. 
	\end{align*}
\end{lemma}
\begin{proof}
	Since $\TT$ is a binary tree with $n$ leaf-nodes, there should exists a node $u$ in $\TT$ with 
	\[
	n/3 \leq \card{\leaves(\TT[u])} \leq 2n/3; 
	\] 
	(the proof is a standard vertex separator argument for binary trees). Let us fix that node $u$ and consider the node $w$ as the parent of $u$ in $\TT$; let $(A,B) = \cut(\TT[w])$ with $A$ being the side of cut assigned to $u$, i.e., $A = \leaves(\TT[u])$.
	
	Now consider the cut $(A,\bar{A})$ which is a global cut of $G$. For any edge $(x,y)$ of this cut, $x \vee y$ is either $w$ or some node on the path from the root to $w$. This, combined with~\Cref{eq:hc-cost}, implies that
	\[
		\cost_G(\TT) \geq w(A,\bar{A}) \cdot \card{A \cup B} \geq w(A,\bar{A}) \cdot n/3,
	\]
	as $\card{A} \geq n/3$. Moreover, the cut $(A,\bar{A})$ satisfies the property that $n/3 \leq \card{A},\card{\bar{A}} \leq 2n/3$. As such, the minimum in RHS of the lemma statement is at most $w(A,\bar{A})$, 
	which implies the lemma. 
\end{proof}

\begin{proof}[Proof of~\Cref{prop:beta-exist-weak}]
	Consider the following process for constructing $\TT_{balanced}$: 
	\begin{enumerate}[label=$(\roman*)$]
		\item Pick a cut $(S,\bar{S})$ of $G$ minimizing $w(S,\bar{S})$ subject to $n/3 \leq \card{S},\card{\bar{S}} \leq 2n/3$. Let the root $r$ of $\TT_{balanced}$ be such that $\cut(\TT_{balanced}[r]) = (S,\bar{S})$ (this uniquely identifies the root). 
		\item Let $G_{S}$ and $G_{\bar{S}}$ be the induced subgraphs of $G$ on $S$ and $\bar{S}$, respectively. Recursively run the same process for $G_S$ and $G_{\bar{S}}$ and let the root of their corresponding trees be the left-child and 
		right-child node of $r$, respectively (the base case is when the sets have size $1$ in which case they form leaf-nodes of $\TT_{balanced}$).  
	\end{enumerate}

	It is clear that $\TT_{balanced}$ is valid HC-tree for $G$ and that it is $(1/3)$-balanced by~\Cref{def:beta-tree}, simply by the ``splitting rule'' of part $(i)$. Moreover, $\TT_{balanced}$ being $(1/3)$-balanced implies that
	the depth of this tree is $O(\log{n})$, which we will use in proving the upper bound on the cost of the tree. 
	
	Consider all nodes $u_1,\ldots,u_t$ at some depth $d$ of the tree (for some $d=O(\log{n})$). For $i \in [t]$, let $G_i$ be the induced subgraph of $G$ on vertices in $\leaves(\TT[u_i])$ and $(S_i,\bar{S_i})$ be the cut 
	chosen for this node in the process above. 
	By~\Cref{obs:partition}, 
	\[
		\sum_{i=1}^{t} \OPT(G_i) \leq \OPT(G). 
	\] 
	At the same time, by~\Cref{lem:balanced-cut}, for every $i \in [t]$, 
	\[
		\OPT(G_i) \geq 1/3 \cdot \card{S_i \cup \bar{S_i}} \cdot w(S_i,\bar{S_i}), 
	\]
	which, together with the previous bound, implies that 
	\[
		\sum_{i=1}^{t} \card{S_i \cup \bar{S_i}} \cdot w(S_i,\bar{S_i}) \leq 3 \cdot \OPT(G). 
	\]
	Finally, by combining this with~\Cref{eq:hc-cost-cut}, we have that, 
	\begin{align*}
		\cost_G(\TT_{balanced}) = \sum_{d=1}^{O(\log{n})} \sum_{i=1}^{t_d} \card{S_i \cup \bar{S_i}} \cdot w(S_i,\bar{S_i}) \leq \sum_{d=1}^{O(\log{n})} 3 \cdot \OPT(G) = O(\log{n}) \cdot \OPT(G),
	\end{align*}
	as the  depth of the tree is $O(\log{n})$. This concludes the proof. 
\end{proof}

\end{document}